%% file: arxiv_Nov2026.tex
\newif \ifproofs
\proofstrue

\documentclass[11pt]{article}
\usepackage{amssymb}
\usepackage{mathrsfs,amsmath}
\usepackage{graphicx}
\usepackage{amsmath}
\usepackage{amsthm}
\usepackage{bbm}
\usepackage{dsfont}
\usepackage{listing}
\usepackage{hyperref}
\usepackage{color}
\usepackage[ruled, vlined,linesnumbered, noend]{algorithm2e}
\usepackage{comment}
\usepackage{enumitem}
\usepackage{booktabs}
\usepackage{tabulary}
\usepackage{soul}

\usepackage{float}
\usepackage{algpseudocode}

\usepackage{subcaption}

\usepackage{tikz}
\usepackage{pgfplots}
\pgfplotsset{compat=newest}

\usepackage{textcomp}
\usepackage{siunitx}

\newtheorem{theorem}{Theorem}
\newtheorem{corollary}{Corollary}
\newtheorem{lemma}{Lemma}

\newtheorem{proposition}{Proposition}
\theoremstyle{definition}
\newtheorem{definition}{Definition}
\newtheorem{remark}{Remark}

\usepackage{bm}

\usepgfplotslibrary{groupplots,fillbetween}

\usepackage{textcomp}
\usepackage{siunitx}

\colorlet{blue}{black}

\DeclareMathOperator*{\argmax}{arg\,max}

\usepackage{natbib}
 \bibpunct[, ]{(}{)}{,}{a}{}{,}%
\newif\ifarxiv   
\arxivfalse 

\ifarxiv
\else

\fi

\input{tex/macros2}

\usepackage{fancyhdr}
\usepackage{calc}
\fancyheadoffset[RE]{\marginparsep+\marginparwidth}

\usepackage[top=1.5in, bottom=1in, left=1in, right=1in]{geometry}

\graphicspath{ {fig/} }

\title{Informal and Privatized Transit: Incentives, Efficiency and Coordination}
\author{Devansh Jalota\\ Georgia Tech\\{\tt devansh.jalota@isye.gatech.edu}
        \and 
        Matthew Tsao\\ Lyft\\ {\tt mattxtsao@gmail.com}
}

\begin{document}
\maketitle

\begin{abstract}
Informal and privatized transit services, such as minibuses and shared auto-rickshaws, are integral to daily travel in large urban metropolises, providing affordable commutes where formal public transport is inadequate and other options are unaffordable. A defining feature of these systems is their decentralized, market-based organization, with drivers providing service in response to rider demand and earning opportunities. While this structure helps fill critical mobility gaps, it also generates inefficient service patterns when profit-driven driver route choices do not align with system-wide mobility goals. 

We develop an analytically tractable game-theoretic framework to study incentives underlying informal and privatized transit with a fixed menu of routes, quantify efficiency losses from decentralized driver route choice, and design incentive mechanisms to mitigate these inefficiencies. Here, profit-maximizing informal operators (\emph{drivers}) decide where to provide service and cost-minimizing commuters (\emph{riders}) decide whether to use these services. Within this framework, we establish a \emph{tight} price of anarchy bound for cumulative driver profit and an instance-wise rider throughput bound whose dependence on per-rider profit heterogeneity is tight, showing that decentralized, profit-maximizing driver behavior can lead to substantial losses in cumulative driver profit and rider demand served. Moreover, we show that these losses can be mitigated through targeted interventions: \emph{budget-balanced cross-subsidization}, which uses route-specific tolls/subsidies to shape driver payoffs, and \emph{fare optimization}, which changes rider demand and driver margins through centrally regulated route-level fares. Finally, numerical experiments calibrated using data from an informal transit system in Nalasopara, India, reinforce these findings.
\end{abstract}

\section{Introduction}
\vspace{-3pt}
Governments in many cities struggle to provide adequate public transport to meet rising travel demand due to chronic fiscal constraints. Although on-demand services, such as Uber, have partially filled this gap, they remain unaffordable for most low- and middle-income users. Thus, commuters face a stark mobility gap: on one side, low-frequency and geographically limited public transit; on the other, high-cost individualized services that are unaffordable. To fill this gap, informal transport (or intermediate public transport) services, such as minibuses, shared auto-rickshaws, matatus, and jeepneys, have become integral to daily travel, providing affordable commutes where a formal public transport system is inadequate and other options are unaffordable~\citep{cervero2000informal}. Operated largely by private providers, these services form the backbone of urban mobility in many developing nation cities, often accounting for 50–100\% of urban transit trips~\citep{Conwell2025PrivatizedTransit}. 

Informal or otherwise privatized fixed-route services are not limited to developing nations, with similar services filling gaps even in mature public transit systems in the developed world. In New York, for example, a range of private operators have emerged, from long-standing jitney services in outer-borough corridors~\citep{GOLDWYN2020102309} to Uber’s recent airport shuttle offerings~\citep{uber-shuttle}, which pool riders along fixed routes to offer a lower cost shared alternative to individualized ride-hailing. Together, these examples highlight a broader global pattern: privatized and, often, informal shared transit systems are crucial to meeting the mobility needs of hundreds of millions of daily travelers worldwide that formal public transit leaves unmet.

A defining feature of these services that enables them to fill mobility gaps, often at low or no fiscal cost to governments, is their decentralized, market-based organization, wherein private drivers or operators serve particular routes in response to rider demand and earning opportunities. While this decentralized structure allows these systems to extend service where formal transit is inadequate, it can produce inefficient service patterns, with operators concentrating on lucrative corridors while peripheral neighborhoods with lower returns remain underserved. Thus, although these systems often resemble public transit operationally, with shared vehicles running on fixed routes, profit-driven decentralized incentives shape where service is provided, creating inefficiencies that can compromise system efficiency and equity~\citep{KLOPP2019657}. These inefficiencies raise a central policy question: \emph{how large are the losses from decentralized driver route choices, and how should the incentives of informal and privatized transit operators be shaped to mitigate them?} Quantifying these losses is important for governments, driver associations, and platform-mediated shared transit providers seeking to understand the value of coordination, especially as informal and privatized transit are increasingly being integrated, professionalized, and corporatized into formal businesses~\citep{ssatp2021capetown}. At the same time, the policy problem is not only whether coordination is valuable, but how much of that value can be recovered through incentive mechanisms. 


\textbf{Our Contributions:} 
This work addresses the above question by developing a framework for analyzing the incentives underlying informal and privatized transit systems and quantifying the performance loss from decentralized driver behavior. We further study two natural monetary mechanisms to shape incentives in these systems, characterize their optimal design, and identify the limits of these mechanisms, i.e., when losses from selfish driver behavior remain unavoidable.

In Section~\ref{sec:model}, we introduce a novel game-theoretic model of a privatized informal transit system with a fixed set of routes. On these routes, cost-minimizing \emph{riders} choose between using informal transit to complete their trip or an outside option, such as walking or staying at home, while profit-maximizing informal or privatized operators (\emph{drivers}) decide which routes to serve. Our framework captures two key features shaping incentives and equilibrium outcomes in fixed-route informal transit: (i) driver competition for profitable routes, which leads to over-provision on some routes and under-provision in others, and (ii) rider queuing and waiting delays, a defining feature of informal transit during morning and evening peaks that shapes riders’ willingness to use these services. 


To capture driver competition and its effect on their route choices, we use a non-atomic congestion game in which congestion reflects the impact of service over or under-provision on rider and driver payoffs, rather than physical road congestion. To capture rider queuing and waiting delays, we adapt Vickrey's seminal bottleneck model of peak-period congestion~\citep{vickrey1969congestion} to our informal transit setting. A key departure from the standard bottleneck model is that a route's service capacity is not fixed exogenously, but emerges endogenously from drivers’ profit-maximizing route choices. In this sense, the rider-side problem of our framework can be viewed as a \emph{bottleneck model with endogenous service capacities}. By jointly capturing both driver competition and rider delays in a unified framework, our model is sufficiently rich to capture the key operational features of informal and privatized transit systems, yet tractable enough to admit closed-form solutions and isolate the core forces at play. We thus view the \emph{formulation} of our model as a central contribution. 

For this model, in Section~\ref{sec:poa-bounds}, we derive price of anarchy (PoA) guarantees under both cumulative driver profit and rider throughput objectives, where throughput is measured by the total number of riders served by informal transit, a central policy objective in transit planning. Our PoA \color{blue} notions \color{black} quantify the worst-case ratio (over problem instances) between the optimal and equilibrium outcomes, thereby capturing performance losses from drivers' selfish route choices. Specifically, \color{blue} we establish a \emph{tight} PoA bound of $2$ for cumulative driver profit and an instance-wise rider throughput bound of $1+\frac{p_{\max}}{p_{\min}}$, where $p_{\min}$ and $p_{\max}$ denote the minimum and maximum per-rider profits across routes in that instance; moreover, the dependence of this bound on per-rider profit heterogeneity is tight. \color{black}
Thus, our results highlight that selfish driver decisions can result in \emph{substantial inefficiencies}.

To mitigate these inefficiencies, we study two mechanisms to steer informal transit operators to improved system outcomes: \emph{cross-subsidization} (Section~\ref{sec:cross-subsidies}) and \emph{fare optimization} (Section~\ref{sec:price-optimization}).

First, we study cross-subsidization, in which a central planner uses route-specific tolls or subsidies to influence driver route choices. We show that for any informal transit instance, there exists a budget-balanced (i.e., requiring no net spending) cross-subsidy scheme that aligns driver incentives \color{blue} with any desired driver allocation as an equilibrium, and can be computed in polynomial time for cumulative driver profit or rider throughput maximizing allocations. We further establish conditions under which the induced equilibrium is unique and show that, when per-rider profits are strictly positive, a unique equilibrium can achieve either objective arbitrarily close to its centralized optimum. These results demonstrate that cross-subsidies can essentially eliminate inefficiencies from selfish driver behavior without incurring net fiscal expenditure at the resulting equilibria.

\color{black}

\color{blue}
While cross-subsidization provides a benchmark for the efficiency gains achievable through full alignment of driver incentives, it can be difficult to implement in informal transit systems, as monitoring driver behavior and enforcing compliance can be challenging. \color{black} We thus study fare optimization, where a central planner directly sets route-level fares, more closely reflecting how driver associations and governments regulate informal and privatized transit systems in practice. 
In this setting, we define a fare-optimized PoA notion that compares two benchmarks: a centralized benchmark, where the planner jointly optimizes fares and driver allocations to maximize either objective, and an equilibrium benchmark, where the planner chooses fares but drivers choose routes selfishly in response. 
Unlike cross-subsidies, fare optimization has \emph{asymmetric power} across objectives: it can recover the centralized rider throughput benchmark arbitrarily closely, but cannot improve the factor-two worst-case bound for cumulative profit. Moreover, despite the non-convexity of the joint fare-setting and driver allocation problems, we develop polynomial-time algorithms to compute the fare-optimized centralized and equilibrium benchmarks under both objectives.

We complement our theory with numerical experiments in Section~\ref{sec:experiments} \color{blue} calibrated to \color{black} an informal transit system in Nalasopara, India. 
\color{blue}
Our calibrated numerical experiments under the existing fares in Nalasopara show that, although the inefficiencies from selfish driver behavior do not reach the worst-case levels characterized by our PoA analysis, the resulting cumulative driver profit and rider throughput remain meaningfully below their respective centralized optima.
\color{black}
We further examine a setting in which a central planner can adjust route-level fares, showing that fare optimization is a practically relevant policy instrument for improving outcomes in informal transit systems, but its efficacy depends on how fare adjustments are paired with driver earning objectives or constraints.

The appendix contains additional related literature, proofs omitted from the main text, additional theoretical results, and numerical implementation details. While framed in the context of informal transit, our model may be applicable to other settings, such as ride-sharing, in which the forces we study, including driver competition for service and rider queuing delays, play a central role.

\vspace{-9pt}
\section{Related Literature} \label{sec:literature}
\vspace{-4pt}

A large literature documents the key role of informal transit in urban mobility and its defining aspects, including institutional arrangements, regulatory challenges, and service patterns. Despite taking many forms across contexts, ``informal transport is, stripped to the basics, about profit-seeking operators serving consumer demands''~\citep{cervero2000informal}. Motivated by this core feature, we study incentive interactions in fixed-route informal transit, common in cities worldwide. However, unlike the largely qualitative and case study based informal transit literature rooted in urban planning and political science~\citep{KLOPP2021191,Kerzhner_2022}, we develop a formal framework that complements these diagnostic studies with prescriptive tools for incentive design in such systems.

By studying the incentives and efficiency of informal transit, our work connects to the literature on the efficiency of privatized transportation. Prior work studies competition among private service providers~\citep{rosaia-2020} and emphasizes search frictions~\citep{brancaccio2023search,buchholz2015spatial}, where spatial mismatch leads to matching costs or delays. In the context of informal transit,~\citet{mbonu2024market} document how territorial segmentation among minibus associations affects service provision, while~\citet{mittal2024efficient} show that informal transit routes often deviate less from shortest paths than formal public transit. Closest to our work,~\citet{Conwell2025PrivatizedTransit} develops a model of privatized shared transit with profit-maximizing drivers and rider queuing delays, and studies the efficiency gains from reorganizing privately operated minibuses. However, unlike the discrete stochastic queuing model of~\citet{Conwell2025PrivatizedTransit}, we adopt a deterministic continuum queuing framework and leverage ideas from congestion games, enabling a closed-form analysis.

Our work also extends the PoA literature on efficiency losses from selfish behavior~\citep{roughgarden2005selfish}, typically focused on a single class of strategic agents (e.g., commuters in routing games), to informal transit systems, where service provision and demand are jointly shaped by the strategic actions of both drivers and riders. \color{blue}
Relatedly, prior work on resource allocation games establishes general factor-two PoA guarantees for certain system objectives under finite and continuum populations~\citep{marden2013distributed,li2026effect}. While our factor-two profit PoA bound resembles these guarantees, establishing it in our endogenous informal transit setting requires a closed-form characterization of rider demand and its resulting monotonicity properties. 
\color{black}
Related work on informal transit considers two-sided incentives but relies primarily on numerical comparisons rather than PoA bounds~\citep{CHAVIS2013277,SANGVERAPHUNSIRI20221}, while others focus on equilibrium computation without efficiency comparisons~\citep{fernandez-1992}. In contrast, we derive explicit PoA bounds for \color{blue} both cumulative driver profit and rider throughput, \color{black} which, to our knowledge, are the first such results for informal transit systems.

A key component of our framework is Vickrey’s bottleneck model~\citep{vickrey1969congestion}, which we adapt to model rider behavior, where we treat waiting time for informal transit akin to queuing delay at a bottleneck~\citep{KRAUS2002170}. Moreover, following prior extensions of the bottleneck model, we incorporate outside options~\citep{d0907f84-e14a-3d98-ad20-759f41491d6e,GONZALES20121519}, allowing riders to forgo transit when queuing delays become large. \color{blue}
Despite these similarities, unlike our work, existing bottleneck-based models largely focus on rider behavior under fixed service capacities. A closer analogue is \citet{BASSO201947}, who endogenize transit frequency in a dynamic congestion model with boarding queues; however, frequency in their setting is chosen by a central planner. In contrast, we embed bottleneck-style rider behavior within an informal transit system in which route-level service capacity emerges endogenously from decentralized driver route choices.
\color{black}


By studying the allocation of drivers to serve riders in mobility systems, our work relates to the ride-sharing and shared rides literature. Queuing models have been used to study pricing in ride-hailing~\citep{banerjee-johari-2015,banerjee-johari-2016}, typically in stochastic and stationary settings, unlike Vickrey's deterministic and non-stationary bottleneck model. Shared rides services have been studied from the lens of online stochastic optimization~\citep{aouad-or-2022,ashlagi-2019} and dynamic programming~\citep{yan-pricing-2024,jacquillat-2025}. These works treat rider arrival as exogenous, and do not model the option for riders to strategically time their entry into the market, which we incorporate in our model. 
These works study exclusively matching pairs of riders together, with the exception of~\cite{alonso-mora-2017}, which considers higher capacity shared rides but abstracts from strategic rider and driver behavior and pricing. Overall, these works study discrete models while we use a fluid model, though fluid approximations are often useful even in discrete models for designing algorithms with guarantees~\citep{aouad-or-2022,feng-2023-two}. 

Our work also relates to incentive design in decentralized mobility systems, public transit design, and the economic impacts of transit infrastructure, which we review in Appendix~\ref{apdx:related-work}.

\vspace{-8pt}
\section{Model} \label{sec:model}
\vspace{-3pt}

This section introduces a stylized model of informal and privatized transit with a fixed menu of routes. On these routes, \emph{riders} seek to make trips, and service is provided by profit-seeking informal or privatized operators, whom we refer to as \emph{minibus drivers} (or simply \emph{drivers}). Motivated by the practical operations of these systems, our framework captures two key features shaping incentives and equilibrium outcomes: (i) driver route choice and competition for profitable routes, which endogenously determine service capacities across routes, and (ii) rider queuing and waiting delays, which influence their willingness to use informal transit. We first introduce notation and features of the informal transit environment we study (Section~\ref{subsec:preliminaries}) and specify rider demand as a function of driver supply (Section~\ref{subsec:rider-payoffs}). Then, we define driver payoffs and our equilibrium notion in Section~\ref{subsec:eq-def}. Next, we introduce the metrics of cumulative driver profit and rider throughput and present the price of anarchy notion in Section~\ref{subsec:poa}. Finally, we discuss modeling assumptions in Section~\ref{subsec:model-assumptions}.

\vspace{-4pt}
\subsection{Preliminaries: Notation and Features of Informal Transit Environment} \label{subsec:preliminaries}
\vspace{-2pt}
We consider an informal transit system during the morning or evening peak period, defined by the interval $[t_1, t_2]$, with a fixed menu of $n$ routes. In this system, a continuum of $D$ \emph{minibus drivers} choose which routes to serve, where each minibus has a fixed capacity $F$, denoting the mass of riders that can be served in a single trip. Consistent with observed practice during peak periods, a driver serving a route operates at full capacity, picking up riders at the origin, transporting them to their destination, and then returning empty to the origin to continue service (see left of Figure~\ref{fig:warmup-model}).

Each route $i \in [n]$ is characterized by a vector $\Theta_i = (\Bar{p}_i, l_i, c_i)$, where $\Bar{p}_i$ is the per-trip fare paid by minibus riders, $l_i$ is the fixed one-way minibus travel time (so the round-trip travel time is $2 l_i$), and $c_i$ is the driver's per-trip operating cost. We first study a fixed-fare setting in this section and in Sections~\ref{sec:poa-bounds} and~\ref{sec:cross-subsidies}, in which the route-level fares $\Bar{p}_i$ are exogenously given. This captures many informal transit systems, including Nalasopara in Section~\ref{sec:experiments}, where fares remain stable over short- to medium-run horizons, and allows us to focus on equilibrium behavior under decentralized driver route choice. Section~\ref{sec:price-optimization} generalizes the model to allow fares to be optimized by a central planner.

As in Vickrey's seminal bottleneck model~\citep{vickrey1969congestion}, let $\Lambda_i$ be the mass of riders seeking to travel on route $i$ with desired arrival times at the destination uniformly distributed over the peak-period interval $[t_1, t_2]$, and let $\lambda_i = \frac{\Lambda_i}{t_2 - t_1}$ be the deterministic constant desired arrival rate. Then, riders on each route choose between the minibus (M) and an outside option (O), e.g., walking or a local public transit option, based on their travel costs defined in Section~\ref{subsec:rider-payoffs} (see right of Figure~\ref{fig:warmup-model}).

A central primitive of our model is the mapping from driver supply on a route to rider demand served by minibuses, as only a fraction of riders may be served when service capacity is limited. Let $\x = (x_i)_{i \in [n]}$ denote the driver allocation across routes, where $x_i \geq 0$ is the mass of drivers serving route $i$, and let $\Omega = \{ \x \in \mathbb{R}_{\geq 0}^n: \sum_{i \in [n]} x_i = D \}$ denote the feasible set of driver allocations. \color{blue} Our framework also naturally extends to the case when $\sum_{i\in[n]}x_i\leq D$, allowing drivers to opt out via an outside-option route (see Remark~\ref{rmk:dummyroute} in Appendix~\ref{apdx:welf-max-res-wage}); we focus on the equality-constrained case for analytical clarity. \color{black} Given a driver allocation $\x$, we define $\LLambda^M(\x) = (\Lambda_i^M(x_i))_{i \in [n]}$ as the vector of rider demands served by the minibuses, where $\Lambda_i^M(x_i)$ is the number of riders served when $x_i$ drivers operate on that route. The remaining $\Lambda_i - \Lambda_i^M(x_i)$ riders take the outside option. 

We let $\mathcal{L}$ denote the class of admissible rider demand functions $\LLambda^M(\cdot)$, induced under particular modeling choices for rider costs and queuing behavior, which we specify in Section~\ref{subsec:rider-payoffs}. Then, we have the following description of an instance of an informal transit system.

\vspace{-2pt}
\begin{definition}[The Informal Transit System]\label{def:informal_transit_instance}
An instance of an informal transit system is specified by the tuple $I = (n, F, \overline{\textbf{p}}, \textbf{l}, \textbf{c}, t_1, t_2, D, \LLambda^M)$, where $n$ is the number of routes, $F$ is the vehicle capacity, and $\overline{\textbf{p}}, \textbf{l}, \textbf{c} \in \mathbb{R}^n$ denote the vectors of per-trip fares, travel times, and driver costs for the $n$ routes. $[t_1, t_2]$ is the peak-period interval over which riders' desired destination arrival times are uniformly distributed, $D$ is the total mass of minibus drivers, and $\LLambda^M : \mathbb{R}^n \rightarrow \mathbb{R}^n$ specifies the rider demand for each route as a function of the driver supply on that route. Then, under a rider demand function family $\mathcal{L}$, the set of all instances of the informal transit system is given by:
{\setlength{\abovedisplayskip}{0pt}
\setlength{\belowdisplayskip}{0pt}
\setlength{\jot}{1pt}
\begin{align*}
    \mathcal{I}_{\mathcal{L}} := \{ (n, F, \overline{\textbf{p}}, \textbf{l}, \textbf{c}, t_1, t_2, D, \LLambda^M) : n\geq 1, \overline{\textbf{p}}, \textbf{l}, \textbf{c} \in \mathbb{R}^n_{> 0}, F, t_1, t_2, D \in \mathbb{R}_{> 0}, t_2>t_1, \LLambda^M \in \mathcal{L} \}.
\end{align*}}
\end{definition}

\vspace{-15pt}
\begin{figure}[tbh!]
      \centering
      \includegraphics[width=140mm]{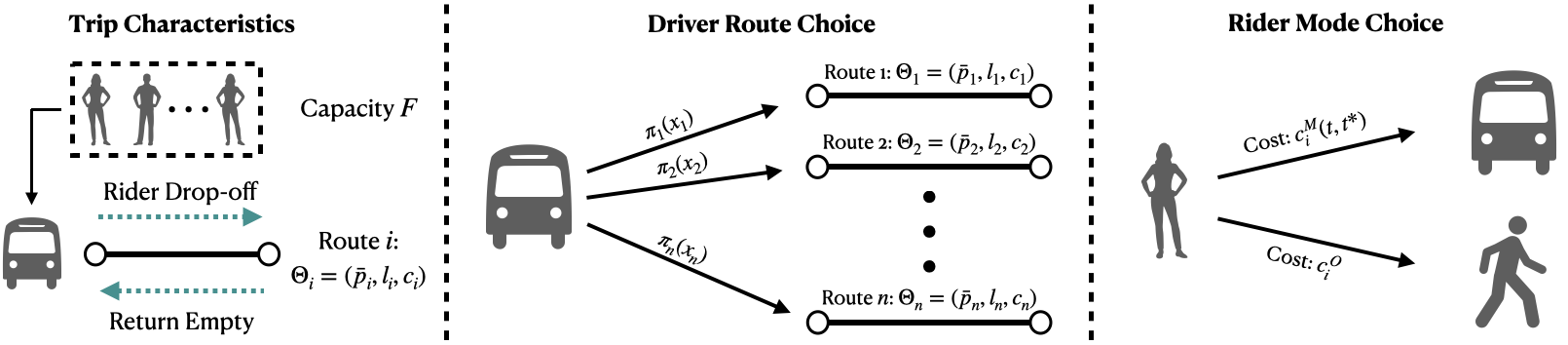}
      \vspace{-9pt}
      \caption{\small \sf Characteristics of an informal transit system in which minibus drivers serve riders on a fixed menu of routes. The left panel depicts minibus trips, where each minibus transports $F$ riders from a route’s origin to its destination and returns empty to the origin to serve additional riders. The middle panel illustrates driver route choice, with each driver selecting one of $n$ routes based on profitability. The right panel shows the mode choice decision of riders who either travel by minibus or take the outside option.
      }
      \label{fig:warmup-model} 
   \end{figure}    

\vspace{-20pt}

\subsection{Minibus Rider Demand Function} \label{subsec:rider-payoffs}
\vspace{-2pt}
This section specifies the class $\mathcal{L}$ of admissible demand functions $\LLambda^M(\cdot)$ studied in this work. We consider demand functions arising endogenously from equilibrium rider behavior, for each level of driver supply, under an adaptation of Vickrey’s bottleneck model of peak-period congestion, in which rider costs reflect queuing delays, a central feature of informal transit during morning or evening peaks. While we focus on this class of demand functions, our formulation is sufficiently general as it subsumes an important special case of capacity-constrained demand, closely related to formulations commonly studied in adjacent domains such as ride-sharing, where rider queuing delays are abstracted away~\citep{bimpikis2019spatial}, and central to our results (see Propositions~\ref{prop:tightness-poa-profit}, \ref{prop:tightness-poa-demand}, and Theorem~\ref{thm:fare-opt-profit-poa}). That said, extending our analysis to alternate demand classes arising from different modeling choices of rider costs and queuing behavior is an important direction for future work.

\textbf{Demand Function Under Rider Queuing:} To specify the rider demand function class we study, we first describe rider behavior under an adaptation of Vickrey's bottleneck model~\citep{vickrey1969congestion,GONZALES20121519}. On each route $i \in [n]$, cost-minimizing riders choose either an outside option (e.g., walking) with a fixed cost $c_i^O$ or a minibus service. Given $x_i$ drivers operating on route $i$, the cost of using the minibus depends on its service capacity $\mu_i(x_i) = \frac{x_i F}{2 l_i }$, where $F$ is the vehicle capacity and $2l_i$ is the round-trip travel time. In the \emph{over-supplied} regime, when the service capacity exceeds the desired rider arrival rate, i.e., $\mu_i(x_i) \geq \lambda_i$, all riders can arrive at the destination at their desired arrival times when using the minibus without queuing or schedule delays. In this case, the cost of using the minibus is a constant given by $c_i^M = \Bar{p}_i + \eta_T l_i$, where $\eta_T$ represents the common value of time across riders. Without loss of generality, we assume that the minibus absent rider queuing and schedule delays is weakly more attractive than the outside option, i.e., $S_i := c_i^O - c_i^M \geq 0$ for all routes $i$; otherwise, such routes can be removed from consideration, as riders would always prefer the outside option regardless of minibus supply.

In the \emph{under-supplied} regime when $\mu_i(x_i) < \lambda_i$, rider queuing and schedule delays may arise, as is common during morning or evening peaks in informal transit systems. In this case, riders, beyond choosing their mode (i.e., outside option or minibus), select their arrival time when using the minibus. Accordingly, in addition to the fixed cost $c_i^M$, the cost of using the minibus consists of (i) queuing (or waiting) delays $w(t)$, representing the time spent waiting for the minibus by a rider whose trip is completed at time $t$, and (ii) schedule delays $|t^* - t|$, capturing the discrepancy between a rider's desired destination arrival time $t^*$ and actual arrival time $t$. Let $\eta_W$ be the penalty for incurring a unit of waiting time delay, and $\eta_E$ and $\eta_L$ be the schedule delay penalties for arriving early or late, respectively, with $0 < \eta_E < \eta_W$ and $\eta_L > 0$, as is standard in the literature. Then, normalizing the waiting time penalty $\eta_W = 1$ (with all other parameters scaled accordingly), the minibus travel cost for a user with a desired destination arrival time $t^*$ and actual arrival time $t$ is: 
{\setlength{\abovedisplayskip}{1pt}
\setlength{\belowdisplayskip}{1pt}
\setlength{\jot}{1pt}
\begin{align} \label{eq:cost-minibus-queuing}
    c_i^M(t, t^*) = \Bar{p}_i + \eta_T l_i + w(t) + \eta_E(t^* - t)_+ + \eta_L(t - t^*)_+.
\end{align}}
Under this cost structure, for a given driver supply $x_i$ on a route $i$, a \emph{rider equilibrium} arises when all $\Lambda_i$ riders choose a mode (minibus or outside option) and, if using the minibus, an arrival time, to minimize travel costs, with no rider having an incentive to deviate. For a fixed driver supply $x_i$, the resulting equilibrium is akin to that in Vickrey's bottleneck model with an outside option~\citep{GONZALES20121519}, and as $x_i$ varies, these equilibria trace out the rider demand function $\Lambda_i^M(\cdot)$ depicted on the left of Figure~\ref{fig:induced-rider-demand}. We provide a closed-form characterization of this demand function in Equation~\eqref{eq:rider-demand-function} in Proposition~\ref{prop:rider-demand-func-queuing} (see Section~\ref{subsec:rider-demand-derivation-vickrey}), along with a discussion of the regimes that arise. Given this characterization, we define the class $\mathcal{L}_V$ as the set of all rider demand functions in Equation~\eqref{eq:rider-demand-function} for arbitrary choices of strictly positive primitives $F, \overline{\textbf{p}}, \textbf{l}, \textbf{c}, t_2, t_1, \eta_E, \eta_L, (\Lambda_i)_{i \in [n]}$ and non-negative values of $(S_i)_{i \in [n]}$, subject to the parameter restrictions specified above. Recall that $S_i = c_i^O - c_i^M$ captures the relative attractiveness of the minibus absent queuing and schedule delays. 

\textbf{Capacity-Constrained Demand Function:} We consider a special case of the above demand function, with $S_i = 0$ on route $i$. Here, no queuing or schedule delays arise at equilibrium, as such delays would make the minibus more costly than the outside option. \color{blue} To resolve indifference between the minibus and the outside option, we select the rider equilibrium with maximal minibus demand for each driver supply $x_i$. \color{black} Hence, the demand function reduces to Equation~\eqref{eq:rider-demand-function-reduction} (see Section~\ref{subsec:profit-poa}), exhibiting a piecewise-linear relation on the right of Figure~\ref{fig:induced-rider-demand}, with demand increasing linearly with driver supply up to a threshold $k_i^* = \frac{2 l_i \lambda_i}{F}$ at which $\mu_i(x_i) = \lambda_i$, beyond which all riders are served.

The capacity-constrained demand function mirrors standard equilibrium relations where the served demand is limited by the minimum of the demand arrival rate and service capacity, with full demand capture occurring only once the saturation condition $\mu_i(x_i) = \lambda_i$ holds. While both demand functions in Figure~\ref{fig:induced-rider-demand} coincide in the over-supplied regime (i.e., $\mu_i(x_i) \geq \lambda_i$) when driver supply exceeds $k_i^*$ and all riders are served, they diverge in the under-supplied regime. In this regime, unlike the $S_i = 0$ case corresponding to the capacity-constrained demand function, when $S_i > 0$, riders are willing to incur queuing or schedule delays, allowing the minibus to capture additional demand even when $\mu_i(x_i) < \lambda_i$. In this case, we show that the demand function is concave quadratic below a threshold $\Tilde{k}_i^*$, which need not coincide with the saturation threshold at which $\mu_i(x_i) = \lambda_i$, i.e., $k_i^* \neq \Tilde{k}_i^*$; in particular, there are regimes where $\mu_i(x_i)<\lambda_i$, yet all arriving riders are served by the minibus at equilibrium. For a further discussion on these demand functions, see Section~\ref{subsec:rider-demand-derivation-vickrey}.

\vspace{-14pt}

\begin{figure}[tbh!]
      \centering
      \includegraphics[width=100mm]{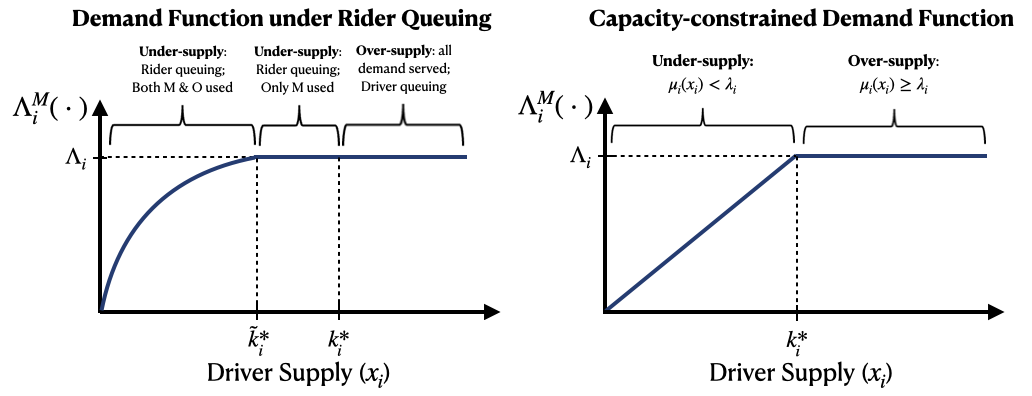}
      \vspace{-10pt}
      \caption{\small \sf Minibus demand function with rider queuing (left) and under a capacity-constrained formulation without rider queuing (right). When rider queuing is incorporated via a Vickrey-style bottleneck model, the minibus rider demand is concave quadratic below a threshold $\Tilde{k}_i^* \leq k_i^*$, and remains flat thereafter. When the cost difference $S_i = 0$, this demand function reduces to the capacity-constrained formulation, where the minibus demand increases linearly with driver supply until $k_i^*$ at which $\mu_i(x_i) = \lambda_i$, and is flat thereafter.
      }
      \label{fig:induced-rider-demand} 
   \end{figure}    

\vspace{-20pt}
\vspace{-5pt}
\subsection{Driver Payoffs and Equilibrium Notion} \label{subsec:eq-def}
\vspace{-3pt}
We now describe the supply-side of the model and how drivers allocate themselves across routes. Motivated by the profit motive of informal transit operators, we assume drivers choose routes to maximize individual profits (see middle of Figure~\ref{fig:warmup-model}). Let $p_i = \Bar{p}_i - \frac{c_i}{F} \geq 0$ denote the per-rider profit on route $i$, equal to the trip fare minus per-rider operating cost. Then, given $x_i$ drivers serving $\Lambda_i^M(x_i)$ riders on route $i$, the total route profit is $p_i \Lambda_i^M(x_i)$, which we assume is shared equally among drivers, yielding a per-driver profit $\pi_i(x_i) = \frac{\Lambda_i^M(x_i) p_i}{x_i}$. Under the demand functions in Equation~\eqref{eq:rider-demand-function}, $\pi_i(\cdot)$ is well-defined and bounded for $x_i> 0$, and we define per-driver profits on unused routes by their limiting values $\pi_i(0):=\lim_{x_i\to 0^+}\pi_i(x_i)$. Drivers select routes to maximize this payoff, taking the demand functions and the allocation of other drivers as given. This strategic interaction induces a Wardrop-style~\citep{wardrop-ue} equilibrium driver allocation, in which no driver can profitably deviate. We formalize this notion below and let $\Omega^{Eq}$ denote the set of equilibrium driver allocations.

\vspace{-2pt}
\begin{definition}[Equilibrium Driver Allocation] \label{def:eq-def}
Consider a driver allocation $\x$, with $\Lambda_i^M(\cdot)$ denoting the minibus rider demand function on each route $i$. Then, $\x$ is an equilibrium driver allocation if for any route $i$ with $x_{i}>0$, it holds that $\pi_i(x_i) \geq \pi_j(x_j)$ for all routes $j \in [n]$.
\end{definition}
\color{blue}
This equilibrium notion provides a natural benchmark for decentralized route choice, as differences in earnings across routes create incentives for drivers to shift to more profitable ones~\citep{doi:10.1287/msom.2023.1221,doi:10.1287/stsy.2018.0018}.
\color{black}

\vspace{-4pt}
\subsection{Price of Anarchy} \label{subsec:poa}
\vspace{-2pt}
Since drivers selfishly maximize individual profits, the equilibrium allocation of drivers across routes may not be optimal for a central planner (e.g., government or driver association). We evaluate the equilibrium outcomes for an informal transit instance $I = (n, F, \overline{\textbf{p}}, \textbf{l}, \textbf{c}, t_1, t_2, D, \LLambda^M)$ using two metrics: (i) \emph{cumulative driver profits}, defined as $P_I(\x) = \sum_{i = 1}^n p_i \Lambda_i^M(x_i)$, an aggregate notion of driver welfare, and (ii) \emph{rider throughput} $R_I(\x) = \sum_{i = 1}^n \Lambda_i^M(x_i)$, which measures the cumulative rider demand served by the minibus system and is a central policy objective in transit planning~\citep{bertsimas2021data,YU201258}. \color{blue}
These metrics provide directly interpretable measures of system performance from both sides of the market. Cumulative driver profit captures the financial incentives of the drivers, while rider throughput measures how effectively the system meets passenger travel demand. While we focus on these metrics, Appendix~\ref{apdx:total-social-welfare-obj} complements the analysis in Section~\ref{sec:poa-bounds} by considering a total social welfare metric, which combines driver profit with riders' gains relative to their outside options and explicitly incorporates waiting and schedule delay costs.
\color{black}

To quantify the performance loss under these metrics due to selfish driver route choices, we adopt the classical notion of the \emph{price of anarchy} (PoA), which quantifies the worst-case ratio (over problem instances) between the equilibrium system objective and optimal system objective obtained under a centralized authority coordinating the drivers' actions. 
\vspace{-2pt}
\begin{definition}[Price of Anarchy] \label{def:poa}
For a given family of demand functions $\mathcal{L}$, the profit PoA, denoted $\text{P-PoA}(\mathcal{L})$, is defined to be the worst case profit ratio over informal transit instances (see Definition~\ref{def:informal_transit_instance}) of the optimum cumulative driver profit to that achieved by an equilibrium driver allocation. The rider throughput PoA, denoted $\text{R-PoA}(\mathcal{L})$, is defined similarly.
{\setlength{\abovedisplayskip}{2pt}
\setlength{\belowdisplayskip}{4pt}
\setlength{\jot}{1pt}
\begin{align}
    \text{P-PoA}(\mathcal{L}) = \sup_{I \in \mathcal{I}_{\mathcal{L}}} \sup_{\x^* \in \Omega} \sup_{\x^{Eq} \in \Omega^{Eq}} \frac{P_I(\x^*)}{P_I(\x^{Eq})} \quad \quad \quad \text{R-PoA}(\mathcal{L}) = \sup_{I \in \mathcal{I}_{\mathcal{L}}} \sup_{\x^* \in \Omega} \sup_{\x^{Eq} \in \Omega^{Eq}} \frac{R_I(\x^*)}{R_I(\x^{Eq})}
\end{align}}
\end{definition}
When both the optimal and equilibrium objective values are zero, we define the corresponding ratio to be one. In the remainder of this work, we focus on the demand function class $\mathcal{L}_V$ induced under an adaptation of Vickrey's bottleneck model (see Section~\ref{subsec:rider-payoffs}). Hence, for notational simplicity, we will drop the dependency on $\mathcal{L}$ and use $\text{P-PoA}, \text{R-PoA}$ in place of $\text{P-PoA}(\mathcal{L}), \text{R-PoA}(\mathcal{L})$ respectively. Similarly, we will use $P, R$ instead of $P_I, R_I$ when the choice of instance $I$ is unambiguous. 

\vspace{-5pt}
\subsection{Additional Discussion of Modeling Assumptions} \label{subsec:model-assumptions}
\vspace{-2pt}

\color{blue}
We provide additional discussion of our modeling assumptions, which are motivated by documented operating practices in informal transit systems, including established routes, standardized fares, and demand-responsive service provision, reflected in our numerical setting in Nalasopara, India~\citep{kandathil2026nalasopara}, as well as evidence of drivers adjusting route choices in response to profitability differences across informal transit settings~\citep{cervero2000informal,bjorkegren2025public}.

First, in modeling a minibus trip (left of Figure~\ref{fig:warmup-model}), we assume that after dropping off riders, drivers return empty to the origin and continue serving the same route, reflecting informal transit operations during morning or evening peaks when demand is highly unidirectional~\citep{cervero2000informal,CERVERO2007445}. For instance, substantial empty running by auto-rickshaws, including empty return trips, has been documented in India~\citep[Section~2.3.3]{khatua2017report}, consistent with our empirical Nalasopara setting, while \citet[Section~4.2]{bjorkegren2025public} report that 87\% of surveyed minibus drivers in Lagos operate from an origin to a destination and back again.

Next, we assume minibuses serve riders at capacity $F$, consistent with widespread fill-and-depart practices in informal transit systems. In particular, 96\% of minibuses in Lagos depart terminals only once full~\citep[Section~3.3]{bjorkegren2025public}, tro-tros in Ghana similarly wait until all seats are occupied before departure~\citep{arroyoarroyo2021connecting}, and comparable operating practices are documented more broadly across informal transit systems~\citep{kumar2022myths}.

Moreover, while we assume a round-trip travel time of $2 l_i$, our results also hold when the outbound and return legs have asymmetric travel times. Specifically, letting \(l_i^{\mathrm{out}}\) and \(l_i^{\mathrm{ret}}\) denote the travel times of the two legs of the trip and \(r_i=l_i^{\mathrm{out}}+l_i^{\mathrm{ret}}\), it suffices to replace the service rate \(\mu_i(x_i) = \frac{x_iF}{2l_i}\) by \(\frac{x_iF}{r_i}\), while the rider's minibus travel cost depends on \(l_i^{\mathrm{out}}\). This substitution preserves the derivation and functional form of rider demand, with only the route-specific coefficients adjusted to reflect the generalized round-trip time. Consequently, all subsequent equilibrium analysis and results continue to hold, including the PoA bounds, which do not depend on round-trip travel times (see Appendix~\ref{app:asymmetric-travel-times}). 

\color{black}




Next, we abstract from physical road congestion and assume fixed travel times $l_i$ independent of driver supply on a route, consistent with prior work~\citep{bjorkegren2025public,Conwell2025PrivatizedTransit}. Instead, our congestion game captures \emph{competition among drivers for riders}, reflecting how service over or under-provision affects payoffs in informal transit systems. \color{blue} That said, the PoA upper bounds in Section~\ref{sec:poa-bounds} continue to hold under physical road congestion, where route travel times increase with the number of minibuses actually circulating on the route (see Appendix~\ref{apdx:road-congestion}). \color{black}

Further, consistent with operational realities and institutional norms in informal transit systems, we assume drivers serve a single route throughout the study horizon (e.g., the morning or evening peak) and do not switch routes between trips. \color{blue} Such short-run route commitment is a reasonable approximation over the peak-period horizon we study and is also observed in our Nalasopara setting. We also show in Appendix~\ref{apdx:route-switching} that the gains from route switching are bounded at equilibrium, so the equilibrium remains stable whenever switching costs exceed this bound. That said, studying alternative models with endogenous route switching is an interesting direction for future research. \color{black}


Finally, our assumption of a fixed menu of routes reflects many informal transit systems, where the set of routes remains stable over short to medium horizons (e.g., several months), as in Nalasopara. Given a set of routes, our analysis focuses on the incentive and equilibrium effects in informal transit systems, where changes to route structures are often infeasible or of limited relevance due to regulatory constraints. In this sense, our work serves as a natural starting point for studying higher-level planning decisions such as route design, particularly when new informal transit systems are planned or old ones substantially redesigned. Overall, while informal transit systems are complex, our model abstracts from some operational details to isolate the core forces at play while remaining faithful to how these systems operate in practice. These abstractions enable a tractable framework that yields clear insights into the incentives and equilibrium behavior in informal transit systems, and extending the model to incorporate additional operational features is a valuable future direction.

\vspace{-6pt}
\section{Price of Anarchy Bounds} \label{sec:poa-bounds}
\vspace{-2pt}
We begin the study of our informal transit system through a PoA analysis under both cumulative driver profit and rider throughput metrics. To this end, under the rider behavior model in Section~\ref{subsec:rider-payoffs}, we first derive the minibus rider demand as a function of driver supply and analyze its properties in Section~\ref{subsec:rider-demand-derivation-vickrey}. Leveraging this characterization, \color{blue} we establish a tight PoA bound for cumulative driver profit  and an instance-wise bound for rider throughput, whose dependence on per-rider profit heterogeneity is tight, \color{black} in Sections~\ref{subsec:profit-poa} and~\ref{subsec:demand-poa}, respectively. Our results show that decentralized, self-interested driver decisions can lead to substantial inefficiencies in both objectives. 

\vspace{-4pt}
\subsection{Derivation of Minibus Demand Function Under Rider Queuing} \label{subsec:rider-demand-derivation-vickrey}
\vspace{-2pt}
This section derives the rider demand served by the minibus as a function of driver supply on a route and studies its properties, which underpin our PoA analysis in Sections~\ref{subsec:profit-poa} and~\ref{subsec:demand-poa}. 

We first characterize, for any driver supply $x_i$ on a route, the total demand served by the minibus at equilibrium under the costs defined in Equation~\eqref{eq:cost-minibus-queuing}. In presenting this result, we focus on the regime when the cost difference between the outside option and the minibus without rider queuing and schedule delays, given by $S_i = c_i^O - c_i^M$, is non-negative. Recall that when $S_i < 0$, regardless of the driver supply, the outside option is always preferred to the minibus, i.e., $\Lambda_i^M(x_i) = 0$ for all $x_i$.
\vspace{-3pt}
\begin{proposition}[Minibus Demand Under Rider Queuing] \label{prop:rider-demand-func-queuing}
Consider an informal transit system with a fixed menu of routes, where riders on each route $i \in [n]$ choose either the minibus and incur a cost as defined in Equation~\eqref{eq:cost-minibus-queuing} or an outside option with a fixed cost $c_i^O$. Further, suppose that $S_i = c_i^O - c_i^M \geq 0$, and define the driver allocation threshold $\Tilde{k}_i^* = \min \left\{ \frac{2 l_i \lambda_i}{F}, \frac{2 l_i \Lambda_i \eta_E \eta_L}{F S_i (\eta_E + \eta_L)} \right\}$, where the second term in the minimum is interpreted as $+\infty$ when $S_i=0$. Then: 
\vspace{-8pt}
\begin{align} \label{eq:rider-demand-function}
    \Lambda_i^M(x_i) = 
    \begin{cases}
        \frac{F}{2l_i} \left( t_2 - t_1 + S_i \left( \frac{\eta_E + \eta_L}{\eta_E \eta_L} \right) \right) x_i - \Big(\left(\frac{F}{2l_i}\right)^2 S_i \frac{\eta_E + \eta_L}{\eta_E \eta_L} \frac{t_2 - t_1}{\Lambda_i} \Big) x_i^2, & \text{if } x_i \in [0, \Tilde{k}_i^*) \\[-4pt]
        \Lambda_i, & \text{if } x_i \geq \Tilde{k}_i^*
    \end{cases}
\end{align}
\end{proposition}
\vspace{-12pt}

\begin{proof}[Proof Sketch.]
Fix a route $i$ with $S_i \geq 0$ and $x_i$ drivers, resulting in a service rate of $\mu_i(x_i) = \frac{x_i F}{2 l_i}$. Then, defining the time horizon over which drivers service the route as $T_i(x_i)$, the total mass of riders served is: $\Lambda_i^M(x_i) = \min \{ \mu_i(x_i), \lambda_i \} T_i(x_i)$. Note $T_i(x_i) \geq \Delta := t_2 - t_1$, as riders may be willing to arrive earlier or later than their desired destination arrival times under the costs in Equation~\eqref{eq:cost-minibus-queuing}. 

To establish an expression for $T_i(x_i)$, we first characterize the equilibrium of riders’ mode choice (minibus versus outside option) and arrival-time decisions on a route $i$ with $x_i$ drivers by adapting the equilibrium characterization in Vickrey’s bottleneck model with an outside option from~\cite{GONZALES20121519} to our setting (see Proposition~\ref{prop:gonzales-eq-characterization} in Appendix~\ref{apdx:pf-rider-demand-gonzalez}). Combining this equilibrium characterization with the property that riders' desired destination arrival times are uniformly distributed, we derive a closed-form expression for the total service time. When the system is under-supplied (i.e., $\mu_i(x_i) < \lambda_i$), $T_i(x_i) = \Delta + \min \{ S_i, \Bar{S}(x_i) \} \left( \frac{\eta_E + \eta_L}{\eta_E \eta_L} \right) \left( 1 - \frac{\Delta x_i F}{2 l_i \Lambda_i} \right)$, where $\Bar{S}(x_i) \! = \! \frac{\Lambda_i \eta_E \eta_L}{\mu_i(x_i)(\eta_E + \eta_L)}$ (see Corollary~\ref{cor:service-time} in Appendix~\ref{apdx:pf-rider-demand-gonzalez}). When the system is over-supplied, $T_i(x_i) \! = \! \Delta$. 

Finally, we evaluate $\Lambda_i^M(x_i) = \min\{\mu_i(x_i),\lambda_i\}T_i(x_i)$ by substituting the derived service time relation in the different regimes for driver supply, yielding the expression in Equation~\eqref{eq:rider-demand-function}.
\end{proof}


For a complete proof, see Appendix~\ref{apdx:pf-rider-demand-gonzalez}. Proposition~\ref{prop:rider-demand-func-queuing} implies that the rider demand function is concave quadratic up to a driver supply threshold $\Tilde{k}_i^*$, beyond which all rider demand is served. The resulting demand function in Equation~\eqref{eq:rider-demand-function}, shown on the left of Figure~\ref{fig:induced-rider-demand}, highlights a key departure of our setting from standard equilibrium models, in which full demand capture occurs only once the saturation condition $\mu_i(x_i) = \lambda_i$ is met, as with the capacity-constrained demand function on the right of Figure~\ref{fig:induced-rider-demand}. Specifically, Proposition~\ref{prop:rider-demand-func-queuing} implies that all rider demand is served at a threshold $\Tilde{k}_i^*$, which may be strictly smaller than the saturation threshold $k_i^* = \frac{2 l_i \lambda_i}{F}$ at which $\mu_i(x_i) = \lambda_i$. Thus, all rider demand can be served even in an under-supplied system ($\mu_i(x_i)<\lambda_i$) if the outside option is sufficiently unattractive, albeit while inducing rider queuing. Once the driver supply exceeds the saturation threshold $k_i^* = \frac{2 l_i \lambda_i}{F}$, rider queuing is eliminated at equilibrium.

We now leverage Proposition~\ref{prop:rider-demand-func-queuing} to establish key continuity and monotonicity properties. 

\vspace{-4.5pt}
\begin{corollary}[Continuity and Monotonicity of Demand] 
\label{cor:monotonicity-rider-demands}
The demand function $\Lambda_i^M(\cdot)$ in Equation~\eqref{eq:rider-demand-function} is continuous and non-decreasing in the driver allocation $x_i$ for every route $i$.
\end{corollary}

\vspace{-9pt}
\begin{corollary}[Continuity and Monotonicity of Per-Driver Profits] 
\label{cor:monotonicity-per-driver-profits}
For every route $i$ with a demand function $\Lambda_i^M(\cdot)$ in Equation~\eqref{eq:rider-demand-function}, the per-driver profit $\pi_i(x_i) \! = \! \frac{\Lambda_i^M(x_i) p_i}{x_i}$ is continuous and non-increasing in the driver allocation $x_i$.
\end{corollary}
\vspace{-3.5pt}

For proofs of Corollaries~\ref{cor:monotonicity-rider-demands} and~\ref{cor:monotonicity-per-driver-profits}, see Appendices~\ref{apdx:pf-cor1} and~\ref{apdx:pf-cor2}, respectively. Corollary~\ref{cor:monotonicity-rider-demands} implies a natural monotonicity condition that the served rider demand is non-decreasing in the driver supply, implying that the cumulative profit $p_i \Lambda_i^M(x_i)$ on a route is also non-decreasing. However, Corollary~\ref{cor:monotonicity-per-driver-profits} shows that per-driver profit is non-increasing in driver supply, a consequence of the concavity of the rider demand function in Equation~\eqref{eq:rider-demand-function} and, hence, of the profit function $p_i \Lambda_i^M(x_i)$.

\vspace{-5pt}
\subsection{Price of Anarchy for Cumulative Driver Profit} \label{subsec:profit-poa}
\vspace{-3pt}
Building on the results in the previous section, we analyze the PoA in our informal transit system with respect to cumulative driver profits. We show that the profit PoA is at most 2 and this bound is tight: for any $\varepsilon > 0$, there exists an instance $I \in \mathcal{I}_{\mathcal{L}_V}$ for which the profit ratio is at least $2 - \varepsilon$. 
\vspace{-3pt}
\begin{theorem}[PoA Upper Bound for Cumulative Driver Profit] \label{thm:poa-driver-profit}
Consider an informal transit system with a class $\mathcal{L}_V$ of rider demand functions $(\Lambda_i^M(\cdot))_{i \in [n]}$ in Equation~\eqref{eq:rider-demand-function}. Then, $\text{P-PoA} \leq 2$.
\end{theorem}

\vspace{-7pt}

\begin{proposition}[Tightness of P-PoA Bound] \label{prop:tightness-poa-profit}
For any $\varepsilon > 0$, there exists an informal transit instance $I \in \mathcal{I}_{\mathcal{L}_V}$ such that the profit ratio exceeds $2-\varepsilon$.
\end{proposition}

\vspace{-3.5pt}

These results show that in the worst case, selfish driver behavior can reduce cumulative driver profits to half of the centralized optimum, implying a substantial efficiency loss. Yet, unlike settings with unbounded PoAs~\citep{kannan2013price}, the resulting loss in cumulative driver profits remains bounded due to the concavity of the demand function in Equation~\eqref{eq:rider-demand-function}. We now prove these results. 

\vspace{-3pt}

\begin{proof}[Proof of Theorem~\ref{thm:poa-driver-profit}.]
Fix an informal transit instance $I$ with rider demand functions given by Equation~\eqref{eq:rider-demand-function} for all routes $i \in [n]$. Since our proof applies to all feasible instances $I$, we drop the dependency on $I$ in our notation. For this instance, let $\x^*$ be the cumulative profit-maximizing allocation, i.e., $\x^* \in \argmax_{\x \in \Omega} P(\x) = \sum_{i = 1}^n p_i \Lambda_i^M(x_i)$, and let $\x^{Eq}$ be any equilibrium allocation. Moreover, let $Q$ be the set of routes for which the cost difference $S_i \geq 0$. Note for any route $i \in [n] \backslash Q$, the minibus serves no rider demand and generates no profit regardless of the allocation. 

At any equilibrium $\x^{Eq}$, all routes $i$ with $x_i^{Eq} > 0$ have the same per-driver profit $\pi^{Eq} = \frac{\Lambda_i^M(x_i^{Eq}) p_i}{x_i^{Eq}}$. Further, by Definition~\ref{def:eq-def}, $\pi^{Eq} \geq \pi_j(x_j^{Eq})$ for all $j \in [n]$. Thus, the total profit under $\x^{Eq}$ can be re-expressed as $P(\x^{Eq}) = \sum_{i \in Q} p_i \Lambda_i^M(x_i^{Eq}) \stackrel{(a)}{=} \sum_{i \in Q} \pi_i(x_i^{Eq}) x_i^{Eq} \stackrel{(b)}{=} \sum_{i \in Q} \pi^{Eq} x_i^{Eq},$ where (a) follows from the definition of $\pi_i(\cdot)$ and (b) follows as $\pi^{Eq} = \pi_i(x_i^{Eq})$ for all routes $i$ with $x_i^{Eq} > 0$. Analogously, letting $\pi_i^* = \pi_i(x_i^*)$ be the per-driver profit on each route $i \in Q$ at $\x^*$, the optimal cumulative profit is $P(\x^*) = \sum_{i \in Q} \pi_i^* x_i^*$. Then, defining $L_1 = \{ i \in Q: x_i^{Eq} \geq x_i^* \}$ and $L_2 = \{ i \in Q: x_i^{Eq} < x_i^* \}$: 
{\setlength{\abovedisplayskip}{-1pt}
\setlength{\belowdisplayskip}{-4pt}
\setlength{\jot}{0pt}
\begin{align*}
    P(\textbf{x}^*) &= \sum_{i \in Q} \pi_i^* x_i^* = \sum_{i \in L_1} \pi_i^* x_i^* + \sum_{i \in L_2} \pi_i^* x_i^* \stackrel{(a)}{\leq} \sum_{i \in L_1} \pi^{Eq} x_i^{Eq} + \sum_{i \in L_2} \pi_i^* x_i^* \stackrel{(b)}{\leq} P(\x^{Eq}) + \sum_{i \in L_2} \pi_i^* x_i^*, \\
    &\stackrel{(c)}{\leq} P(\x^{Eq}) + \sum_{i \in L_2} \pi^{Eq} x_i^* \stackrel{(d)}{\leq} P(\x^{Eq}) + P(\x^{Eq}) = 2 P(\x^{Eq}) 
\end{align*}}
where (a) follows from Corollary~\ref{cor:monotonicity-rider-demands}, as $x_i^{Eq} \geq x_i^*$ for all $i \in L_1$, (b) follows as $\sum_{i \in L_1} \pi^{Eq} x_i^{Eq} \leq \sum_{i \in Q} \pi^{Eq} x_i^{Eq}$ as $L_1 \subseteq Q$, (c) follows as $\pi^{Eq} \geq \pi_i^*$ for all $i \in L_2$ by Corollary~\ref{cor:monotonicity-per-driver-profits}, and (d) follows as $\sum_{i \in L_2} \pi^{Eq} x_i^{*} \leq \sum_{i \in Q} \pi^{Eq} x_i^{*} = \pi^{Eq} \sum_{i \in Q} x_i^{Eq}$ as $L_2 \subseteq Q$ and $\sum_{i \in Q} x_i^* = \sum_{i \in Q} x_i^{Eq} = D$. 
The above analysis holds for any equilibrium allocation $\x^{Eq}$ and instance $I$, 
thus establishing our claim. 
\end{proof}

\vspace{-2pt}

\begin{proof}[Proof of Proposition~\ref{prop:tightness-poa-profit}.]
Consider an instance with $n = 2$ routes and a driver supply $D = 1$. Moreover, suppose that $S_i = 0$ on both routes. Then, the minibus rider demand function on both routes reduces to the capacity-constrained demand function shown on the right of Figure~\ref{fig:induced-rider-demand}: 
{\setlength{\abovedisplayskip}{0pt}
\setlength{\belowdisplayskip}{-1pt}
\setlength{\jot}{1pt}
\begin{align} \label{eq:rider-demand-function-reduction}
    \Lambda_i^M(x_i) = 
    \begin{cases}
        \frac{F (t_2 - t_1)}{2l_i} x_i, & \text{if } x_i \in [0, \Tilde{k}_i^*] \\[-2pt]
        \Lambda_i, & \text{if } x_i > \Tilde{k}_i^*,
    \end{cases}
\end{align}}
where $\Tilde{k}_i^* = k_i^* = \frac{2 l_i \lambda_i}{F}$. Given Equation~\eqref{eq:rider-demand-function-reduction}, consider an instance where $\Tilde{k}_1^* + \Tilde{k}_2^* = D = 1$. Then, the cumulative profit-maximizing allocation is $\x^* = (\Tilde{k}_1^*, \Tilde{k}_2^*)$, resulting in a total profit $P(\x^*) = p_1 \Lambda_1 + p_2 \Lambda_2$. Next, suppose that the per-driver profit on the two routes satisfies $\pi_1(1) = \pi_2(0)$, i.e., the system supports an equilibrium $\x^{Eq} = (1, 0)$. Then, the following equalities must hold: $p_1 \Lambda_1 = \pi_1(1) = \pi_2(0) = p_2 \frac{F(t_2 - t_1)}{2 l_2} = p_2 \Lambda_2 \frac{F}{2 l_2 \lambda_2} = p_2 \Lambda_2 \frac{1}{\Tilde{k}_2^*}$. Using this relation, the profit ratio $\frac{P(\x^*)}{P(\x^{Eq})} = \frac{p_1 \Lambda_1 + p_2 \Lambda_2}{p_1 \Lambda_1} = 1 + \frac{p_2 \Lambda_2}{p_1 \Lambda_1} = 1 + \Tilde{k}_2^*.$ Finally, taking the limit as $\Tilde{k}_2^* \rightarrow 1$ and $\Tilde{k}_1^* \rightarrow 0$ while satisfying $\Tilde{k}_1^* + \Tilde{k}_2^* = D = 1$, the above analysis implies that the profit ratio approaches two. 
\end{proof}
\vspace{-2pt}
\color{blue}
These proofs rely on key properties of our informal transit system established in Section~\ref{subsec:rider-demand-derivation-vickrey}. 
\color{blue}
The profit PoA upper bound in Theorem~\ref{thm:poa-driver-profit} relies on the monotonicity properties in Corollaries~\ref{cor:monotonicity-rider-demands} and~\ref{cor:monotonicity-per-driver-profits}. These properties play a role analogous to the rate-monotonicity structure underlying the factor-two PoA guarantee of \citet{li2026effect}. However, unlike \citet{li2026effect}, where the relevant monotonicity structure is imposed as an assumption of their resource allocation setting, in our framework these properties are derived from the model primitives through the rider demand characterization in Proposition~\ref{prop:rider-demand-func-queuing}.
\color{black}
For tightness, Proposition~\ref{prop:tightness-poa-profit} constructs an instance with $S_i = 0$ for all routes $i$. In this regime, the demand function reduces to a piecewise-linear form on the right of Figure~\ref{fig:induced-rider-demand}, implying that the capacity-constrained demand can attain the worst-case profit PoA.
\color{black}



\vspace{-5pt}
\subsection{Price of Anarchy for Rider Throughput} \label{subsec:demand-poa}
\vspace{-2pt}
This section extends the PoA analysis to the throughput metric. \color{blue} 
For any instance \(I\), letting \(p_{\max}(I)\) and \(p_{\min}(I) > 0\) denote the maximum and minimum per-rider minibus profits across the \(n\) routes, we show that the rider throughput ratio between the centralized optimum and any equilibrium is bounded above by $1 + \frac{p_{\max}(I)}{p_{\min}(I)}$. Moreover, the dependence of this bound on \(p_{\max}(I)/p_{\min}(I)\) is tight.

\vspace{-3.5pt}
\begin{theorem}[Instance-Wise Upper Bound for Rider Throughput] \label{thm:poa-demand-served}
Consider an informal transit system with a class $\mathcal{L}_V$ of rider demand functions $(\Lambda_i^M(\cdot))_{i \in [n]}$ in Equation~\eqref{eq:rider-demand-function}. Then, for any rider throughput maximizing allocation \(\mathbf{x}^*\) and any equilibrium allocation \(\mathbf{x}^{Eq}\), $\frac{R_I(\x^*)}{R_I(\x^{Eq})} \leq 1 + \frac{p_{\max}(I)}{p_{\min}(I)}$. 
\end{theorem}
\vspace{-8pt}
\begin{proposition}[Tightness of Rider Throughput Bound] \label{prop:tightness-poa-demand}
Fix any $0<p_{\min}\leq p_{\max}$. For any $\varepsilon > 0$, there exists an informal transit instance $I \in \mathcal{I}_{\mathcal{L}_V}$ with minimum and maximum per-rider profits \(p_{\min}\) and \(p_{\max}\), respectively, and an equilibrium allocation $\x^{Eq}$ such that $\frac{R_I(\x^*)}{R_I(\x^{Eq})} > 1 + \frac{p_{\max}}{p_{\min}}-\varepsilon$, where $\x^*$ is the rider throughput maximizing allocation. 
\end{proposition}
\vspace{-4pt}

For notational brevity, we henceforth write \(p_{\max}\) and \(p_{\min}\) for \(p_{\max}(I)\) and \(p_{\min}(I)\) when the instance \(I\) is clear from context. \color{black} These results highlight that selfish driver decisions can substantially reduce cumulative rider demand served to as little as a $1/(1 + \frac{p_{\max}}{p_{\min}})$ fraction of the centralized optimum that coordinates driver actions. \color{blue} Moreover, Theorem~\ref{thm:poa-demand-served} and Proposition~\ref{prop:tightness-poa-demand} imply that, for any class of instances satisfying \(\frac{p_{\max}}{p_{\min}}\leq\kappa\), the rider throughput PoA is at most \(1+\kappa\), and this bound is tight; without a uniform bound on \(\frac{p_{\max}}{p_{\min}}\), the rider throughput PoA over an unrestricted class of instances is unbounded. Since  $\frac{p_{\max}}{p_{\min}} \geq 1$, selfish driver behavior can be more detrimental to rider throughput than to cumulative driver profits in the worst case (see Section~\ref{subsec:profit-poa}). That said, the ratio $\frac{p_{\max}}{p_{\min}}$ is typically small in practice, as suggested by our numerical experiments (see Section~\ref{sec:experiments}); for instance, using data from Nalasopara, this ratio equals three, \color{blue} resulting in an instance-wise upper bound of four on the rider throughput ratio. Beyond rider throughput, Appendix~\ref{apdx:total-social-welfare-obj} extends the analysis to a total social welfare objective that combines cumulative driver profit with riders' gains relative to their outside options, thereby explicitly accounting for waiting and schedule delay costs. 

While the proof of Proposition~\ref{prop:tightness-poa-demand} follows analogously to the proof of Proposition~\ref{prop:tightness-poa-profit} (see Appendix~\ref{apdx:pf-tightness-poa-demand}), we emphasize that neither Theorem~\ref{thm:poa-demand-served} nor the analogous total social welfare guarantee in Appendix~\ref{apdx:total-social-welfare-obj} follows directly from the profit PoA bound in Theorem~\ref{thm:poa-driver-profit}. In particular, a naive application of the profit PoA bound or the corresponding factor-two PoA guarantee in the resource allocation setting of~\citet{li2026effect} would yield an instance-wise rider throughput bound of $2 \frac{p_{\max}}{p_{\min}}$. However, establishing the sharper bound in Theorem~\ref{thm:poa-demand-served} (see Appendix~\ref{apdx:pf-poa-rider-welfare}) and the corresponding instance-wise total social welfare bound in Appendix~\ref{apdx:total-social-welfare-obj} requires additional arguments and a different set of inequalities tailored to their respective objectives. \color{black} 

\color{blue}

\color{black}

\vspace{-5.5pt}
\section{Mechanism I: Cross-subsidization} \label{sec:cross-subsidies}
\vspace{-2pt}
To mitigate the efficiency losses from selfish driver behavior, we study mechanisms through which a public authority can steer informal transit operators toward improved system outcomes in terms of cumulative driver profit and rider throughput. This section studies \emph{cross-subsidization}, in which the public authority sets route-specific tolls and subsidies to influence driver route choices and shape service patterns. The key idea is that tolls can deter excessive driver entry on highly lucrative routes, with the resulting revenue used to subsidize service on less profitable routes. Such mechanisms have practical antecedents in adjacent domains such as ride-hailing~\citep{uber_boost_plus_2022}.

\color{blue}
In this section, we show that for any instance of our informal transit system, there exists a budget-balanced (i.e., zero net expenditure) cross-subsidy scheme that aligns individual driver incentives with any desired driver allocation as an equilibrium, and can be computed in polynomial time for allocations maximizing cumulative driver profit or rider throughput. 
We further establish conditions under which the induced equilibrium is unique and show that, when per-rider profits are strictly positive, a unique equilibrium can achieve either objective arbitrarily close to its centralized optimum. These results demonstrate that cross-subsidies can essentially eliminate the PoA inefficiencies from selfish driver behavior without incurring any net fiscal expenditure at the resulting equilibria.

\color{black}

We first introduce cross-subsidization and define the notion of a \emph{budget-balanced} cross-subsidy scheme. In our setting with cross-subsidies, a public authority assigns each route $i \in [n]$ a transfer $\tau_i \in \mathbb{R}$, which influences driver payoffs. Negative transfers represent tolls paid by drivers to operate on a route, while positive transfers represent subsidies. Then, under a driver allocation $x_i$ on route $i$, drivers earn transfer-adjusted profits $\Tilde{\pi}_i(x_i) = \pi_i(x_i) + \tau_i$. With these modified payoffs, the equilibrium driver allocation under cross-subsidies can be defined analogously to Definition~\ref{def:eq-def}, with the two notions coinciding when the transfer vector $\ttau = (\tau_i)_{i \in [n]} = \mathbf{0}$. Crucially, cross-subsidies affect only driver payoffs and have no direct impact on rider payoffs or the rider demand function, with riders being influenced only through induced changes in driver allocations across routes. We focus, in particular, on \emph{budget-balanced} cross-subsidy schemes, under which the induced equilibrium driver allocation $\x$ satisfies  $\sum_{i = 1}^n \tau_i x_i = 0$, so that the collected revenues finance driver subsidies. 

\color{blue}
We now characterize the existence of a budget-balanced cross-subsidy scheme that induces any target driver allocation $\x^*$, e.g., those maximizing cumulative driver profits or rider throughput, as an equilibrium, with the corresponding transfer-adjusted per-driver profits equal to a weighted average of their profits at $\x^*$ without cross-subsidies. We further identify conditions under which $\x^*$ is the unique induced equilibrium and show that, even when these conditions do not fully hold, cross-subsidies can induce a unique equilibrium achieving near-optimal cumulative driver profit or rider throughput without any net fiscal expenditure, provided per-rider profits are strictly positive.
\color{black}

\vspace{-4pt}
\begin{theorem}
\label{thm:cross-subsidies-optimal}
For any informal transit instance with a class $\mathcal{L}_V$ of demand functions in Equation~\eqref{eq:rider-demand-function}, let $\x^* \in \Omega$ be a target driver allocation with $x_j^* > 0$ for at least one route $j$. Then: \vspace{-5pt}\color{blue}
\begin{enumerate}
    \item \textbf{Weak Implementation:} There exists a budget-balanced cross-subsidization scheme $\ttau = (\tau_i)_{i \in [n]}$, which induces $\x^*$ as an equilibrium. Moreover, $\ttau$ can be derived in closed-form and induces equilibrium per-driver profits $\Tilde{\pi}^{Eq} = \frac{\sum_{i \in [n]} \pi_i^* x_i^*}{\sum_{i \in [n]} x_i^*}$, where $\pi_i^* = \pi_i(x_i^*)$. \vspace{-5pt}
    \item \textbf{Exact Strong Implementation:} If the minibus's cost advantage over the outside option (absent queuing and schedule delays), $S_i$, and the net profit per rider, $p_i$, are both strictly positive for all $i \in [n]$, then $\x^*$ is the unique equilibrium induced by this cross-subsidization scheme. \vspace{-5pt}
    \item \textbf{Approximate Strong Implementation:} If $p_i > 0$ for all $i \in [n]$ and $\x^*$ maximizes cumulative driver profit or rider throughput, then for any $\varepsilon > 0$, there exists a cross-subsidization scheme that induces a unique, budget-balanced equilibrium achieving at least a $1/(1+\varepsilon)$ fraction of the centralized optimum for either the cumulative driver profit or rider throughput objective. \vspace{-5pt} 
\end{enumerate} \color{black}
\end{theorem}
\vspace{-3pt}

\begin{proof}[Proof Sketch.]
\color{blue} We sketch ideas behind the weak implementation result. \color{black} For each route $i$, define the per-driver profit at $\x^*$ (without cross-subsidies) as $\pi_i^* = \pi_i(x_i^*)$. To induce $\x^*$ as an equilibrium under a cross-subsidy scheme $\ttau$, for any route $i$ with $x_i^* > 0$, its transfer-adjusted per-driver profits are at least that of any other route. Thus, we construct $\ttau$ that satisfies this equilibrium condition with equality for all routes, i.e., $\pi_i^* + \tau_i = \pi_j^* + \tau_j$ for all $i, j$. Combining these relations with the budget-balance condition $\sum_{i \in [n]} \tau_i x_i^* = 0$ yields a system of $n$ linearly independent equations with $n$ unknowns (see Equation~\eqref{eq:lin-eq-cross-subsidies}), corresponding to the entries of $\ttau$. Solving this linear system yields an expression for $\ttau$, under which per-driver profits are $\Tilde{\pi}^{Eq}$.
\end{proof}
\vspace{-3pt}

For a full proof, see Appendix~\ref{apdx:pf-thm-cross-subsidy}. Theorem~\ref{thm:cross-subsidies-optimal} shows that cross-subsidies can \color{blue} essentially \color{black} eliminate the PoA inefficiencies in Section~\ref{sec:poa-bounds} by aligning driver incentives with any target driver allocation. In this sense, cross-subsidies resemble marginal-cost pricing in congestion games, which restores efficiency under selfish routing~\citep{roughgarden2005selfish}; however, unlike marginal-cost pricing, our scheme is budget-balanced \color{blue} at the induced equilibrium. Crucially, these efficiency gains need not depend on equilibrium selection: when the minibus's cost advantage ($S_i$) over the outside option (absent queuing and schedule delays) and per-rider profits ($p_i$) are strictly positive on all routes, cross-subsidies uniquely induce any desired driver allocation as an equilibrium. Even when some routes have no minibus cost advantage ($S_i=0$), cross-subsidies can induce a unique equilibrium achieving either cumulative driver profit or rider throughput arbitrarily close to its centralized optimum, provided only that per-rider profits are strictly positive (i.e., fares exceed per-rider operating costs) on all routes, a natural condition for economically viable informal transit operations.

\color{black}

The constructive proof of Theorem~\ref{thm:cross-subsidies-optimal} implies a natural algorithm to implement any desired driver allocation as an equilibrium. First, a central planner computes a target allocation $\x^*$ that maximizes a chosen objective. It then computes the per-driver profits $\pi_i^* = \pi_i(x_i^*) = \frac{\Lambda_i^M(x_i^*) p_i}{x_i^*}$ using the demand function in Equation~\eqref{eq:rider-demand-function} for each route $i$ \color{blue} with $x_i^*>0$ (with per-driver profits on unused routes defined by their limiting values $\pi_i(0)=\lim_{x_i\to 0^+}\pi_i(x_i)$). \color{black} Finally, it sets the transfer vector $\ttau$ that induces $\x^*$ as an equilibrium by solving the linear system in Theorem~\ref{thm:cross-subsidies-optimal}'s proof (see Equation~\eqref{eq:lin-eq-cross-subsidies}). \color{blue} Thus, the computational feasibility of inducing any target allocation as an equilibrium via cross-subsidies is limited only by the planner's ability to compute $\x^*$. \color{black} As a corollary, cumulative driver profit and rider throughput maximizing allocations can be implemented in polynomial time as equilibria via \color{blue} cross-subsidy schemes that are budget-balanced at the target allocations. \color{black} For additional properties of budget-balanced cross-subsidy schemes derived in Theorem~\ref{thm:cross-subsidies-optimal}, see Appendix~\ref{apdx:generality-thm-cross-subsidies}.

\vspace{-6.5pt}
\section{Mechanism II: Fare Optimization} \label{sec:price-optimization}
\vspace{-3pt}

\color{blue}
The cross-subsidy results in the previous section provide a benchmark for the efficiency gains achievable through full alignment of driver incentives. 
\color{black}
Despite their efficacy, cross-subsidies can be difficult to implement in informal transit systems, where monitoring drivers' route choices and revenues is challenging. We thus study a more directly enforceable policy lever \color{blue} widely used in informal and privatized transit systems: route-level fare regulation. \color{black} Unlike the fixed-fare model studied thus far, we now allow a central authority, such as a transit agency or driver association, to regulate fares. Such a setting reflects the institutional structure of many informal and privatized transit markets, where fares are regulated rather than freely set by individual drivers. For example, jeepney fares in the Philippines are approved by the Land Transportation Franchising and Regulatory Board~\citep{pna_2026_ltfrb_fare_hikes}, trotro and shared-taxi fares in Ghana are periodically adjusted through transport unions~\citep{ghanaweb_2025_nungua_gprtu_fares}, and minibus-taxi associations in South Africa retain substantial authority over fare adjustments~\citep{thepost_2026_taxi_fare_adjustments}. These examples also highlight that the fare-setting objective may differ by institution, as public regulators may prioritize rider throughput and affordability, while driver unions may place greater weight on driver profitability.

Although cross-subsidies and fare adjustments are both monetary levers, they differ in a fundamental way. Cross-subsidies directly affect driver payoffs while leaving rider payoffs and the rider demand function unchanged, with riders affected only indirectly through the induced reallocation of drivers across routes. Fare adjustments, in contrast, directly affect both driver margins and rider costs, thereby influencing equilibrium driver allocations and realized demand. This distinction between the policies shapes the efficacy of fare optimization in mitigating inefficiencies from selfish driver behavior. Specifically, we show that, unlike cross-subsidies, fare optimization has \emph{asymmetric power} across objectives. Comparing the best fare-optimized centralized outcome to the best fare-optimized equilibrium outcome, formalized through a fare-optimized PoA notion in Section~\ref{subsec:fare-optimization-setup}, we show that optimized fares can recover the fare-optimal centralized benchmark for rider throughput arbitrarily closely (Section~\ref{subsec:fare-welfare}), but cannot improve the tight factor-two PoA bound for cumulative driver profit (Section~\ref{subsec:fare-profit}). In addition, in Section~\ref{subsec:fare-computation}, we develop polynomial-time algorithms to compute the fare-optimized centralized and equilibrium benchmarks under both objectives despite the bi-level and non-convex structure of the joint fare-setting and driver allocation problems.

\vspace{-5pt}

\subsection{Fare Optimization Setup} \label{subsec:fare-optimization-setup}
\vspace{-2pt}

We extend the informal transit model in Section~\ref{sec:model} to the setting where the route-level fare vector $\Bar{\p}$ is chosen by the central planner to optimize either cumulative driver profit or rider throughput. 

\textbf{Price Feasibility Set and Notation:} For each route $i \in [n]$, the fare $\bar p_i$ is chosen from the interval $\mathcal P_i := \big[ \frac{c_i}{F},\; c_i^O-\eta_T l_i\big],$ which we assume is non-empty. The lower bound ensures non-negative per-rider profits, while the upper bound ensures that, absent queuing and schedule delays, the minibus is weakly preferred to the outside option. Let $\mathcal P := \prod_{i=1}^n \mathcal P_i$ denote the set of feasible fare vectors.

Given a fare $\bar p_i$, the per-rider profit on a route can be defined as $p_i(\bar p_i) := \bar p_i - \frac{c_i}{F}$ and the cost advantage of the minibus, without rider queuing and schedule delays, over the outside option is $S_i(\bar p_i) := c_i^O-\eta_T l_i-\bar p_i.$ Both these quantities now depend on the chosen fares, where increasing $\bar p_i$ raises the driver margin $p_i(\bar p_i)$ but lowers the rider-side minibus cost advantage $S_i(\bar p_i)$.

\textbf{Fare-Dependent Minibus Rider Demand:} Since fares are no longer exogenously fixed, we express the demand function in Proposition~\ref{prop:rider-demand-func-queuing} as a function of both the driver allocation and fare: 
{\setlength{\abovedisplayskip}{0pt}
\setlength{\belowdisplayskip}{0pt}
\setlength{\jot}{1pt}
\begin{align}\label{eq:rider-demand-fare-dependent}
    \Lambda_i^M \!(x_i, \bar p_i) = 
    \begin{cases}
        \frac{F}{2l_i} \left( t_2 \! - \! t_1 \! + \! S_i(\bar p_i) \left( \frac{\eta_E + \eta_L}{\eta_E \eta_L} \right) \right) x_i - \Big(\left(\frac{F}{2l_i}\right)^2 \! \! S_i(\bar p_i) \frac{\eta_E + \eta_L}{\eta_E \eta_L} \frac{t_2 - t_1}{\Lambda_i} \Big) x_i^2, & \!\!\! \text{if } x_i \in [0, \Tilde{k}_i^*(\bar p_i)) \\[-4pt]
        \Lambda_i, & \!\!\! \text{if } x_i \geq \Tilde{k}_i^*(\bar p_i)
    \end{cases}
\end{align}}
which is obtained by substituting $S_i=S_i(\bar p_i)$ in Equation~\eqref{eq:rider-demand-function}. Note here that the threshold $\tilde k_i^*(\bar p_i)$ also is fare-dependent, where $\tilde k_i^*(\bar p_i) := \min\left\{ \frac{2l_i\Lambda_i}{F(t_2-t_1)}, \frac{2l_i\Lambda_i\eta_E\eta_L}{F S_i(\bar p_i)(\eta_E+\eta_L)} \right\}$. Given a fare-dependent rider demand function, define the per-driver profit function on a route $i$ as $\pi_i(x_i, \bar p_i) := \frac{p_i(\bar p_i)\Lambda_i^M(x_i, \bar p_i)}{x_i}$. Then, for any fare vector $\bar \p \in \mathcal P$, let $\Omega^{Eq}(\bar{\mathbf p})\subseteq\Omega$ denote the set of induced equilibrium driver allocations, defined analogously to Definition~\ref{def:eq-def}. That is, a driver allocation $\x \in \Omega^{Eq}(\bar{\mathbf p})$ is an equilibrium if every route $i$ with $x_i>0$ yields weakly higher per-driver profit than any other route.

\textbf{Fare-Optimized Benchmarks and PoA:} We now define the fare-optimized cumulative driver profit and rider throughput benchmarks and the associated PoA notions. Throughout this section, an instance $I$ consists of the model primitives in Definition~\ref{def:informal_transit_instance}, except that the route-level fare vector $\bar{\p}$ is not fixed and is instead chosen from $\mathcal P$. For any instance $I$ with a driver allocation $\x$ and fare vector $\bar{\p}\in\mathcal P$, the cumulative driver profit is $P_I(\x,\bar{\p}) := \sum_{i=1}^n p_i(\bar p_i)\Lambda_i^M(x_i,\bar p_i)$ and rider throughput is $R_I(\mathbf x,\bar{\mathbf p}) := \sum_{i=1}^n \Lambda_i^M(x_i,\bar p_i)$. The corresponding fare-optimized centralized benchmarks are:
{\setlength{\abovedisplayskip}{2pt}
\setlength{\belowdisplayskip}{2pt}
\setlength{\jot}{1pt}
\begin{align*}
    P_I^{\mathrm{fare,opt}} := \sup_{\bar{\mathbf p}\in\mathcal P} \sup_{\mathbf x\in\Omega} P_I(\mathbf x,\bar{\mathbf p}), \qquad 
    R_I^{\mathrm{fare,opt}} := \sup_{\bar{\mathbf p}\in\mathcal P} \sup_{\mathbf x\in\Omega} R_I(\mathbf x,\bar{\mathbf p}).
\end{align*}}
Similarly, the fare-optimized equilibrium cumulative driver profit and rider throughput are:
{\setlength{\abovedisplayskip}{2pt}
\setlength{\belowdisplayskip}{2pt}
\setlength{\jot}{1pt}
\begin{align*}
    P_I^{\mathrm{fare,eq}} := \sup_{\bar{\mathbf p}\in\mathcal P} \inf_{\mathbf x\in\Omega^{\mathrm{Eq}}(\bar{\mathbf p})} P_I(\mathbf x,\bar{\mathbf p}), \qquad R_I^{\mathrm{fare,eq}} := \sup_{\bar{\mathbf p}\in\mathcal P} \inf_{\mathbf x\in\Omega^{\mathrm{Eq}}(\bar{\mathbf p})}
    R_I(\mathbf x,\bar{\mathbf p}).
\end{align*}}
Here, for any fixed fare vector, the infimum evaluates the worst equilibrium outcome induced by those fares. The outer supremum over fare vectors captures a central planner's ability to choose fares to optimize a desired objective, a common feature of informal transit systems, and Section~\ref{subsec:fare-computation} shows that the corresponding fares under both objectives can be computed efficiently. We then define the fare-optimized PoAs by comparing the centralized and equilibrium benchmarks:
{\setlength{\abovedisplayskip}{2pt}
\setlength{\belowdisplayskip}{2pt}
\setlength{\jot}{1pt}
\begin{align*}
    \mathrm{P\text{-}PoA}^{\mathrm{fare}} := \sup_{I\in\mathcal I} \frac{P_I^{\mathrm{fare,opt}}}{P_I^{\mathrm{fare,eq}}},
    \qquad
    \mathrm{R\text{-}PoA}^{\mathrm{fare}} := \sup_{I\in\mathcal I} \frac{R_I^{\mathrm{fare,opt}}}{R_I^{\mathrm{fare,eq}}}.
\end{align*}}
These PoA notions are the fare-optimization analogues of the PoA notions in Definition~\ref{def:poa}, with the difference that the planner may choose fares to optimize either objective under both benchmarks. Note also that the fare-optimized PoAs are bounded above by their fixed-fare counterparts, i.e., $\mathrm{P\text{-}PoA}^{\mathrm{fare}} \leq \mathrm{P\text{-}PoA}$ and $\mathrm{R\text{-}PoA}^{\mathrm{fare}} \leq \mathrm{R\text{-}PoA}$, as the planner can recover the fixed-fare PoAs by choosing the centralized optimum fare vector. 

\vspace{-3pt}
\subsection{Fare-Optimized PoA for Rider Throughput} \label{subsec:fare-welfare}
\vspace{-2pt}
This section analyzes the fare-optimized PoA for rider throughput. We show that fare optimization can essentially eliminate rider throughput losses from selfish driver behavior by inducing equilibrium outcomes arbitrarily close to the centralized benchmark (Theorem 4). 
This guarantee is constructive, and the fares needed to realize it can be computed efficiently (see Appendix~\ref{apdx:price-opt-welfare}).
\vspace{-2pt}
\begin{theorem}[Fare-optimized Rider Throughput PoA]
\label{thm:fare-welfare-poa}
Suppose $\frac{c_i}{F} < c_i^O-\eta_T l_i$ for all routes $i\in[n]$. Then, for every instance $I$ and every $\varepsilon>0$, there exists a fare vector $\Bar \p^\varepsilon\in\mathcal P$ such that the induced driver equilibrium $\x^{\mathrm{Eq}}$ is unique and $R_I(\x^{\mathrm{Eq}}, \bar{\p}^{\varepsilon}) \geq \frac{1}{1+\varepsilon} R_I^{\mathrm{fare,opt}}.$ Thus, $\mathrm{R\text{-}PoA}^{\mathrm{fare}} \!= \! 1.$
\end{theorem}
\vspace{-2pt}

The proof of this result relies on the fact that the minibus rider demand is weakly decreasing in the fare, so the fare-optimized rider throughput centralized benchmark is achieved at the minimum feasible fare vector $\bar{\p}^{\min} := \left(\frac{c_1}{F},\ldots,\frac{c_n}{F}\right).$ While minimum fares maximize rider throughput, they also eliminate driver margins and leave drivers indifferent across routes. Then, the key to proving Theorem~\ref{thm:fare-welfare-poa} lies in showing that vanishingly small, route-specific fare perturbations can resolve this driver indifference and steer drivers toward a rider-throughput maximizing allocation while preserving nearly all of the centralized rider throughput. For a complete proof, see Appendix~\ref{apdx:pf-fare-welfare-poa}.

\begin{remark}[Driver Reservation Wages] \label{rem:reservation-wages}
The construction in Theorem~\ref{thm:fare-welfare-poa} uses vanishingly small fares close to their minimum feasible levels, driving per-driver profits arbitrarily close to zero. This is useful for characterizing the fare-optimized PoA, but may be less suitable in practice when drivers require a positive reservation wage $W>0$. Accordingly, Appendix~\ref{apdx:reservation-wages-extension-poa} extends Theorem~\ref{thm:fare-welfare-poa} to a setting in which both the centralized and equilibrium fare-optimized rider throughput benchmarks must satisfy the reservation wage constraint \(\pi_i(x_i,\bar p_i)\geq W\) on every used route. This extension is technically more involved because, when $W>0$, the throughput-maximizing fare for a fixed allocation is no longer simply the minimum feasible fare; the planner must choose the lowest fare satisfying the reservation wage constraint, which depends on the allocation $x_i$. \color{blue}
Appendix~\ref{apdx:reservation-wages-extension-poa} shows that, despite this tighter coupling between fares and allocations, the throughput-optimal allocation can still be induced as an equilibrium. Moreover, if, on every route used in the throughput-optimal allocation, the lowest fare satisfying the reservation wage constraint lies strictly below the route's upper feasible fare $c_i^O-\eta_Tl_i$, then the induced equilibrium is unique and the PoA guarantee of one is preserved. This condition is natural for informal transit, where drivers' reservation wages can be met while the minibus service retains a strict cost advantage over the rider's outside option.


\color{black}
\end{remark}

\vspace{-4pt}

\subsection{Fare-Optimized PoA for Cumulative Driver Profit} \label{subsec:fare-profit}
\vspace{-3pt}

In contrast to the rider throughput result in Section~\ref{subsec:fare-welfare}, we show that fare optimization cannot eliminate equilibrium inefficiency for the cumulative profit objective. Theorem~\ref{thm:fare-opt-profit-poa} shows that even though the planner can optimize fares, the worst-case fare-optimized profit PoA remains exactly two, matching the corresponding bound in the setting with fixed fares (see Theorem~\ref{thm:poa-driver-profit}).
\vspace{-3pt}
\begin{theorem}[Fare-optimized Cumulative Driver Profit PoA] \label{thm:fare-opt-profit-poa}
For every instance $I$, $P_I^{\mathrm{fare,opt}} \! \! \le \! 2 P_I^{\mathrm{fare,eq}},$ and for all $\varepsilon \!> \!0$ there exists an instance $I_\varepsilon$ with $\frac{P_{I_\varepsilon}^{\mathrm{fare,opt}}}{P_{I_\varepsilon}^{\mathrm{fare,eq}}} \! > \! 2 - \varepsilon.$ Thus, $\mathrm{P\text{-}PoA}^{\mathrm{fare}} = 2.$
\end{theorem}
\vspace{-3pt}

For a proof, see Appendix~\ref{apdx:eq:profit-poa-fare}. The contrast between the fare-optimized PoA guarantees for rider throughput and cumulative driver profit stems from a key monotonicity property: rider throughput is weakly decreasing in fares for a fixed driver allocation, whereas cumulative driver profit need not be, since higher fares increase driver margins but may reduce rider participation. Thus, fare regulation can mitigate throughput inefficiencies, but cannot alone overcome the worst-case profit inefficiency from selfish driver behavior. In this sense, fare optimization is strictly less powerful than cross-subsidization for profit maximization. That said, fare optimization can be paired with cross-subsidies to recover the centralized cumulative driver profit benchmark.



\vspace{-5pt}
\subsection{Computing Fare-Optimized Outcomes} \label{subsec:fare-computation}
\vspace{-2pt}

While the fare-optimized PoAs in Sections~\ref{subsec:fare-welfare} and~\ref{subsec:fare-profit} characterize the power of fares in mitigating inefficiencies from selfish driver behavior, their operational relevance depends on efficiently computing the fare-optimized outcomes. Unlike the fixed fare setting, the centralized benchmarks require jointly optimizing fares and driver allocations, and the equilibrium benchmarks have a bi-level structure, since the planner chooses fares while driver allocations are determined by the induced equilibrium. Despite the difficulty of solving such joint pricing-allocation and bi-level problems, we exploit our model's structure to show that the centralized and equilibrium benchmarks introduced in Section~\ref{subsec:fare-optimization-setup} can be computed efficiently for both rider throughput and cumulative driver profit objectives. We focus here on computing the fare-optimized centralized and equilibrium benchmarks for cumulative driver profit; the corresponding results for rider throughput are in Appendix~\ref{apdx:price-opt-welfare}. \color{blue} Throughout this section, we state runtime bounds in terms of the number of routes $n$ and the approximation accuracy $\varepsilon$, treating the remaining numerical parameters of an instance $I$ in Definition~\ref{def:informal_transit_instance} as bounded constants and suppress their dependence in the reported runtimes. 
\color{black}

\textbf{Fare-Optimized Equilibrium Profit:} We first show that the fare-optimized equilibrium profit benchmark $P_I^{\mathrm{fare,eq}}$ can be computed efficiently for any feasible instance $I$. 
For each route $i$, define the largest per-driver profit route $i$ can generate with $x_i \in [0, D]$ drivers when its fare is optimized as $q_i(x_i):=\max_{\bar p_i\in \mathcal P_i} \pi_i(x_i,\bar p_i)$. Moreover, for each candidate common per-driver profit level $\rho \geq 0$, define $u_i(\rho):=\max\{x_i\in[0,D]: q_i(x_i)\geq \rho\},$ with $u_i(\rho)=0$ when $\rho>q_i(0)$. Thus, $u_i(\rho)$ is the largest mass of drivers that route $i$ can support while still achieving per-driver profit at least $\rho$, after optimizing its fare. We now show that for any $\varepsilon>0$, $P_I^{\mathrm{fare,eq}}$ can be computed in time $O(n \log^2(\frac{1}{\varepsilon}))$ by reducing the problem into a one-dimensional binary search over $\rho$.

\vspace{-5pt}
\begin{theorem}
\label{thm:equilibrium-profit-computation}
Fix an instance $I$ and define $\rho_I^* := \sup\left\{\rho\geq 0:\sum_{i=1}^n u_i(\rho)\geq D\right\}.$ Then, $P_I^{\mathrm{fare,eq}} = D\rho_I^*$ can be computed to additive accuracy $\varepsilon>0$ in time $O\!\left(n \log^2\frac{1}{\varepsilon}\right)$.
\end{theorem}
\vspace{-5pt}

\textbf{Fare-Optimized Centralized Profit:} We show that the fare-optimized centralized profit benchmark $P_I^{\mathrm{fare,opt}}$ can be computed efficiently for any feasible instance $I$. The optimal fare for each route admits a semi-closed form as a function of the allocation, as the objective is quadratic in the fare for any fixed allocation. Hence, the profit of each route is determined completely by the number of drivers assigned to it. With centralized control, the problem reduces to optimally allocating drivers across routes. We approximately solve this problem by constructing an $\varepsilon$-net for the feasible driver allocation set, running dynamic programming on this $\varepsilon$-net, and bounding suboptimality using Lipschitzness of the objective. Crucially, route profits are independent of fares and allocations on other routes, so the state space of the dynamic program is linear in the number of routes. 
\vspace{-4pt}
\begin{theorem}[Polynomial Time Computability of Fare-Optimized Centralized Profit]
\label{thm:system-opt-profit-approx}
For every $\varepsilon>0$, there is an algorithm running in $O(\frac{n^3}{\varepsilon^2})$ time that returns a feasible allocation $\hat \x$ and route fares $\hat p_i\in\arg\max_{\bar p_i\in \mathcal P_i}P_i(\hat x_i, \bar p_i) := p_i(\bar p_i) \Lambda_i^M(\hat x_i, \bar p_i)$, such that $\sum_{i=1}^n P_i(\hat x_i,\hat p_i)\ge P_I^{\mathrm{fare,opt}}-\varepsilon.$ 
\end{theorem}
\vspace{-4pt}
For complete proofs of Theorems~\ref{thm:equilibrium-profit-computation} and~\ref{thm:system-opt-profit-approx}, see Appendices~\ref{apdx:eq-profit-computation} and~\ref{apdx:system-opt-profit-computation}, respectively.


\vspace{-7pt}
\section{Numerical Experiments} \label{sec:experiments}
\vspace{-3pt}

This section presents numerical experiments \color{blue} calibrated to \color{black} an informal transit system in Nalasopara, India. 
\color{blue} 
Our results under the existing route-level fares show that, although the inefficiencies from selfish driver behavior do not reach the worst-case levels characterized by our PoA analysis, the cumulative driver profit and rider throughput ratios in the calibrated model remain meaningfully above one, indicating potentially substantial efficiency losses. \color{black} We further examine a setting in which the central planner can adjust route-level fares, showing that fare optimization is a practically relevant policy for improving outcomes in informal transit systems, but its efficacy depends on how fare adjustments are paired with driver earning objectives or constraints. We describe the methodology to calibrate our model parameters from Nalasopara's data in Section~\ref{subsec:setup-calibration} and present our results in Section~\ref{subsec:numerical-results}. 

\vspace{-5pt}
\subsection{Experimental Setup and Model Calibration} \label{subsec:setup-calibration}
\vspace{-2pt}
In Nalasopara, shared auto-rickshaws serve riders on $n = 18$ routes (see Appendix~\ref{apdx:route-selection}) connecting residential neighborhoods to a railway station. 
\color{blue}
As discussed in Section~\ref{subsec:model-assumptions}, this setting reflects several key features of our framework in Section~\ref{sec:model}, including a fixed set of routes, standardized fares, and service provision that responds to rider demand~\citep{kandathil2026nalasopara}. These characteristics make Nalasopara a natural setting for examining decentralized, profit-maximizing driver behavior and its implications for the efficiency of informal transit systems.
\color{black}
For each route $i$, we obtained trip times ($l_i$) and fares ($\Bar{p}_i$) from a local NGO, and estimated the per-trip operating costs ($c_i$) using expenditure statistics, which indicate that auto expenditures constitute roughly 80\% of drivers' daily earnings\color{blue}, consistent with evidence of substantial operating expenses among informal transit operators in India~\citep{arora2016integrating,KUMAR2016102}. \color{black} Accordingly, we set $c_i = 0.8 F \Bar{p}_i$, yielding a per-rider profit $p_i = \Bar{p}_i - \frac{c_i}{F} = 0.2 \Bar{p}_i$, where $F = 4$ is the fixed capacity of the auto-rickshaws. 

We study this system during Nalasopara's evening peak from 5 PM to midnight, when demand is highest. Unlike typical evening peaks between 3-8 PM, Nalasopara exhibits a late-night demand surge, driven by long rail commutes as riders return home from employment hubs and substantial night-time commercial activity near the station. During this evening peak, we estimate rider demand $\Lambda_i$ for each route $i$ using data from the most recent Indian Census (see Appendix~\ref{appendix:demand_calibration} for demand calibration details). \color{blue} These estimates imply a demand of around 90,000 riders across the eighteen routes, consistent in order of magnitude with reported ridership in Nalasopara~\citep{kandathil2026nalasopara}. We also assess the sensitivity of our results to variations in rider demand in Appendix~\ref{apdx:demand-sensitivity}.


Finally, we calibrate the rider cost function parameters. We set the monetary value of both in-vehicle and waiting time to $\eta_T=\eta_W=\text{Rs. }2.5$ per minute, equivalent to an hourly wage of \text{Rs. }150, reflecting average worker earnings in India in 2026~\citep{eri_india_salary_2026}. Consistent with our normalization $\eta_W=1$ in Section~\ref{subsec:rider-payoffs}, fares are expressed in time-equivalent units, and we set the normalized earliness and lateness penalties to $\eta_E=0.61$ and $\eta_L=2.4$, respectively, following~\citet{Small1982}. We then calibrate the outside option cost $c_i^O$ as the monetary equivalent of time to walk from the origin to the destination stop on route $i$, using a walking speed of $1.3$m/s. To keep this cost consistent with plausible fare levels, we cap it at the cost of taking the minibus, excluding waiting and schedule delays, when the route fare is $25\%$ above its current value. This reflects the fact that riders are unlikely to use the minibus once fares become sufficiently high and will switch to an outside option (e.g., forego the trip). While we use this specification as our baseline, additional numerical checks with alternate specifications of $c_i^O$ yield similar results. 


\vspace{-5pt}
\subsection{Results} \label{subsec:numerical-results}
\vspace{-2pt}

\color{blue}
We first examine inefficiencies from selfish driver route choices under existing fares in Nalasopara, using the calibrated parameters from Section~\ref{subsec:setup-calibration}. We measure these inefficiencies through driver profit and rider throughput ratios, each defined as the corresponding centralized optimum divided by its value attained at equilibrium. These comparisons enable us to quantify the potential efficiency gains from coordinating driver allocations under operating conditions informed by the Nalasopara data, illustrating the quantitative significance of the incentive effects captured by our model.
\color{black}
We then turn to fare optimization, examining how centrally chosen fares affect the tradeoff between rider throughput and cumulative driver profit when drivers continue to choose routes selfishly and studying how reservation wage requirements affect equilibrium outcomes under optimized fares.

\textbf{Profit and Throughput Ratios:} Figure~\ref{fig:welfare-profit-ratios} depicts the variation in the profit ratio, rider throughput ratio, and equilibrium per-driver profits as the number of drivers $D$ is varied from 100 to 2,000, while holding route-level fares fixed at their current values in Nalasopara's system. We find that the profit ratio ranges from 1.1-1.3 and the rider throughput ratio ranges from 1.1-1.4 for most values of $D \leq 1300$, \color{blue} indicating substantial potential inefficiencies across a broad range of driver supply levels. Although the number of drivers is not directly observed from the available data, a range of 750--1250 drivers yields equilibrium per-driver profits on the order of Rs.~250 ($\approx \$2.63$) in Figure~\ref{fig:welfare-profit-ratios} (right), broadly consistent with reported driver profits in Nalasopara. \color{black} Once the number of drivers exceeds 1300, both ratios converge to one as profit-maximizing and rider throughput maximizing allocations already absorb all rider demand; hence, additional drivers increase the corresponding objective under the equilibrium allocation without improving the optimal benchmarks.

While the profit and rider throughput ratios \color{blue} in our experiments \color{black} do not reach the worst-case levels \color{blue} characterized in our PoA analysis, with a bound of \(2\) for profit and an instance-wise bound of $1 + \frac{p_{\max}}{p_{\min}} = 4$ for rider throughput in our calibrated instance, where $\frac{p_{\max}}{p_{\min}} = 3$ in our data, they nevertheless indicate substantial potential inefficiencies arising from decentralized driver behavior when $D \leq 1300$. Specifically, the equilibrium allocations yield 9–23\% lower cumulative driver profits and serve 9–29\% fewer riders relative to their corresponding centralized optima. To put these losses in perspective, in a system with approximately 90,000 potential riders, the estimated rider throughput losses could translate into roughly 8,000–26,000 fewer riders served if the centralized allocation served the full demand. Similarly, at equilibrium earnings of Rs. 250 (\$2.63) per driver, profit ratios of 1.1–1.3 would correspond to foregone earnings of approximately Rs. 25–75 per driver. These results highlight the significant potential of cross-subsidization to mitigate inefficiencies arising from selfish driver behavior in such high-stakes informal transit settings (see Section~\ref{sec:cross-subsidies}). 


To assess the robustness of these findings to variations in demand, we repeat the analysis by scaling the calibrated route-level demands $\Lambda_i$ by up to $\pm20\%$. The resulting sensitivity bands, reported in Appendix~\ref{apdx:demand-sensitivity}, show that the profit and rider throughput ratios remain broadly consistent with the baseline findings in Figure~\ref{fig:welfare-profit-ratios}. 
We also examine whether these inefficiencies persist under a broader measure of social welfare, which combines cumulative driver profit with riders' gains relative to their outside options and incorporates waiting and schedule delay costs (see Appendix~\ref{apdx:social-welfare-experiments}).


\color{black} 

Finally, the right of Figure~\ref{fig:welfare-profit-ratios} shows that equilibrium profit per driver decreases monotonically with driver supply, consistent with Corollary~\ref{cor:monotonicity-per-driver-profits} and reflecting greater competition for riders.

\vspace{-10pt}

\begin{figure*}[tbh!]
  \centering \hspace{-35pt}
  \begin{subfigure}[b]{0.27\columnwidth}
      \include{Fig/msom/PoABounds/profit_ratio}
  \end{subfigure} \hspace{15pt}
  \begin{subfigure}[b]{0.27\columnwidth}
      \include{Fig/msom/PoABounds/welfare_ratio}
  \end{subfigure} \hspace{20pt}
  \begin{subfigure}[b]{0.27\columnwidth}
      \include{Fig/msom/PoABounds/profit_per_driver}
  \end{subfigure} 
     \vspace{-40pt}
    \caption{{\small \sf Depiction of the profit ratio (left), rider throughput ratio (center), and equilibrium profit per driver (right) as the number of drivers in the system is varied.}} 
    \label{fig:welfare-profit-ratios}
\end{figure*}

\vspace{-10pt}

\textbf{Fare optimization:}
While the results in Figure~\ref{fig:welfare-profit-ratios} fix route-level fares to their current values in Nalasopara, we now consider a setting in which a central planner can optimize fares. To keep the fares practically meaningful, we restrict the route fares to be no more than $25\%$ above their current values. In this setting, the fare-optimized profit ratio closely mirrors the fixed-fare profit ratio in Figure~\ref{fig:welfare-profit-ratios}, while the fare-optimized rider throughput ratio is approximately one when the reservation wage $W=0$, and fare optimization \color{blue} can induce an equilibrium attaining the centralized benchmark when $W>0$ \color{black} (see Section~\ref{subsec:fare-welfare}). Thus, we focus on the tradeoff between the two objectives under fare optimization. Specifically, we optimize a convex combination of the two objectives, \color{blue} normalized by their respective centralized optima to place both objectives on a common, dimensionless scale, \color{black} indexed by $\alpha \in [0,1]$, where $\alpha=0$ corresponds to \color{blue} normalized \color{black} rider throughput and $\alpha=1$ to \color{blue} normalized \color{black} cumulative driver profit, and compare the resulting centralized and equilibrium outcomes on both objectives. For any $\alpha$, note that the corresponding objective reduces to a special case of the cumulative driver profit objective, \color{blue} given by $\sum_i \tilde p_i(\bar p_i) \Lambda_i^M(x_i, \bar p_i)$, where $\tilde p_i(\bar p_i) = \frac{1-\alpha}{R_{\mathcal I}^{\mathrm{fare,opt}}(W)} + \frac{\alpha p_i(\bar p_i)}{P_{\mathcal I}^{\mathrm{fare,opt}}(W)}$, with $R_{\mathcal I}^{\mathrm{fare,opt}}(W)$ and $P_{\mathcal I}^{\mathrm{fare,opt}}(W)$ denoting the respective fare-optimized centralized benchmarks under the reservation wage constraint $W$. \color{black} In these experiments, we set the reservation wage to $W=\mathrm{Rs.}\ 250$ for all $\alpha$, consistent with average daily driver profits in Nalasopara, and fix the total driver supply at $D=1000$.

Figure~\ref{fig:alpha_profit_welfare} plots the profit and throughput ratios as functions of $\alpha$ for the centralized and equilibrium fare-optimized problems. These ratios are measured relative to the respective fare-optimized centralized benchmarks at $\alpha=1$ (profit) and $\alpha=0$ (rider throughput); hence, ratios near one indicate that the corresponding single-objective benchmark is nearly attained. Our results show that moving from $\alpha=0$ to even a small positive weight on profit substantially improves cumulative driver profit with little loss in throughput. This sharp improvement arises as the pure rider throughput objective sets fares as low as possible subject to the reservation wage constraint, maximizing ridership but leaving little margin for driver profits. A small profit weight raises fares enough to improve driver margins substantially without deterring many riders; thus, the profit ratio falls sharply while the rider throughput ratio remains near one. \color{blue} As $\alpha$ increases further, the equilibrium solutions exhibit a sharper tradeoff, with additional profit gains coming at a noticeable cost to rider throughput. In contrast, the centralized benchmark maintains near-optimal rider throughput even when greater weight is placed on cumulative driver profit, indicating that the resulting throughput loss under centralization is small. Moreover, as in the fixed-fare setting, the equilibrium solutions can remain meaningfully below their centralized counterparts for the same objective weights, highlighting the inefficiencies from selfish driver route choices even when fares are optimized. \color{black} Note that, at the extremes, the rider throughput ratio of the equilibrium solution is one at $\alpha=0$, while the profit ratio at $\alpha=1$ is strictly larger than one but below two, consistent with our guarantees in Section~\ref{sec:price-optimization}.

While including driver profits in the fare-optimization objective can improve cumulative earnings with little loss in throughput, in practice, driver earnings are often protected \color{blue} through minimum earning requirements.  
\color{black} 
We thus examine how the fare-optimized rider throughput equilibrium changes with the reservation wage, allowing for driver non-participation (see Appendix~\ref{apdx:welf-max-res-wage}). Figure~\ref{fig:reservation_wage} shows that throughput is largely insensitive to the reservation wage up to $W=\mathrm{Rs.}\ 450$, indicating that moderate earning guarantees can be met with little ridership loss. Over this range, all drivers participate and earn the reservation wage, so cumulative driver profit rises linearly with $W$. Beyond this point, however, \color{blue} reservation wages become counterproductive: \color{black} both throughput and cumulative profit fall as the system cannot allocate all drivers under the higher reservation wage requirement. 

Overall, these results suggest that fare optimization is \color{blue} a practically relevant policy \color{black} for improving informal transit outcomes, but its efficacy depends on how it is paired with driver earning objectives or constraints. When fares are optimized purely for rider throughput, \color{blue} an induced equilibrium can approach or attain the centralized throughput benchmark, \color{black} but driver profits remain low. Adding even a small weight on profit substantially improves driver earnings with little loss in rider throughput, although selfish route choices \color{blue} can still create a nontrivial gap relative to the centralized outcome. \color{black} Similarly, moderate reservation wage requirements can be accommodated without meaningfully reducing ridership, but \color{blue} excessively high requirements \color{black} may reduce demand and driver participation. Thus, fare adjustments can balance \color{blue} rider throughput \color{black} and driver earnings, but are most effective when earning targets are moderate and inefficiencies from selfish driver choices are accounted for.


\begin{figure}[t]
    \centering
    \input{Fig/msom/revision/figure4_normalized.tex}
    \vspace{-10pt}
    \caption{\small \sf 
    Fare-optimized profit and rider throughput ratios for centralized and equilibrium allocations as the weight $\alpha$ on cumulative driver profit varies in the normalized mixed objective. Here, $\alpha=0$ corresponds to rider throughput maximization and $\alpha=1$ to profit maximization, while $\alpha=0.5$ assigns equal weight to both normalized objectives. The ratios are measured relative to the respective centralized single-objective optima, $\mathbf{x}_1^*$ for profit and $\mathbf{x}_0^*$ for rider throughput. The reservation wage is fixed at $W=250$ throughout.
    }
    \label{fig:alpha_profit_welfare}
\end{figure}

\vspace{-10pt}

\begin{figure*}[tbh!]
  \centering \hspace{-25pt}
  \begin{subfigure}[b]{0.4\columnwidth}
      \include{Fig/msom/priceOpt/reservation_wage_profit}
  \end{subfigure} \hspace{35pt}
  \begin{subfigure}[b]{0.4\columnwidth}
      \include{Fig/msom/priceOpt/reservation_wage_welfare}
  \end{subfigure} 
     \vspace{-40pt}
    \caption{{\small \sf Cumulative driver profit and rider throughput achieved under the fare-optimized equilibrium rider throughput maximizing solution as the reservation wage $W$ is varied.
    }} 
    \label{fig:reservation_wage}
\end{figure*}

\vspace{-22pt}


\vspace{-5pt}

\section{Conclusion and Future Work} \label{sec:conclusion}
\vspace{-3.5pt}

This work developed a framework to analyze incentives in informal and privatized transit systems and proposed mechanisms to mitigate inefficiencies from decentralized driver behavior. \color{blue} Our PoA analysis showed \color{black} that selfish driver behavior can result in substantial losses in cumulative driver profit and rider throughput; however, these losses can be mitigated through targeted interventions. While cross-subsidies can mitigate inefficiencies across both objectives, fare optimization, which is easier to implement and enforce, has asymmetric power: it can recover the centralized rider throughput benchmark arbitrarily well, but cannot improve the tight PoA bound of two for cumulative driver profit. We further reinforced these findings through numerical experiments. 
Overall, our work highlights the promise and limits of incentive design in informal transit systems, showing how practical monetary levers can align operator incentives with system objectives while identifying when losses from decentralized behavior remain unavoidable. In doing so, we introduce a new, relatively understudied, application domain to the operations community.

Several directions merit future study. First, our model can be extended to incorporate heterogeneous driver route preferences, while the rider-side framework could accommodate extensions of Vickrey's bottleneck model, such as heterogeneous values of time or non-uniform desired arrival time distributions. Future work could also examine whether the complexity bounds of our algorithms for fare-optimized equilibrium and centralized benchmarks are tight, whether fare optimization can mitigate efficiency losses under other objectives, \color{blue} and the gains achievable through other regulatory interventions. \color{black} Finally, while we focused on fully privatized informal transit systems, the framework naturally extends to settings in which public and informal transit coexist. A fuller treatment of this extension, omitted here for space and to preserve the paper's focus, would study how public transit service and incentives for decentralized operators should be jointly designed. 

\vspace{-5pt}

\section*{Acknowledgments}
This work was supported by the Data Science Institute Postdoctoral Fellowship at Columbia University. We also thank Dr. Jacqueline Klopp for several insightful discussions and the team at the Society for the Promotion of Area Resource Centers (SPARC), including Alex Mohan Kandathil, Anreerudha Paul, Shreya Baoni, and Kshitija Akre for their support in providing the data on Nalasopara's informal transit system. This data collection was enabled through the Mumbai Living Lab as part of PRISM (Partnership for Research on Informal and Shared Mobility).





\bibliographystyle{unsrtnat} 
\bibliography{main} 

\clearpage






  



%
%
%

\appendix

\section{Additional Related Work} \label{apdx:related-work}

This section surveys additional works in addition to those covered in Section~\ref{sec:literature} that are related to this work.

Our work is also related to the design of mechanisms to mitigate inefficiencies in decentralized mobility systems. The cross-subsidy mechanisms we study use route-specific tolls or subsidies to shape service patterns of informal or privatized transit drivers, akin to targeted driver subsidies and driver financial incentives in ride-sharing~\citep{ZHU2021540,wang2023efficiency,philipp2026a}. In contrast to these works, we show that a \emph{budget-balanced} cross-subsidy mechanism, requiring no net government expenditure, can align operator incentives in the informal transit setting we study, echoing revenue-neutral congestion pricing with revenue-refunding schemes developed for traffic routing contexts~\citep{jalota-tcns}.

Beyond cross-subsidies, we also study fare optimization as a mechanism to mitigate inefficiencies from selfish driver behavior, relating both to congestion pricing in selfish routing and to platform pricing in two-sided mobility markets. As with congestion pricing, route-level fares can improve decentralized outcomes~\citep{roughgarden2005selfish}; however, unlike classical tolls, which only price agents generating congestion externalities, fares in our setting influence both rider costs and driver margins. This two-sided feature of fare adjustments connects our work to pricing models in ride-hailing platforms~\citep{banerjee-johari-2015,banerjee-johari-2016,bimpikis2019spatial,doi:10.1287/msom.2020.0960}. Our setting differs, however, in that drivers choose which routes to serve rather than being matched, dispatched, or repositioned by a platform. Thus, in our setting, fare optimization does not simply balance rider demand and driver supply; it also reshapes the route-level earnings that determine decentralized service patterns.

Our work also contributes to the literature on public transit design, including network and line planning~\citep{schobel2012line,kreindler2023optimal}, multi-modal integration~\citep{luo-banerjee-2025,samitha-real-time-2022}, and micro-transit design~\citep{van2023marta,joc-jacquillat-paratransit}. Whereas these works study centrally coordinated transit systems, we adopt a complementary approach by studying equilibria in decentralized transit systems with a predefined menu of routes. Finally, our work relates to studies on the economic impacts of public transportation infrastructure, which have mainly focused on formal transit systems, including subways~\citep{subways-gu} and BRT lines~\citep{tsivanidis2022evaluating}. In contrast, our analysis centers on informal and privatized transit, highlighting how their incentives shape outcomes in settings where formal public transit is limited.

\section{Proofs}

\subsection{Proof of Proposition~\ref{prop:rider-demand-func-queuing}} \label{apdx:pf-rider-demand-gonzalez}

Fix a route $i$ with $S_i \geq 0$ and $x_i$ minibus drivers operating on that route, resulting in a service rate of $\mu_i(x_i) = \frac{x_i F}{2 l_i}$. Then, defining the total time horizon over which the drivers service the route as $T_i(x_i)$, the total mass of riders served by the minibus is given by: $\Lambda_i^M(x_i) = \min \{ \mu_i(x_i), \lambda_i \} T_i(x_i)$. Note that $T_i(x_i) \geq \Delta = t_2 - t_1$, as riders may be willing to arrive earlier or later than their desired arrival times at the destination under the cost functions in Equation~\eqref{eq:cost-minibus-queuing}. Thus, the key to proving our result is to establish a relation for the total time $T_i(x_i)$ over which riders are serviced on a route $i$, given an allocation of $x_i$ minibus drivers on that route.

Note that if the system is over-supplied with $\mu_i(x_i) \geq \lambda_i$, since $S_i \geq 0$, there will be no rider queuing or waiting delays and all riders can arrive at their destinations at their desired times by the minibus; hence, $T_i(x_i) = \Delta$ in this regime. 

Thus, consider the under-supplied regime when $\mu_i(x_i) < \lambda_i$. In this regime, to derive a relation for the total service time, we first characterize riders' equilibrium mode and arrival time decisions on a given route $i$ with $x_i$ drivers operating on that route. We do so by adapting the associated equilibrium characterization in Vickrey’s bottleneck model with an outside option from~\cite{GONZALES20121519} to our informal transit setting. In particular, defining the threshold $\Bar{S}(x_i) = \frac{\Lambda_i \eta_E \eta_L}{\mu_i(x_i) (\eta_E + \eta_L)}$ given an allocation of $x_i$ drivers to route $i$, we characterize riders' equilibrium mode and arrival time decisions in two cases: (i) $S_i \geq \Bar{S}(x_i)$, and (ii) $S_i \in [0, \Bar{S}(x_i))$.

\textbf{Case (i):} In this setting, the outside option is never cost-effective relative to traveling by the minibus, even for the user experiencing the highest travel cost with a wait time of $\Bar{S}(x_i)$. The resulting equilibrium waiting time profile of riders reduces to the classical Vickrey bottleneck model without an outside option and is depicted on the left of Figure~\ref{fig:eq-waiting-time}, with riders being serviced over the interval $[t_A', t_B']$. In this case, since all riders are serviced by the minibus at equilibrium, the total service time $T_i(x_i) = t_B' - t_A' = \frac{\Lambda_i}{\mu_i(x_i)}$.

\textbf{Case (ii):} In this regime, only a fraction of the total rider demand is served by the minibus. Adapting Proposition 1 from~\cite{GONZALES20121519} to our informal transit setting, we obtain the following equilibrium characterization in the regime when $S_i \in [0, \Bar{S}(x_i))$.

\begin{proposition}[Rider Equilibrium~\citep{GONZALES20121519}] \label{prop:gonzales-eq-characterization}
Suppose $x_i$ drivers operate on route $i$ with $\mu_i(x_i) < \lambda_i$, where users choose between the minibus and an outside option where the cost difference between the outside option and the minibus without rider queuing delays $S_i \in [0, \Bar{S}(x_i))$. Then, assuming riders are serviced in the order of their desired destination arrival times, there exists a rider equilibrium satisfying (see right of Figure~\ref{fig:eq-waiting-time}):
\begin{itemize}
    \item The number of minibus riders that arrive at their destination earlier than their desired time is given by $N_E = \frac{\mu_i(x_i) S_i}{\eta_E}$ and they travel at the beginning of the rush between periods $[t_A, t_B]$.
    \item The number of minibus riders that arrive at their destination later than their desired time is given by $N_L = \frac{\mu_i(x_i) S_i}{\eta_L}$ and they travel at the end of the rush between periods $[t_C, t_D]$.
    \item The number of users that arrive exactly on time by the minibus and those that use the outside option are strictly decreasing functions of the cost difference $S_i$ and they travel in the middle of the rush between $[t_B, t_C]$. 
\end{itemize} 
\end{proposition}

\vspace{-10pt}

\begin{figure}[tbh!]
      \centering
      \includegraphics[width=100mm]{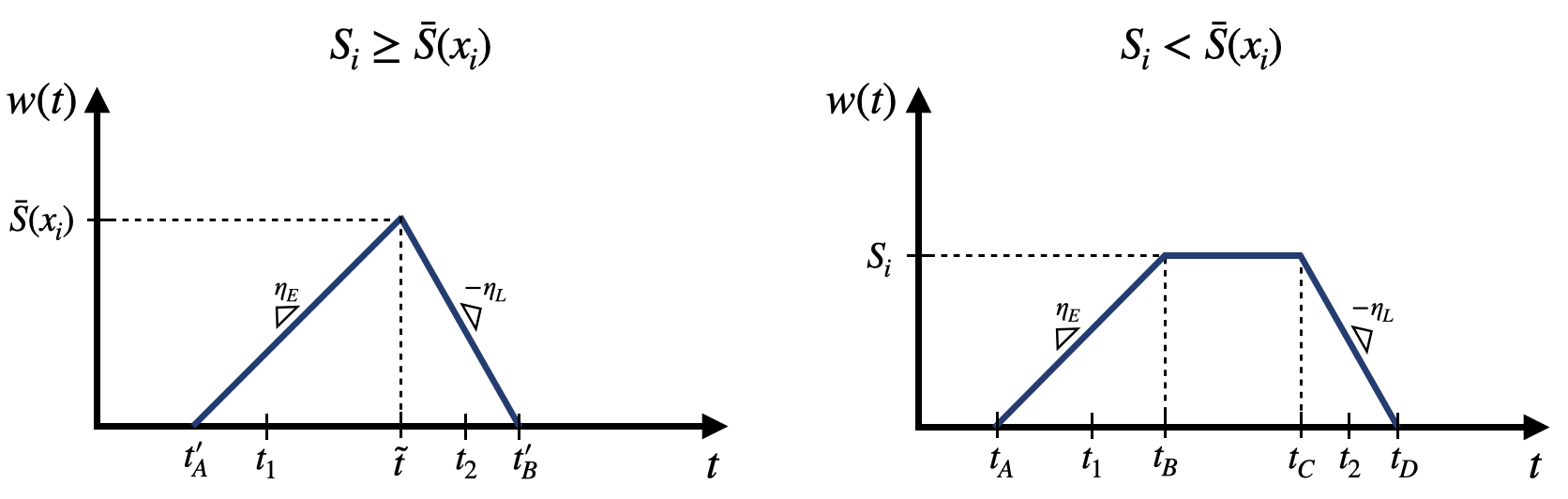}
      \vspace{-10pt}
      \caption{\small \sf Equilibrium rider waiting time profiles when the cost difference between the outside option and that of using the minibus without queuing delays satisfies $S_i \geq \Bar{S}(x_i)$ (left) and $S_i < \Bar{S}(x_i)$ (right).
      }
      \label{fig:eq-waiting-time} 
   \end{figure} 

\vspace{-10pt}
 
While~\cite{GONZALES20121519} characterize the above equilibrium for a broad class of desired destination arrival time distributions for riders, they do not provide a closed-form characterization of the total service time $T_i(x_i) = t_D - t_A$. In our setting, by focusing on uniformly distributed desired destination arrival time distributions, which is consistent with the classical bottleneck model of~\cite{vickrey1969congestion} while also significantly more general than the single departure time formulations most commonly studied in the literature, we obtain a closed-form characterization of the service times in our setting. 

\begin{corollary}[Equilibrium Service Time] \label{cor:service-time}
Suppose $x_i$ minibus drivers operate on route $i$ with $\mu_i(x_i) < \lambda_i$. Then, in the regime when the cost difference between the outside option and the minibus without rider queuing delays is non-negative, i.e., $S_i \geq 0$, the total service time $T_i(x_i) = \Delta + \min \{S_i, \Bar{S}(x_i) \} \left( \frac{\eta_E + \eta_L}{\eta_E \eta_L} \right) \left( 1 - \frac{\Delta x_i F}{2 l_i \Lambda_i} \right)$.
\end{corollary}

\begin{proof}
First consider the regime when $S_i \geq \Bar{S}(x_i)$. Then, substituting the expression for $\Bar{S}(x_i)$ in the expression for the total service time in the statement of the corollary and simplifying, we obtain that $T_i(x_i) = \frac{\Lambda_i}{\mu_i(x_i)}$, consistent with the expression in the analysis of case (i) above.

Next, consider the regime when $S_i < \Bar{S}(x_i)$. In this case, from the right of Figure~\ref{fig:eq-waiting-time}, note that the service time $T_i(x_i) = t_D - t_A = (t_D - t_C) + (t_C - t_B) + (t_B - t_A)$. Next, from Proposition~\ref{prop:gonzales-eq-characterization}, since all riders are served by the minibus in the period $[t_A, t_B]$, it follows that $N_E = \frac{\mu_i(x_i) S_i}{\eta_E} = \mu_i(x_i)(t_B - t_A)$, which implies $t_B - t_A = \frac{S_i}{\eta_E}$. Analogously, it follows that $t_D - t_C = \frac{S_i}{\eta_L}$. 

Finally, to derive the relation for $t_C - t_B$, note that the total fraction of the demand $\Lambda_i$ that is served early or late during the periods $[t_A, t_B]$ and $[t_C, t_D]$ is $\frac{S_i}{\Bar{S}(x_i)}$ by the linearity of the equilibrium waiting time function in Figure~\ref{fig:eq-waiting-time}. Hence, it follows that $t_C - t_B = \left(1-\frac{S_i}{\Bar{S}(x_i)} \right) \Delta$. Then:
{\setlength{\abovedisplayskip}{0pt}
\setlength{\belowdisplayskip}{2pt}
\setlength{\jot}{1pt}
\begin{align*}
    T_i(x_i) &= t_D - t_A = (t_D - t_C) + (t_C - t_B) + (t_B - t_A) = \frac{S_i}{\eta_L} + \frac{S_i}{\eta_E} + (t_2-t_1)(1 - \frac{S_i}{\Bar{S}(x_i)}), \\
    &= \Delta + S_i \left( \frac{1}{\eta_L} + \frac{1}{\eta_E} - \frac{\Delta}{\Bar{S}(x_i)}\right) = \Delta + S_i \left( \frac{\eta_E + \eta_L}{\eta_E \eta_L} - \frac{t_2-t_1}{\frac{\Lambda_i \eta_E \eta_L}{\mu_i(x_i) (\eta_E + \eta_L)}} \right) \\
    &= \Delta + S_i \left( \frac{\eta_E + \eta_L}{\eta_E \eta_L} \right) \left( 1 - \frac{\Delta x_i F}{2 l_i \Lambda_i} \right),
\end{align*}}
where we substitute $\mu_i(x_i) = \frac{x_i F}{2 l_i}$ to obtain the final equality. This establishes our total service time relation. 
\end{proof}

Having established the relations for the total service time in both the under and over-supplied regimes, we now evaluate the rider demand $\Lambda_i^M(x_i) =  \min \{ \mu_i(x_i), \lambda_i \} T_i(x_i)$ by substituting the derived expressions for the service time. 

First consider the over-supplied regime when $\mu_i(x_i) \geq \lambda_i$, corresponding to the setting when $x_i \geq \frac{2 l_i \lambda_i}{F}$. In this case, recalling that $T_i(x_i) = \Delta$, we have $\Lambda_i^M(x_i) =  \min \{ \mu_i(x_i), \lambda_i \} T_i(x_i) = \lambda_i \Delta = \Lambda_i$.

Next, in the under-supplied case, suppose that it further holds that $S_i \geq \Bar{S}(x_i)$, corresponding to the setting when $x_i \geq \frac{2 l_i \Lambda_i \eta_E \eta_L}{F S_i (\eta_E + \eta_L)}$. Then, substituting the expression for $\Bar{S}(x_i)$ in the total service time expression in Corollary~\ref{cor:service-time}, it follows that $T_i(x_i) = \frac{\Lambda_i}{\mu_i(x_i)}$. Hence, in this regime, it again follows that $\Lambda_i^M(x_i) = \Lambda_i$. Together, the above relations imply that for any driver allocation $x_i \geq \min \{ \frac{2 l_i \lambda_i}{F}, \frac{2 l_i \Lambda_i \eta_E \eta_L}{F S_i (\eta_E + \eta_L)} \} = \Tilde{k}_i^*$, all the rider demand is served by the minibus, i.e., $\Lambda_i^M(x_i) = \Lambda_i$ for all $x_i \geq \Tilde{k}_i^*$.

Finally, consider the regime when $x_i < \Tilde{k}_i^*$. In this case, we leverage the relation for the total service time derived in Corollary~\ref{cor:service-time} to derive the minibus rider demand as a function of the driver allocation on that route. In this regime:
{\setlength{\abovedisplayskip}{2pt}
\setlength{\belowdisplayskip}{2pt}
\setlength{\jot}{1pt}
\begin{align*}
    \Lambda_i^M(x_i) &= \min \{ \mu_i(x_i), \lambda_i \} T_i(x_i) \stackrel{(a)}{=} \mu_i(x_i) T_i(x_i) \stackrel{(b)}{=} \frac{x_i F}{2l_i} \left(\Delta + S_i \left( \frac{\eta_E + \eta_L}{\eta_E \eta_L} \right) \left( 1 - \frac{\Delta x_i F}{2 l_i \Lambda_i} \right) \right), \\
    &= \frac{F}{2 l_i} \left( \Delta + S_i \left( \frac{\eta_E + \eta_L}{\eta_E \eta_L} \right) \right) x_i - \left(\left(\frac{F}{2l_i}\right)^2 S_i \frac{\eta_E + \eta_L}{\eta_E \eta_L} \frac{\Delta}{\Lambda_i} \right) x_i^2
\end{align*}}
where (a) follows as $\mu_i(x_i) < \lambda_i$ in the regime when $x_i < \Tilde{k}_i^*$ and (b) follows by substituting the relation for $\mu_i(x_i)$ and that for $T_i(x_i)$ derived in Corollary~\ref{cor:service-time}. This establishes our desired relation for the minibus rider demand in the regime when $x_i < \Tilde{k}_i^*$, which proves our result.

\vspace{-3pt}
\subsection{Proof of Corollary~\ref{cor:monotonicity-rider-demands}} \label{apdx:pf-cor1}

If $S_i=0$, the demand function in Equation~\eqref{eq:rider-demand-function} is piecewise linear and the result follows immediately. Hence, suppose $S_i>0$. To establish continuity of the rider demand function, it suffices to show that $\Lambda_i^M(\Tilde{k}_i^*) = \Lambda_i$. This follows directly by substituting the two candidate expressions for $\Tilde{k}_i^*$, i.e., $\frac{2 l_i \lambda_i}{F}$ and $\frac{2 l_i \Lambda_i \eta_E \eta_L}{F S_i (\eta_E + \eta_L)}$, into the quadratic expression for $\Lambda_i^M(\cdot)$. 

Next, to establish monotonicity, it suffices to show that $\Lambda_i^M(\cdot)$ is monotonically non-decreasing in the range $x_i \in [0, \Tilde{k}_i^*]$. To see this, we consider the regime when the derivative of the minibus rider demand function is non-negative:
{\setlength{\abovedisplayskip}{2pt}
\setlength{\belowdisplayskip}{2pt}
\setlength{\jot}{1pt}
\begin{align*}
    \frac{d \Lambda_i^M(x_i)}{d x_i} = \frac{F}{2l_i} \left( t_2 - t_1 + S_i \left( \frac{\eta_E + \eta_L}{\eta_E \eta_L} \right) \right) - 2 \left(\left(\frac{F}{2l_i}\right)^2 S_i \frac{\eta_E + \eta_L}{\eta_E \eta_L} \frac{t_2-t_1}{\Lambda_i} \right) x_i \geq 0.
\end{align*}}
Rearranging the above inequality and simplifying, we obtain the following condition on the driver allocation $x_i$ for the non-negativity of the derivative of the minibus demand function $\Lambda_i^M(\cdot)$: $x_i \leq \frac{1}{2} \left( \frac{2 l_i \lambda_i}{F} + \frac{2 l_i \Lambda_i \eta_E \eta_L}{F S_i (\eta_E + \eta_L)} \right).$ Since the expression on the right hand side of this inequality is at least $\Tilde{k}_i^* = \min \left\{ \frac{2 l_i \lambda_i}{F}, \frac{2 l_i \Lambda_i \eta_E \eta_L}{F S_i (\eta_E + \eta_L)} \right\}$, it follows that $\Lambda_i^M(\cdot)$ is monotonically non-decreasing in the range $x_i \in [0, \Tilde{k}_i^*]$.

\subsection{Proof of Corollary~\ref{cor:monotonicity-per-driver-profits}} \label{apdx:pf-cor2}

In the regime when $x_i > \Tilde{k}_i^*$, the per-driver profit $\pi_i(x_i) = \frac{\Lambda_i p_i}{x_i}$. On the other hand, when $x_i \leq \Tilde{k}_i^*$, substituting the relation for the demand function $\Lambda_i^M(\cdot)$, the per-driver profit is given by:
{\setlength{\abovedisplayskip}{2pt}
\setlength{\belowdisplayskip}{2pt}
\setlength{\jot}{1pt}
\begin{align} \label{eq:per-driver-profit-relation}
    \pi_i(x_i) = \frac{\Lambda_i^M(x_i) p_i}{x_i} = p_i \frac{F}{2l_i} \left( t_2 - t_1 + S_i \left( \frac{\eta_E + \eta_L}{\eta_E \eta_L} \right) \right) - p_i \left(\left(\frac{F}{2l_i}\right)^2 S_i \frac{\eta_E + \eta_L}{\eta_E \eta_L} \frac{t_2-t_1}{\Lambda_i} \right) x_i.
\end{align}}
From these relations and the continuity of the rider demand function $\Lambda_i^M(\cdot)$, the per-driver profit function is continuous and non-increasing in the driver allocation $x_i$.

\subsection{Proof of Theorem~\ref{thm:poa-demand-served}} \label{apdx:pf-poa-rider-welfare}

Fix an informal transit instance $I$ with associated rider demand functions belonging to the class $\mathcal{L}_V$, given by Equation~\eqref{eq:rider-demand-function} for all routes $i \in [n]$. Since our proof applies to all feasible instances $I$, we drop the dependency on $I$ in our notation for the remainder of this proof.

For this instance, let $\x^*$ be the rider throughput maximizing allocation, i.e., $\x^* \in \argmax_{\x \in \Omega} R(\x) = \sum_{i = 1}^n \Lambda_i^M(x_i)$, and let $\x^{Eq}$ be any equilibrium driver allocation. Moreover, as in the proof of Theorem~\ref{thm:poa-driver-profit}, it suffices to restrict attention to the set of routes $Q$ for which the difference $S_i \geq 0$.

Then, defining the two sets: $L_1 = \{ i \in Q: x_i^{Eq} \geq x_i^* \}$ and $L_2 = \{ i \in Q: x_i^{Eq} < x_i^* \}$, we have:
{\setlength{\abovedisplayskip}{2pt}
\setlength{\belowdisplayskip}{2pt}
\setlength{\jot}{1pt}
\begin{align*}
    R(\x^*) &= \sum_{i \in Q} \Lambda_i^M(x_i^*) = \sum_{i \in L_1} \Lambda_i^M(x_i^*) + \sum_{i \in L_2} \Lambda_i^M(x_i^*) \stackrel{(a)}{\leq} \sum_{i \in L_1} \Lambda_i^M(x_i^{Eq}) + \sum_{i \in L_2} \Lambda_i^M(x_i^*), \\
    &\stackrel{(b)}{\leq} \sum_{i \in L_1} \Lambda_i^M(x_i^{Eq}) + \frac{1}{p_{\min}} \sum_{i \in L_2} p_i \Lambda_i^M(x_i^*) \stackrel{(c)}{=} \sum_{i \in L_1} \Lambda_i^M(x_i^{Eq}) + \frac{1}{p_{\min}} \sum_{i \in L_2} \pi_i^* x_i^*, \\
    &\stackrel{(d)}{\leq} \sum_{i \in L_1} \Lambda_i^M(x_i^{Eq}) + \frac{1}{p_{\min}} \pi^{Eq} \sum_{i \in L_2} x_i^* \stackrel{(e)}{\leq} \sum_{i \in L_1} \Lambda_i^M(x_i^{Eq}) + \frac{1}{p_{\min}} \pi^{Eq} \sum_{i \in Q} x_i^{Eq}, \\
    &= \sum_{i \in L_1} \Lambda_i^M(x_i^{Eq}) + \frac{1}{p_{\min}} \sum_{i \in Q} p_i \Lambda_i^M(x_i^{Eq}) \leq \sum_{i \in L_1} \Lambda_i^M(x_i^{Eq}) + \frac{p_{\max}}{p_{\min}} \sum_{i \in Q} \Lambda_i^M(x_i^{Eq}), \\ 
    &\stackrel{(f)}{\leq} \left( 1 + \frac{p_{\max}}{p_{\min}} \right) R(\x^{Eq}),
\end{align*}}
where (a) follows by the monotonicity of the minibus rider demand in the driver allocation as established in Corollary~\ref{cor:monotonicity-rider-demands}, (b) follows as $p_{\min} \leq p_i$ for all routes $i$, and (c) follows by the definition of $\pi_i^* = \frac{p_i \Lambda_i^M(x_i^*)}{x_i^*}$. Moreover, (d) follows as $\pi^{Eq} \geq \pi_i^*$ for all $i \in L_2$ by the monotonicity of the per-driver profit function in Corollary~\ref{cor:monotonicity-per-driver-profits}, (e) follows as $\sum_{i \in L_2} x_i^* \leq \sum_{i \in Q} x_i^* = \sum_{i \in Q} x_i^{Eq} = D$, as the same number of drivers are allocated under the optimal and equilibrium allocations, and (f) follows as $\sum_{i \in L_1} \Lambda_i^M(x_i^{Eq}) \leq \sum_{i \in Q} \Lambda_i^M(x_i^{Eq}) = R(\x^{Eq})$, as $L_1 \subseteq Q$. \color{blue} The above analysis holds for any equilibrium allocation and any fixed instance $I$, establishing the instance-wise bound $\frac{R_I(\x^*)}{R_I(\x^{Eq})} \leq 1 + \frac{p_{\max}(I)}{p_{\min}(I)}$, where \(p_{\max}(I)\) and \(p_{\min}(I)\) are the maximum and minimum per-rider profits, respectively, for instance \(I\). \color{black}


\subsection{Proof of Proposition~\ref{prop:tightness-poa-demand}} \label{apdx:pf-tightness-poa-demand}

We proceed in the vein of Proposition~\ref{prop:tightness-poa-profit}, and consider an informal transit instance with $n = 2$ routes, a normalized driver supply $D = 1$, and the cost difference $S_i = 0$ for both routes. Consequently, we obtain piece-wise linear rider demand functions for both routes, as in Equation~\eqref{eq:rider-demand-function-reduction} with $\Tilde{k}_i^* = \frac{2 l_i \lambda_i}{F}$ for $i \in \{ 1, 2\}$. Moreover, consider the setting where $\Tilde{k}_1^* + \Tilde{k}_2^* = D = 1$. In this case, the rider throughput maximizing allocation is $\x^* = (\Tilde{k}_1^*, \Tilde{k}_2^*)$, resulting in a total maximum achievable rider demand served of $R(\x^*) = \Lambda_1^M(\Tilde{k}_1^*) + \Lambda_2^M(\Tilde{k}_2^*) = \Lambda_1 + \Lambda_2$.

Next, suppose that profit per driver on the two routes satisfies $\pi_1(1) = \pi_2(0)$, i.e., the informal transit system supports an equilibrium driver allocation $\x^{Eq} = (1, 0)$. For this condition to hold, the following equalities must be satisfied:
\begin{align*}
    p_1 \Lambda_1 = \pi_1(1) = \pi_2(0) = p_2 \frac{F(t_2 - t_1)}{2 l_2} = p_2 \Lambda_2 \frac{F}{2 l_2 \lambda_2} = p_2 \Lambda_2 \frac{1}{\Tilde{k}_2^*}
\end{align*}
Using this relation, the minibus rider demand ratio between the optimal and above-defined equilibrium allocation is:
\begin{align*}
    \frac{R(\x^*)}{R(\x^{Eq})} = \frac{\Lambda_1 + \Lambda_2}{\Lambda_1} = 1 + \frac{\Lambda_2}{\Lambda_1} = 1 + \Tilde{k}_2^* \frac{p_1}{p_2}.
\end{align*}
Finally, taking the limit as $\Tilde{k}_2^* \rightarrow 1$ and $\Tilde{k}_1^* \rightarrow 0$ while satisfying $\Tilde{k}_1^* + \Tilde{k}_2^* = D = 1$, the above analysis implies that the minibus rider demand ratio approaches $1 + \frac{p_{\max}}{p_{\min}}$, thus establishing our claim.

\subsection{Proof of Theorem~\ref{thm:cross-subsidies-optimal}} \label{apdx:pf-thm-cross-subsidy}

\textbf{Weak Implementation:} For brevity of notation, we define the per-driver profit on route $i$ (without cross-subsidies) at the target allocation $\x^*$ as $\pi_i^* = \pi_i(x_i^*)$, where the rider demand function is given by Equation~\eqref{eq:rider-demand-function}. Then, to induce $\x^*$ as an equilibrium driver allocation under a cross-subsidy scheme defined by $\ttau$, the following equilibrium condition must hold: for any route $i$ with $x_i^* > 0$, the per-driver profit under cross-subsidy transfers is at least that of any other route $j$, i.e.,
{\setlength{\abovedisplayskip}{2pt}
\setlength{\belowdisplayskip}{2pt}
\setlength{\jot}{1pt}
\begin{align} \label{eq:eq-condition-pf-cross-subsidy}
    \Tilde{\pi}_i(x_i^*) \geq \Tilde{\pi}_j(x_j^*) \quad \implies \quad  \pi_i^* + \tau_i \geq \pi_j^* + \tau_j,
\end{align}}
where the inequality is met with an equality for all routes $j$ with $x_j^* > 0$.

In the remainder of this proof, we construct a transfer vector $\ttau$ that satisfies the above equilibrium condition with equality for all routes, i.e., $\pi_i^* + \tau_i = \pi_j^* + \tau_j$ for all routes $i, j$, while satisfying budget balance, i.e., $\sum_{i \in [n]} \tau_i x_i^* = 0$. Thus, we seek to satisfy the following $n$ linearly independent equations with $n$ unknowns, corresponding to the entries of the transfer vector $\ttau$: $\pi_1^* + \tau_1 = \pi_2^* + \tau_{2},$ $\pi_1^* + \tau_1 = \pi_3^* + \tau_{3},$ $\ldots$ $\pi_1^* + \tau_1 = \pi_n^* + \tau_{n},$ $\sum_{i = 1}^n \tau_i x_i^* = 0.$
Rearranging the above equations, our goal is to find a vector $\ttau$ that solves the following system:
\begin{equation} \label{eq:lin-eq-cross-subsidies}
    \begin{bmatrix}
1 & -1 & 0 & \ldots & 0 \\
1 & 0 & -1 & \ldots & 0 \\
& & & \vdots & \\
1 & 0 & 0 & \ldots & -1 \\
x_1^* & x_2^* & x_3^* & \ldots & x_n^* \\
\end{bmatrix}
\begin{bmatrix}
\tau_1 \\
\tau_2 \\
\tau_3 \\
\vdots \\
\tau_n \\
\end{bmatrix} = 
\begin{bmatrix}
\pi_2^* - \pi_1^* \\
\pi_3^* - \pi_1^* \\
\vdots \\
\pi_n^* - \pi_1^* \\
0 \\
\end{bmatrix}
\end{equation}
Note that the above matrix is full-rank and thus admits a unique solution if there is any route $j$ such that $x_j^* > 0$. Then, note from the budget-balance constraint that:
{\setlength{\abovedisplayskip}{2pt}
\setlength{\belowdisplayskip}{2pt}
\setlength{\jot}{1pt}
\begin{align*}
    0 &= \sum_{i = 1}^n \tau_i x_i^* = \sum_{i \in [n]: x_i^* > 0} \tau_i x_i^* = \tau_j x_j^* +  \sum_{i \in [n]: x_i^* > 0, i \neq j} \tau_i x_i^* \stackrel{(a)}{=} \tau_j x_j^* +  \sum_{i \in [n]: x_i^* > 0, i \neq j} (\pi_j^* + \tau_j - \pi_i^*) x_i^* \\
    &= \sum_{i \in [n]: x_i^* > 0, i \neq j} (\pi_j^* - \pi_i^*) x_i^* + \tau_j \sum_{i \in [n]: x_i^* > 0} x_i^*,
\end{align*}}
where (a) follows from the equation $\pi_i^* + \tau_i = \pi_j^* + \tau_j$ for all pairs of routes $i, j$. Rearranging the above expression, we obtain $\tau_j = \frac{\sum_{i \in [n]: x_i^* > 0, i \neq j} (\pi_i^* - \pi_j^*) x_i^*}{\sum_{i \in [n]: x_i^* > 0} x_i^*}.$ Next, for the remaining $i' \neq j$, we have: $\tau_{i'} = \pi_j^* - \pi_{i'}^* + \tau_j = \pi_j^* - \pi_{i'}^* + \frac{\sum_{i \in [n]: x_i^* > 0, i \neq j} (\pi_i^* - \pi_j^*) x_i^*}{\sum_{i \in [n]: x_i^* > 0} x_i^*}.$ The resulting per-driver profit under these tolls, which is fixed across routes is given by:
{\setlength{\abovedisplayskip}{2pt}
\setlength{\belowdisplayskip}{2pt}
\setlength{\jot}{1pt}
\begin{align*}
    \Tilde{\pi}^{Eq} = \Tilde{\pi}_j(x_j^*) = \pi_j^* + \tau_j = \pi_j^* + \frac{\sum_{i \in [n]: x_i^* > 0, i \neq j} (\pi_i^* - \pi_j^*) x_i^*}{\sum_{i \in [n]: x_i^* > 0} x_i^*} = \frac{\sum_{i \in [n]: x_i^* > 0} \pi_i^* x_i^*}{\sum_{i \in [n]: x_i^* > 0} x_i^*} = \frac{\sum_{i \in [n]} \pi_i^* x_i^*}{\sum_{i \in [n]} x_i^*},
\end{align*}}
which establishes our claim.

\color{blue}
\textbf{Strong Implementation:} 
Suppose $S_i>0$ and $p_i>0$ for all $i\in[n]$.
From Equation~\eqref{eq:per-driver-profit-relation}, the per-driver profit functions $\pi_i(\cdot)$ are strictly decreasing, and hence so are the transfer-adjusted profits $\tilde{\pi}_i(\cdot)=\pi_i(\cdot)+\tau_i$.  Thus, applying the equilibrium uniqueness argument in Appendix~\ref{apdx:pf-uniqueness} to these transfer-adjusted profit functions establishes that $\x^*$ is the unique induced equilibrium.

\textbf{Approximate Strong Implementation.}
Suppose $p_i>0$ for all $i\in[n]$. Let $J\in\{P,R\}$ denote either cumulative driver profit or rider throughput, and let $\x^*\in\arg\max_{\x\in\Omega}J(\x)$, with $J^*=J(\x^*)$.

Let $Z=\{i\in[n]:S_i=0\}$. From Equation~\eqref{eq:rider-demand-function-reduction}, for each $i\in Z$, we have
\[
\Lambda_i^M(x_i)=\min\{a_i x_i,\Lambda_i\},
\quad
a_i=\frac{F(t_2-t_1)}{2l_i},
\quad
k_i^*=\frac{\Lambda_i}{a_i}.
\]
Thus, the contribution of route $i\in Z$ to the objective is $J_i(x_i)=v_i\min\{x_i,k_i^*\},$ where $v_i=a_i$ for rider throughput and $v_i=p_i a_i$ for cumulative driver profit.

We first observe that there exists an optimal
allocation $\bar{\x}$ with at most one route
$r\in Z$ satisfying $0<\bar{x}_r<k_r^*$.
Indeed, if two such routes have different
slopes $v_i$, shifting a sufficiently small mass of drivers from the route with the smaller slope to the one with the larger slope would improve the objective, contradicting optimality.
If their slopes are equal, drivers can be
reallocated between them without changing
the objective until one route reaches either
zero or its saturation threshold. Repeating
this argument establishes the claim.

Now, let
$C=\{i\in Z:\bar{x}_i=k_i^*\}$.
If both an interior route $r$ and at least
one route in $C$ exist, construct a perturbed
allocation $\x^\varepsilon$ by decreasing
$\bar{x}_r$ by $\delta>0$ and increasing
$\bar{x}_i$ by $\delta_i>0$ for each
$i\in C$, where
$\sum_{i\in C}\delta_i=\delta<\bar{x}_r$.
Otherwise, set $\x^\varepsilon=\bar{\x}$.

In the perturbed case, increasing driver
allocations beyond the saturation thresholds
does not change the objective, while decreasing
the allocation on route $r$ reduces it by
$v_r\delta$. Thus, choosing $\delta$
sufficiently small ensures $J(\x^\varepsilon) =J^*-v_r\delta \geq \frac{J^*}{1+\varepsilon}.$ The same guarantee holds trivially in the
unperturbed case. By construction,
$\x^\varepsilon$ has at most one route
strictly within a flat-profit region, and
if such a route exists, no other used route
lies at the endpoint of a flat-profit region.

We now construct transfers supporting
$\x^\varepsilon$. Define $\mu^\varepsilon =\frac{\sum_{i\in[n]} \pi_i(x_i^\varepsilon)x_i^\varepsilon}{D},$ and set
\[
\tau_i^\varepsilon=
\begin{cases}
\mu^\varepsilon-\pi_i(x_i^\varepsilon),
& x_i^\varepsilon>0,\\
\mu^\varepsilon-\pi_i(0)-\gamma_i,
& x_i^\varepsilon=0,
\end{cases}
\]
where $\gamma_i>0$ is arbitrary.
As in the weak-implementation proof,
$\x^\varepsilon$ is an equilibrium and
the transfers are budget-balanced, since all used routes receive the same transfer-adjusted profit and the additional penalties apply only to unused routes.

It remains to establish uniqueness.
Suppose $\y$ is another equilibrium.
Using the same variational inequality argument
in Appendix~\ref{apdx:pf-uniqueness} and the
non-increasing property of per-driver profits,
\[
0\leq
\sum_{i\in[n]}
\left(\Tilde{\pi}_i(x_i^\varepsilon)
-\Tilde{\pi}_i(y_i)\right)
(x_i^\varepsilon-y_i)
\leq 0.
\]
Hence, on every route whose allocation differs,
$\pi_i(x_i^\varepsilon)=\pi_i(y_i)$.
Since both equilibrium inequalities above must hold with equality, and the transfer-adjusted profits at \(\x^\varepsilon\) equal \(\mu^\varepsilon\) on used routes and \(\mu^\varepsilon-\gamma_i\) on unused routes, we obtain
\[
0=\sum_{i\in[n]}
\Tilde{\pi}_i(x_i^\varepsilon)
(y_i-x_i^\varepsilon)
=-\sum_{i:x_i^\varepsilon=0}\gamma_i y_i.
\]
Thus, no unused route at $\x^\varepsilon$
can receive drivers at $\y$.

Since $p_i>0$, per-driver profits are strictly
decreasing on routes with $S_i>0$ and on
the saturated branch of routes with $S_i=0$.
Consequently, allocations on such routes
cannot differ between the two equilibria.
For routes with $S_i=0$ and
$x_i^\varepsilon=k_i^*$, allocations can
only decrease while preserving per-driver
profits.

By construction, if no route lies strictly
within its flat-profit region, no route can
increase its allocation while preserving
per-driver profits. Otherwise, exactly one
route lies strictly within its flat-profit
region, and the allocations on all other
used routes must remain unchanged.
In either case, the fixed total driver mass
$\sum_i y_i=D$ implies $\y=\x^\varepsilon$.

Thus, $\x^\varepsilon$ is the unique equilibrium
under the constructed budget-balanced transfers,
and satisfies
$J(\x^\varepsilon)\geq J^*/(1+\varepsilon)$. 

\subsection{Proof of Theorem~\ref{thm:fare-welfare-poa}} \label{apdx:pf-fare-welfare-poa}

To prove this claim, we first show the following lemma, which establishes the monotonicity of the minibus rider demand $\Lambda_i^M(x_i,\Bar{p}_i)$ in the fare for any fixed driver allocation $x_i \geq 0$.

\color{black}

\begin{lemma}[Monotonicity of Rider Demand in Route Fare]
\label{lem:fare-monotonicity}
Fix a route $i$ and a driver allocation $x_i \geq 0$. Then, the rider demand
$\Lambda_i^M(x_i,\Bar{p}_i)$ is weakly decreasing in the fare $\Bar{p}_i$ on
$\left[\frac{c_i}{F},\, c_i^O - \eta_T l_i\right]$.
\end{lemma}

\begin{proof}
Fix route $i$ and allocation $x_i \ge 0$, and let $\Delta := t_2 - t_1$. Moreover, we define $S_i(\Bar{p}_i) := c_i^O - \eta_T l_i - \Bar{p}_i,$ so $S_i(\Bar{p}_i) \ge 0$ on the feasible fare interval and $S_i(\Bar{p}_i)$ is strictly
decreasing in $\Bar{p}_i$. Then, we have the following relation for the rider demand served by the minibus:
{\setlength{\abovedisplayskip}{2pt}
\setlength{\belowdisplayskip}{2pt}
\setlength{\jot}{1pt}
\begin{align*}
\Lambda_i^M(x_i,\Bar{p}_i) =
\begin{cases}
\frac{F}{2l_i}\left(\Delta + \frac{\eta_E+\eta_L}{\eta_E\eta_L}S_i(\Bar{p}_i)\right)x_i
-
\left(\frac{F}{2l_i}\right)^2
S_i(\Bar{p}_i)\frac{\eta_E+\eta_L}{\eta_E\eta_L}\frac{\Delta}{\Lambda_i}x_i^2,
& \text{if } x_i \in [0,\Tilde{k}_i^*(\Bar{p}_i)), \\[6pt]
\Lambda_i,
& \text{if } x_i \ge \Tilde{k}_i^*(\Bar{p}_i),
\end{cases}
\end{align*}}
where $\Tilde{k}_i^*(\Bar{p}_i) = \min\left\{k_i^*,\, \frac{2l_i\Lambda_i\eta_E\eta_L}{F\,S_i(\Bar{p}_i)(\eta_E+\eta_L)} \right\}$, and $k_i^* := \frac{2l_i\lambda_i}{F} = \frac{2l_i\Lambda_i}{F\Delta}.$

For the quadratic branch of the rider demand function, it is convenient to rewrite it as:
{\setlength{\abovedisplayskip}{2pt}
\setlength{\belowdisplayskip}{2pt}
\setlength{\jot}{1pt}
\begin{align}
\Lambda_i^M(x_i,\Bar{p}_i)
&=
A_i(x_i) + B_i(x_i)S_i(\Bar{p}_i),
\qquad \text{for } x_i \in [0,\Tilde{k}_i^*(\Bar{p}_i)),
\label{eq:rider-demand-affine-in-S}
\end{align}}
where $A_i(x_i) := \frac{F\Delta}{2l_i}x_i$ and $B_i(x_i) = \frac{\eta_E+\eta_L}{\eta_E\eta_L}\frac{F}{2l_i}x_i
\left(1-\frac{F\Delta}{2l_i\Lambda_i}x_i\right).$

In the quadratic branch of the rider demand function, we have $x_i < \Tilde{k}_i^*(\Bar{p}_i) \le k_i^* = \frac{2l_i\Lambda_i}{F\Delta}$; hence, $1-\frac{F\Delta}{2l_i\Lambda_i}x_i \geq 0,$ which implies that $B_i(x_i)\ge 0$. It then follows from \eqref{eq:rider-demand-affine-in-S} and the fact that
$S_i(\Bar{p}_i)$ is decreasing in $\Bar{p}_i$ that, in the quadratic branch,
$\Lambda_i^M(x_i,\Bar{p}_i)$ is weakly decreasing in $\Bar{p}_i$.

We now establish our monotonicity result for the two regimes corresponding to $\Tilde{k}_i^*(\Bar{p}_i)$.

\noindent\emph{Case 1: $\Tilde{k}_i^*(\Bar{p}_i)=k_i^*$.}
In this case, the threshold $\Tilde{k}_i^*(\Bar{p}_i)$ is independent of the fare $\Bar{p}_i$. In this case, if $x_i<k_i^*$, then the rider demand is in the quadratic branch and, by the argument above, $\Lambda_i^M(x_i,\Bar{p}_i)$ weakly decreases with $\Bar{p}_i$. If $x_i\geq k_i^*$, then the non-quadratic branch applies and $\Lambda_i^M(x_i,\Bar{p}_i)=\Lambda_i$, which is constant in $\Bar{p}_i$.

\noindent\emph{Case 2: $\Tilde{k}_i^*(\Bar{p}_i)=\dfrac{2l_i\Lambda_i\eta_E\eta_L}{F\,S_i(\Bar{p}_i)(\eta_E+\eta_L)}$.}
In this regime, since $S_i(\Bar{p}_i)$ is decreasing in $\Bar{p}_i$, the threshold
$\Tilde{k}_i^*(\Bar{p}_i)$ is weakly increasing in $\Bar{p}_i$ (until it possibly reaches the cap
$k_i^*$). Let $\Bar{p}_i' \ge \Bar{p}_i$.

If $x_i < \Tilde{k}_i^*(\Bar{p}_i)$, then also
$x_i < \Tilde{k}_i^*(\Bar{p}_i')$, so both rider demands lie in the quadratic branch under these prices. Since $B_i(x_i)\geq 0$
and $S_i(\Bar{p}_i') \le S_i(\Bar{p}_i)$, we obtain
{\setlength{\abovedisplayskip}{2pt}
\setlength{\belowdisplayskip}{2pt}
\setlength{\jot}{1pt}
\begin{align*}
    \Lambda_i^M(x_i,\Bar{p}_i') = A_i(x_i)+B_i(x_i)S_i(\Bar{p}_i') \leq A_i(x_i)+B_i(x_i)S_i(\Bar{p}_i) = \Lambda_i^M(x_i,\Bar{p}_i).
\end{align*}}

If instead $x_i \geq \Tilde{k}_i^*(\Bar{p}_i)$, then at fare $\Bar{p}_i$, the rider demand is not in the quadratic branch; hence, $\Lambda_i^M(x_i,\Bar{p}_i)=\Lambda_i.$ For the higher fare $\Bar{p}_i'$, either $x_i \ge \Tilde{k}_i^*(\Bar{p}_i')$, in which case the rider demand is still in the non-quadratic branch and the rider demand remains $\Lambda_i$, or $x_i < \Tilde{k}_i^*(\Bar{p}_i')$, in which case the rider demand moves to the quadratic branch. In the latter case, by the definition of $\Tilde{k}_i^*(\Bar{p}_i')$ as the
threshold beyond which the full demand $\Lambda_i$ is served, it must be the case that $\Lambda_i^M(x_i,\Bar{p}_i') \leq \Lambda_i = \Lambda_i^M(x_i,\Bar{p}_i).$ Thus, in all cases, $\Lambda_i^M(x_i,\Bar{p}_i') \le \Lambda_i^M(x_i,\Bar{p}_i).$

Combining the two cases analyzed above establishes that $\Lambda_i^M(x_i,\Bar{p}_i)$ is weakly decreasing in
$\Bar{p}_i$ for every fixed $x_i$. 
\end{proof}

We now leverage this monotonicity property to complete the proof of Theorem~\ref{thm:fare-welfare-poa}.

\begin{proof}[Proof of Theorem~\ref{thm:fare-welfare-poa}.]
Fix an instance $I$ and $\varepsilon>0$. For each route $i$, define $\Bar p_i^{\min}:=\frac{c_i}{F}$ and $\Bar \p^{\min}:=(\Bar p_1^{\min},\dots,\Bar p_n^{\min}).$ By Lemma~\ref{lem:fare-monotonicity}, for every feasible allocation $\x\in\Omega$ the rider throughput $R_I(\x,\Bar p):=\sum_{i=1}^n \Lambda_i^M(x_i,\Bar p_i)$ is weakly decreasing in each coordinate $\Bar p_i$. Therefore, the fare vector that maximizes rider throughput is $\Bar \p^{\min}$. Thus, we define $R_I^{\mathrm{fare,opt}} = \max_{\x\in\Omega} R_I(\x,\Bar p^{\min})$ and let $\x^* \in \arg\max_{\x\in\Omega} R_I(\x,\Bar p^{\min})$ be any rider-throughput maximizing allocation under these prices. Note, however, that at the minimum fare vector, the per-rider profit margin on all routes is $p_i(\Bar{p}_i) = p_i(\frac{c_i}{F}) = \frac{c_i}{F} - \frac{c_i}{F} = 0$, so drivers are indifferent across routes and every feasible allocation $\x \in \Omega$ is an equilibrium. Thus, although the minimum fares maximize rider throughput, $\inf_{\x\in\Omega} R_I(\mathbf x,\bar{\p}^{\min})$ can be arbitrarily small.

To eliminate these low throughput equilibria, we introduce route-specific price perturbations of the minimum fare vector that induce $\x^*$ as the unique equilibrium driver allocation. To this end, define the sets $U:=\{i\in[n]:x_i^*>0\}$ and $Z:=[n]\setminus U$, and for each route $i$, define $M_i:=c_i^O-\eta_T l_i-\frac{c_i}{F}>0.$ Since $M_i>0$, each route admits strictly interior fares of the form
$\Bar p_i=\frac{c_i}{F}+q_i$ with $q_i\in(0,M_i)$, and for such fares one has $S_i(\Bar p_i)=M_i-q_i>0.$

We first choose a fare vector in a small neighborhood of $\Bar \p^{\min}$ under which the rider throughput at $\x^*$ remains
arbitrarily close to optimal. Since, for fixed $\x^*$, each map
$\Bar p_i \mapsto \Lambda_i^M(x_i^*,\Bar p_i)$ is continuous on $\mathcal P_I$, the function
$\Bar \p \mapsto R_I(\x^*,\Bar p)$ is continuous at $\Bar \p^{\min}$. Hence there exist $\delta_i \in (0,M_i)$ for all $i\in[n]$ such that whenever $\Bar p_i \in \left[\frac{c_i}{F},\, \frac{c_i}{F}+\delta_i\right]$ for all $i$, it follows that: $R_I(\x^*,\Bar \p)\ge \frac{1}{1+\varepsilon}R_I(\x^*,\Bar \p^{\min}) = \frac{1}{1+\varepsilon}R_I^{\mathrm{fare,opt}}.$

We now choose route-specific fare perturbations that make $\x^*$ the unique equilibrium.
For each used route $i\in U$, define $f_i(q) := \pi_i\!\left(x_i^*;\frac{c_i}{F}+q\right) = \frac{\Lambda_i^M\!\left(x_i^*,\,\frac{c_i}{F}+q\right)}{x_i^*}\, q,$ for $q\in[0,\delta_i].$ Each $f_i$ is continuous, satisfies $f_i(0)=0$, and is strictly positive for every $q>0$.
Therefore, define $\underline{\kappa} := \min_{i\in U} f_i(\delta_i) >0$ and choose any $\kappa \in (0,\underline{\kappa}).$ By the intermediate value theorem, for each $i\in U$ there exists $q_i^\varepsilon \in (0,\delta_i]$ such that $f_i(q_i^\varepsilon)=\kappa$, i.e., $\pi_i\!\left(x_i^*;\frac{c_i}{F}+q_i^\varepsilon\right)=\kappa$ for all $i \in U$.

For each unused route $j\in Z$, define $f_j(q) := \pi_j\!\left(0;\frac{c_j}{F}+q\right)$ for $q\in[0,\delta_j].$ Since $f_j$ is continuous and $f_j(0)=0$, we may choose $q_j^\varepsilon \in (0,\delta_j]$ sufficiently small so that $f_j(q_j^\varepsilon)<\kappa$ for all $j \in Z$.

Now define the fare vector $\Bar \p^\varepsilon$ route-wise by $\Bar p_i^\varepsilon := \frac{c_i}{F}+q_i^\varepsilon$ for all $i \in [n]$. 
By construction, for every used route $i\in U$, $\pi_i(x_i^*;\Bar p_i^\varepsilon)=\kappa$, while for every unused route $j\in Z$, $\pi_j(0;\Bar p_j^\varepsilon)<\kappa.$ Hence, at the allocation $\x^*$, every used route yields the same per-driver profit and every unused route yields strictly lower per-driver profit. Therefore $\x^*$ is an equilibrium allocation under $\Bar \p^\varepsilon$.

\color{blue}
To see uniqueness, note that by construction, $p_i(\bar{p}_i^\varepsilon)=q_i^\varepsilon>0$ and $S_i(\Bar p_i^\varepsilon)= M_i-q_i^\varepsilon>0$ for all $i\in[n]$. Thus, the per-driver profit functions are strictly decreasing, implying uniqueness of the equilibrium under $\Bar \p^\varepsilon$ by the argument in Appendix~\ref{apdx:pf-uniqueness}. \color{black} Hence the equilibrium induced by $\Bar \p^\varepsilon$ is unique, and in fact $\Omega^{Eq}(\bar{\p}^\varepsilon)=\{x^*\}.$


Finally, since each $q_i^\varepsilon \leq \delta_i$, the fare vector $\Bar \p^\varepsilon$ lies in the
continuity neighborhood chosen above. Therefore, $R_I(x^*,\Bar p^\varepsilon) \geq \frac{1}{1+\varepsilon}R_I^{\mathrm{fare,opt}}.$ Since $x^*$ is the unique equilibrium under $\Bar \p^\varepsilon$, we obtain
\[
R_I^{\mathrm{fare,eq}}
\ge
\min_{x\in\Omega^{Eq}(\bar{\p}^\varepsilon)} R_I(x,\Bar \p^\varepsilon)
=
R_I(x^*,\Bar \p^\varepsilon)
\ge
\frac{1}{1+\varepsilon}R_I^{\mathrm{fare,opt}}.
\]
Equivalently, $\frac{R_I^{\mathrm{fare,opt}}}{R_I^{\mathrm{fare,eq}}}\le 1+\varepsilon.$ Since the instance $I$ was arbitrary, the same bound holds uniformly over all instances, which proves the claim. 
\end{proof}

\subsection{Proof of Theorem~\ref{thm:fare-opt-profit-poa}} \label{apdx:eq:profit-poa-fare}

We first prove the upper bound. Fix an instance $I$ and a feasible fare vector $\bar{\p}\in\mathcal P$. Conditional on $\bar{\p}$, the per-rider margins $p_i(\bar p_i)$ and the minibus cost advantages $S_i(\bar p_i)$ are fixed, so the induced rider demand functions remain exactly of the form in Equation~\eqref{eq:rider-demand-function} in the fixed-fare model. Hence, the fixed-fare profit PoA guarantee applies to the induced instance: $\sup_{\x\in\Omega} P_I(\x,\bar{\p}) \leq 2 \inf_{\x\in\Omega^{\mathrm{Eq}}(\bar{\p})} P_I(\x,\bar{\p}).$ Taking the supremum over $\bar{\p}\in\mathcal P$ gives $P_I^{\mathrm{fare,opt}} \leq 2P_I^{\mathrm{fare,eq}}.$

It remains to show tightness. To do so, for any $\varepsilon>0$, we construct a two-route instance parametrized by $\rho\in(0, \min\{\varepsilon, 1\})$. For this instance, let $D=1,$ $k_1^*=\rho,$ $k_2^*=1-\rho,$ $\beta := \frac{\eta_E+\eta_L}{\eta_E\eta_L (t_2 - t_1)}$, and choose a constant $K>0$ satisfying $\beta K < 1.$ Next, choose $\Lambda_1>0$ arbitrarily, and set $\Lambda_2 := k_2^* \Lambda_1.$ Finally, choose route parameters so that $k_i^*=\frac{2l_i\lambda_i}{F}$ and $c_i^O-\eta_T l_i-\frac{c_i}{F}=K$ for $i \in \{ 1, 2\}$. Thus, on both routes, the feasible per-rider profit margin interval is $p_i(\bar p_i)\in[0,K].$ For this instance, we establish the desired lower bound in two steps. 

\emph{Step 1: Lower bounding the centralized fare-optimized profit.} Set the maximum fare on both routes, so $p_i(\bar p_i)=K$ and $S_i(\bar p_i)=0$. Then, the rider demand function reduces to the capacity-constrained form in Equation~\eqref{eq:rider-demand-function-reduction} and $\x^*=(k_1^*,k_2^*)$ serves the full demand on both routes. Then: $P_I^{\mathrm{fare,opt}} \geq K(\Lambda_1+\Lambda_2) = K\Lambda_1(1+k_2^*) = K\Lambda_1(2-\rho).$

\emph{Step 2: Upper bounding the best equilibrium profit over all fare vectors.} Fix any feasible $\bar \p$, and let $\x \in \Omega^{\mathrm{Eq}}(\bar{\p})$ be any equilibrium under $\bar \p$. We show that $P_I(\x,\bar \p)\le K\Lambda_1$ by considering two cases.

\emph{Case 1: $x_2 = 0$.} Here, $x_1=1$ and at most $\Lambda_1$ riders are served; hence, $P_I(x,\bar p) \leq p_1(\bar p_1)\Lambda_1 \le K\Lambda_1.$ 

\emph{Case 2: $x_2>0$.} In this case, let $\pi^{\mathrm{eq}}$ be the common per-driver profit on the used routes. Since $D = 1$, $P_I(\x,\bar{\p})=\pi^{\mathrm{eq}}.$ Since the per-driver profit on a route is weakly decreasing in its own driver allocation, $\pi^{\mathrm{eq}} = \pi_2(x_2,\bar p_2) \leq \pi_2(0,\bar p_2).$ Now, let $p_2:=p_2(\bar p_2)\in[0,K]$. Since $S_2(\bar p_2)=K-p_2,$ the minibus rider demand function in Equation~\eqref{eq:rider-demand-function} implies that, for $x_2<\tilde k_2^*(\bar p_2)$,
{\setlength{\abovedisplayskip}{2pt}
\setlength{\belowdisplayskip}{2pt}
\setlength{\jot}{1pt}
\begin{align*}
    \Lambda_2^M(x_2,\bar p_2) = \frac{\Lambda_2}{k_2^*}\Bigl(1+\beta(K-p_2)\Bigr)x_2 - \frac{\Lambda_2}{(k_2^*)^2}\beta(K-p_2)x_2^2.
\end{align*}}
Therefore, $\pi_2(0,\bar p_2) = \frac{\Lambda_2}{k_2^*}\,p_2\Bigl(1+\beta(K-p_2)\Bigr).$ Then, consider the scalar function $h(p):=p\Bigl(1+\beta(K-p)\Bigr)$ for $p\in[0,K].$ Its derivative is $h'(p)=1+\beta K-2\beta p \ge 1-\beta K > 0,$ where the strict inequality uses $\beta K<1$. Thus $h$ is strictly increasing on $[0,K]$, so $\pi_2(0,\bar p_2) \le \frac{\Lambda_2}{k_2^*}\,K = \frac{k_2^*\Lambda_1}{k_2^*}\,K = K\Lambda_1.$ Consequently, $P_I(\x,\bar \p)=\pi^{\mathrm{eq}}\le K\Lambda_1.$ Since both cases yield the same bound, for every feasible fare vector and every equilibrium under that fare vector, $P_I(\x,\bar \p)\le K\Lambda_1.$ Hence, $P_I^{\mathrm{fare,eq}} \le K\Lambda_1.$

Combining the lower bound on $P_I^{\mathrm{fare,opt}}$ with the upper bound on $P_I^{\mathrm{fare,eq}}$,
we obtain
{\setlength{\abovedisplayskip}{2pt}
\setlength{\belowdisplayskip}{2pt}
\setlength{\jot}{1pt}
\begin{align*}
    \frac{P_I^{\mathrm{fare,opt}}}{P_I^{\mathrm{fare,eq}}} \ge \frac{K\Lambda_1(2-\rho)}{K\Lambda_1} = 2-\rho > 2-\varepsilon.
\end{align*}}
Since $\varepsilon>0$ was arbitrary, the factor two upper bound is tight, establishing that $\mathrm{P\text{-}PoA}^{\mathrm{fare}} = 2.$ 

\subsection{Proof of Theorem~\ref{thm:equilibrium-profit-computation}} \label{apdx:eq-profit-computation}

We begin by introducing some notation. For an instance $I$, let $M_i:=c_i^O-\eta_Tl_i-\frac{c_i}{F},$ $\Delta:=t_2-t_1,$ $a_i:=\frac{F\Delta}{2l_i},$ $b_i:=\frac{F(\eta_E+\eta_L)}{2l_i\eta_E\eta_L},$ $k_i^*:=\frac{\Lambda_i}{a_i}=\frac{2l_i\Lambda_i}{F\Delta},$ and $H_i:=\frac{\Lambda_i}{b_i}
=\frac{2l_i\Lambda_i\eta_E\eta_L}{F(\eta_E+\eta_L)}.$
Furthermore, let the per-rider profit be $p_i:=\bar p_i-\frac{c_i}{F}\in[0,M_i].$ Then, for $x_i<\tilde k_i^*(p_i)$, Equation~\eqref{eq:rider-demand-function} implies $\Lambda_i^M(x_i,p_i) = a_i x_i+b_i(M_i-p_i)x_i\left(1-\frac{x_i}{k_i^*}\right),$ while for $x_i\ge \tilde k_i^*(p_i)$, $\Lambda_i^M(x_i,p_i)=\Lambda_i.$ Accordingly, re-defining the per-driver profit as a function of the per-rider profit (rather than the fare), we have:
{\setlength{\abovedisplayskip}{2pt}
\setlength{\belowdisplayskip}{2pt}
\setlength{\jot}{1pt}
\begin{align*}
    \pi_i(x_i,p_i)= 
    \begin{cases}
    p_i\left(a_i+b_i(M_i-p_i)\left(1-\dfrac{x_i}{k_i^*}\right)\right),
    & x_i<\tilde k_i^*(p_i),\\[8pt]
    \dfrac{p_i\Lambda_i}{x_i},
    & x_i\ge \tilde k_i^*(p_i),
    \end{cases}
\end{align*}}
and the total profit on the route is $P_i(x_i,p_i):=p_i\Lambda_i^M(x_i,p_i)=x_i\,\pi_i(x_i,p_i).$

Next, to prove our claim, fix a route $i$. For every fixed fare $\Bar{p}_i$ (and, hence, fixed per-rider profit $p_i$), first note that the map $x\mapsto \pi_i(x,p_i)$ is weakly decreasing; hence, $q_i(\cdot)$, which is the point-wise maximum of continuous weakly decreasing functions, is also continuous and weakly decreasing.

We next show that $q_i(x)$ is computable in $O(1)$ time for each $x$. If $x\ge k_i^*$,
the entire demand is served for every $p_i\in[0,M_i]$; hence, $q_i(x)=\max_{p_i\in[0,M_i]}\frac{p_i\Lambda_i}{x}=\frac{M_i\Lambda_i}{x}.$ If, on the other hand, $x<k_i^*$, the entire demand is served when $p_i \leq M_i - \frac{H_i}{x}$. This relation implies that when $p_i\le p_i^{\mathrm{sat}}(x):=\max\left\{0,\,M_i-\frac{H_i}{x}\right\}$, it follows that $\pi_i(x,p_i)=\frac{p_i\Lambda_i}{x}$; hence, the highest profit in this regime is achieved when $p_i=p_i^{\mathrm{sat}}(x)$. When $p_i\ge p_i^{\mathrm{sat}}(x)$, it follows that:
{\setlength{\abovedisplayskip}{2pt}
\setlength{\belowdisplayskip}{2pt}
\setlength{\jot}{1pt}
\begin{align*}
    \pi_i(x,p_i) = p_i\left(a_i+b_i(M_i-p_i)\left(1-\frac{x}{k_i^*}\right)\right),
\end{align*}}
which is a concave quadratic function of $p_i$. Hence, its maximizer is the projection of the unique stationary point $p_i^{\mathrm{quad}}(x) = \frac12\left(M_i+\frac{H_i}{k_i^*-x}\right)$ onto the interval $[p_i^{\mathrm{sat}}(x),M_i]$. Thus, $q_i(x)$ is obtained by evaluating a constant number (three) of candidate prices.

Now fix $\rho\ge 0$. Since $q_i$ is continuous and weakly decreasing, $u_i(\rho)$ is well-defined, and for $x>0$:
{\setlength{\abovedisplayskip}{2pt}
\setlength{\belowdisplayskip}{2pt}
\setlength{\jot}{1pt}
\begin{align*}
    q_i(x)\ge \rho \quad\Longleftrightarrow\quad x\in(0,u_i(\rho)].
\end{align*}}
Next, we show that an equilibrium with common profit $\rho$ exists if and only if $\sum_{i=1}^n u_i(\rho)\geq D$. To this end, suppose first that $\sum_{i=1}^n u_i(\rho)\geq D$. Since zero allocation is permitted on every route, we can choose $x_i\in[0,u_i(\rho)]$ for all $i$ such that $\sum_i x_i=D$. For every route with $x_i>0$, we have $q_i(x_i)\ge\rho$, i.e., there exists some price with per-driver profit at least $\rho$ at the allocation $x_i$. Since $p_i\mapsto \pi_i(x_i,p_i)$ is continuous and satisfies $\pi_i(x_i,0)=0$, there exists $p_i\in[0,M_i]$ such that $\pi_i(x_i,p_i)=\rho$. For routes with $x_i=0$, set $p_i=0$. Under these fares, every used route yields a per-driver profit of exactly $\rho$, while every unused route yields profit $0\le \rho$. Since each route's per-driver profit is weakly decreasing in its own allocation, no driver can profitably deviate. Hence, the resulting allocation is an equilibrium, and its cumulative profit is $\sum_{i=1}^n P_i(x_i,p_i)=\sum_{i=1}^n x_i\rho=D\rho.$

Conversely, suppose $\x$ is an equilibrium driver allocation, $\p$ is the vector of per-rider profits, with common per-driver profit $\rho$ on all used routes. Then, for every used route $i$, $\rho =\pi_i(x_i,p_i)\le q_i(x_i),$ so $x_i\le u_i(\rho)$. Summing over all routes gives $D=\sum_{i=1}^n x_i\le \sum_{i=1}^n u_i(\rho).$ Thus, we have established that an equilibrium allocation with common profit $\rho$ exists if and only if $\sum_{i=1}^n u_i(\rho)\geq D$.

Note that, for any fixed fare vector, all equilibrium allocations yield the same common per-driver profit, by the monotonicity of per-driver profits and the fixed total driver mass $D$. Consequently, the worst-equilibrium profit under any fare vector equals the profit of any induced equilibrium, including the equilibrium constructed above. It follows that the highest equilibrium common profit is $\rho_I^{*}=\sup\{\rho\ge 0:\sum_{i=1}^n u_i(\rho)\ge D\},$ and thus the highest equilibrium cumulative profit is $D\rho_I^{*}$. Finally, since every $u_i(\rho)$ is weakly decreasing in $\rho$ for $i \in [n]$, the function $\sum_{i=1}^n u_i(\rho)$ is also weakly decreasing. Hence $\rho_I^{*}$ can be found by binary search on $[0,\overline\pi]$, where, for each candidate $\rho$, $u_i(\rho)$ is computed using binary search on $x\in[0,D]$. Here $\overline\pi = \max_i \max_{p_i \in [0, M_i]} \pi_i(0, p_i)$. Finally, recalling here that $q_i(x)$ can be evaluated in $O(1)$ time, we obtain the desired complexity.

\subsection{Proof of Theorem \ref{thm:system-opt-profit-approx}}\label{apdx:system-opt-profit-computation}

In the following proof, we express all quantities as a function of the per-rider profits $p_i = \bar p_i - \frac{c_i}{F}$ for each route $i$. Define $G_i(x):=\max_{p_i\in[0,M_i]} P_i(x,p_i)$ for $x\in[0,D]$ and consider the same quantities $a_i, b_i, M_i$ defined in the proof of Theorem~\ref{thm:equilibrium-profit-computation}. Then, for each fixed $x$, the maximization defining $G_i(x)$ is one-dimensional and can be solved in $O(1)$ time akin to the proof of Theorem~\ref{thm:equilibrium-profit-computation}. Next, for every fixed per-rider profit $p_i = \Bar{p}_i - \frac{c_i}{F}$, the route profit $x\mapsto P_i(x,p_i)$ is weakly increasing and on the quadratic branch has derivative
{\setlength{\abovedisplayskip}{2pt}
\setlength{\belowdisplayskip}{2pt}
\setlength{\jot}{1pt}
\begin{align*}
    \frac{\partial P_i(x,p_i)}{\partial x} = p_i\left(a_i+b_i(M_i-p_i)\left(1-\frac{2x}{k_i^*}\right)\right) \le p_i\bigl(a_i+b_i(M_i-p_i)\bigr)\le L_i^P,
\end{align*}}
where $L_i^P:=\max_{\bar p_i\in \mathcal P_i} p_i(\bar p_i)\bigl(a_i+b_i(M_i-p_i)\bigr).$ On the saturated branch the derivative is $0$. Thus, each map $x\mapsto P_i(x,p_i)$ is $L_i^P$-Lipschitz, and so $G_i$, the pointwise maximum of such maps, is also $L_i^P$-Lipschitz. Moreover $G_i$ is weakly increasing.

Since $G_i$ is weakly increasing for all $i$, the profit maximization problem with constraints $\sum_i x_i=D$ and $\sum_i x_i\le D$ have the same optimum value. Next, define $h:=\frac{\varepsilon}{\sum_{i=1}^n L_i^P}$ and $m:=\left\lfloor\frac{D}{h}\right\rfloor.$ Further, for $i \in [n]$ and $s\in\{0,1,\dots,m\}$, define
{\setlength{\abovedisplayskip}{2pt}
\setlength{\belowdisplayskip}{2pt}
\setlength{\jot}{1pt}
\begin{align*}
    V_i(s):= \max\left\{\sum_{j=1}^i G_j(x_j): x_j\in\{0,h,2h,\dots\},\ \sum_{j=1}^i x_j\le sh \right\}.
\end{align*}}
Then \(V_0(s)=0\) and $V_i(s)=\max_{0\le t\le s}
\Bigl\{V_{i-1}(s-t)+G_i(th)\Bigr\}.$
This dynamic program has $n(m+1)$ states, and each state scans at most $m+1$ transitions, so the
running time is $O(nm^2) = O\left(n \left( \frac{D}{h} \right)^2 \right) = O\left(n \left( \frac{D \sum_{i=1}^n L_i^P}{\varepsilon} \right)^2 \right) = O\left(\frac{n^3}{\varepsilon^2}\right)$, where the final equality uses that the route-level primitives are bounded independently of $n$, which in particular implies $L_i^P \leq M_i(a_i+b_iM_i)=O(1).$

Next, let $\x^*$ be an optimal solution of the (continuous) profit maximization problem, and define the rounded allocation $\tilde x_i:=h\Bigl\lfloor \frac{x_i^*}{h}\Bigr\rfloor$ for all $i$. Then, $\sum_i \tilde x_i\le D$, so $\tilde \x$ is feasible for the discretized problem with the inequality constraint $\sum_i \tilde x_i\le D$. Since each $G_i$ is $L_i^P$-Lipschitz, it follows that:
{\setlength{\abovedisplayskip}{2pt}
\setlength{\belowdisplayskip}{2pt}
\setlength{\jot}{1pt}
\begin{align*}
    G_i(x_i^*)-G_i(\tilde x_i)\le L_i^P(x_i^*-\tilde x_i)\le L_i^P h.
\end{align*}}
Summing the above inequality over all routes gives $P_I^{\mathrm{fare,opt}}-\sum_{i=1}^n G_i(\tilde x_i) \le h\sum_{i=1}^n L_i^P = \varepsilon.$ Since the dynamic program computes the best grid allocation:
{\setlength{\abovedisplayskip}{2pt}
\setlength{\belowdisplayskip}{2pt}
\setlength{\jot}{1pt}
\begin{align*}
    V_n(m)\ge \sum_{i=1}^n G_i(\tilde x_i)\ge P_I^{\mathrm{fare,opt}}-\varepsilon.
\end{align*}}
\color{blue} If $\sum_{i=1}^n \hat{x}_i<D$, assign the remaining driver mass to any route, say route one, by replacing $\hat{x}_1 \leftarrow \hat{x}_1+\left(D-\sum_{i=1}^n\hat{x}_i\right).$ Since $G_1$ is weakly increasing, this does not decrease the objective, and the resulting allocation satisfies $\sum_{i=1}^n\hat{x}_i=D$. \color{black}

Finally, for the dynamic program solution \(\hat \x\), choose $\hat p_i\in\arg\max_{p_i\in[0,M_i]} P_i(\hat x_i,p_i)$ for each route $i$, which can be performed in $O(1)$ time. Then, $\sum_{i=1}^n P_i(\hat x_i,\hat p_i)=\sum_{i=1}^n G_i(\hat x_i)\ge P_I^{\mathrm{fare,opt}}-\varepsilon,$ which proves our claim.

\section{Fare Optimization for Rider Throughput}\label{apdx:price-opt-welfare}

\subsection{Computational Tractability of Fare-optimized Rider Throughput Maximization with no Reservation Wage}

The implementing perturbation in Theorem~\ref{thm:fare-welfare-poa} is constructive and can be computed efficiently. First, the rider-throughput maximizing allocation $\x^*$ at the minimum fare vector $\bar{\p}^{\min}$ can be computed efficiently, since the resulting problem is a separable concave resource-allocation problem. Given $\x^*$, for any candidate common profit level $\kappa$, we compute the perturbation $q_i^{\varepsilon}$ on each used route $i$ by solving $f_i(q_i^{\varepsilon})=\kappa$, where $f_i(q_i^{\varepsilon}) = \pi_i(x_i^*, \frac{c_i}{F} + q_i^{\varepsilon})$ is either linear or quadratic in $q_i^{\varepsilon}$, and hence $q_i^{\varepsilon}$ can be computed in closed form; for each unused route $j$, we choose any sufficiently small $q_j^{\varepsilon}$ so that $f_j(q_j^{\varepsilon})<\kappa$. Moreover, for fixed $\x^*$, the following Lipschitz continuity relation can be established: the rider throughput loss induced by these perturbations is at most linear in $\kappa$. In other words, there exists a constant \(C>0\) such that $R_I(\x^*,\bar{\p}^{\min})-R_I(\x^*,\bar{\p}^{\,\kappa})\le C\kappa$. Therefore, choosing $\kappa=O(\varepsilon)$ suffices to obtain a $(1+\varepsilon)$-approximate implementation, and once $\x^*$ is known, the corresponding implementing fare vector can be computed in $O(n)$ time. Note that if one prefers to determine such a $\kappa$ by search rather than by an explicit bound, this adds only $O(\log(1/\varepsilon))$ iterations, for an overall complexity of $O(n\log(1/\varepsilon))$ after computing $\x^*$.

\begin{proof}[Proof of Lipschitz continuity Relation.]
Fix the numbers $\delta_i>0$ in the proof of
Theorem~\ref{thm:fare-welfare-poa}, and define $\underline{\kappa}:=\min_{i\in U} f_i(\delta_i)>0$ and $m_i:=\frac{\Lambda_i^M\!\left(x_i^*,\,\frac{c_i}{F}+\delta_i\right)}{x_i^*}>0$, where $i \in U$. For each $i \in U$ and $\kappa\in(0,\underline{\kappa}]$, let $q_i(\kappa)\in(0,\delta_i]$ be any solution of $f_i(q)=\kappa$. 

Next, define $r_i(q):=\Lambda_i^M\!\left(x_i^*,\,\frac{c_i}{F}+q\right).$ By Lemma~\ref{lem:fare-monotonicity}, $r_i$ is non-increasing, so for every $i \in U$, $\kappa=f_i(q_i(\kappa)) =\frac{r_i(q_i(\kappa))}{x_i^*}\,q_i(\kappa) \ge m_i q_i(\kappa),$ and therefore $q_i(\kappa)\le \frac{\kappa}{m_i}.$ For fixed $x_i^*$, the branch condition depends on $q$ through
{\setlength{\abovedisplayskip}{2pt}
\setlength{\belowdisplayskip}{2pt}
\setlength{\jot}{1pt}
\begin{align*}
    \tilde{k}_i^*\!\left(\frac{c_i}{F}+q\right)
=
\min\left\{
k_i^*,\,
\frac{2l_i\Lambda_i\eta_E\eta_L}{F(M_i-q)(\eta_E+\eta_L)}
\right\},
\qquad M_i:=c_i^O-\eta_Tl_i-\frac{c_i}{F}.
\end{align*}}
Hence \(r_i(q):=\Lambda_i^M\!\left(x_i^*,\frac{c_i}{F}+q\right)\) is either
(i) always on the quadratic branch, (ii) always equal to \(\Lambda_i\), or (iii) equal to
\(\Lambda_i\) up to the unique switching point $\hat q_i = M_i-\frac{2l_i\Lambda_i\eta_E\eta_L}{F(\eta_E+\eta_L)x_i^*},$ and affine thereafter. On the quadratic branch, we have $r_i(q)=A_i(x_i^*)+B_i(x_i^*)(M_i-q),$ so its slope is $-B_i(x_i^*)$, while on the saturated branch its slope is $0$. Therefore $r_i$ is globally $L_i$-Lipschitz on $[0,\delta_i]$, with
{\setlength{\abovedisplayskip}{2pt}
\setlength{\belowdisplayskip}{2pt}
\setlength{\jot}{1pt}
\begin{align*}
    L_i=
\begin{cases}
B_i(x_i^*), & x_i^*<k_i^*,\\[4pt]
0, & x_i^*\ge k_i^*.
\end{cases}
\end{align*}}
Thus, for every $i \in U$, $0\le r_i(0)-r_i(q_i(\kappa)) \le L_i q_i(\kappa) \le \frac{L_i}{m_i}\kappa.$ Summing over $i \in U$ yields $R(\x^*, \Bar{\p}^{\min})-R_I(\x^*,\bar{\p}^{\,\kappa}) \le C\kappa,$ where $C:= 1+ \sum_{i\in U}\frac{L_i}{m_i}.$ Therefore, any $\kappa \le \min\left\{ \underline{\kappa}, \frac{\varepsilon}{1+\varepsilon}\frac{R_I^{\mathrm{fare, opt}}}{C} \right\}$ guarantees $R_I(\x^*,\bar{\p}^{\,\kappa})\ge \frac{1}{1+\varepsilon}R_I^{\mathrm{fare, opt}}.$ 
\end{proof}

\subsection{Fare-Optimized Rider Throughput PoA under Reservation Wages} \label{apdx:reservation-wages-extension-poa}

Suppose drivers have a reservation wage $W \geq 0$, interpreted as an individual rationality constraint on per-driver profit. When $W = 0$, this reduces to the setting studied in Theorem~\ref{thm:fare-welfare-poa}; below, we focus on the case $W>0$. Under a reservation wage, the centralized fare-optimized rider-throughput benchmark is
{\setlength{\abovedisplayskip}{2pt}
\setlength{\belowdisplayskip}{2pt}
\setlength{\jot}{1pt}
\begin{align*}
    R_I^{\mathrm{fare,opt}}(W) :=
    \max_{\substack{\x\in\Omega,\ \bar \p\in\mathcal P \\
    \pi_i(x_i,\bar p_i)\ge W\ \forall i:\,x_i>0}}
    R_I(\x,\bar \p).
\end{align*}}
For a fare vector \(\bar \p\in\mathcal P\), define the set of reservation-wage-feasible equilibria as $\Omega_W^{\mathrm{Eq}}(\bar \p) := \left\{\x\in\Omega^{\mathrm{Eq}}(\bar \p): \pi_i(x_i,\bar p_i)\ge W, \ \forall i \text{ with } x_i>0 \right\}.$ The corresponding fare-optimized equilibrium benchmark is
{\setlength{\abovedisplayskip}{2pt}
\setlength{\belowdisplayskip}{2pt}
\setlength{\jot}{1pt}
\begin{align*}
    R_I^{\mathrm{fare,eq}}(W)
    :=
    \sup_{\substack{\bar{\p}\in\mathcal{P}:\\
\Omega_W^{Eq}(\bar{\p})\neq\emptyset}}
    \inf_{\x\in\Omega_W^{\mathrm{Eq}}(\bar \p)}
    R_I(\x,\bar \p).
\end{align*}}

For each route $i$ and allocation $x_i>0$, define the smallest fare that meets the reservation wage as $\bar p_i^{W}(x_i) := \inf\left\{ \bar p_i\in\left[\frac{c_i}{F},\,c_i^O-\eta_T l_i\right]: \pi_i(x_i,\bar p_i)\geq W \right\},$ whenever this set is nonempty. By the monotonicity of rider demand in fares, rider throughput is weakly decreasing in $\bar p_i$ for fixed $x_i$. Hence, for any fixed feasible driver allocation $\x$, the rider throughput maximizing fare on each used route is precisely $\bar p_i^W(x_i)$. Therefore, the reservation wage constrained rider throughput optimum can be written as
{\setlength{\abovedisplayskip}{2pt}
\setlength{\belowdisplayskip}{2pt}
\setlength{\jot}{1pt}
\begin{align*}
    R_I^{\mathrm{fare,opt}}(W) = \max_{\x\in\Omega_W}
    \sum_{i:\,x_i>0} \Lambda_i^M\!\bigl(x_i,\bar p_i^W(x_i)\bigr),
\end{align*}}
where $\Omega_W$ denotes the set of allocations for which $\bar p_i^W(x_i)$ exists on every used route.

Let \(\x^{W,*}\) be an optimizer of the above problem, and set
\begin{align*}
    \bar p_i^{W,*}
    :=
    \begin{cases}
    \bar p_i^W(x_i^{W,*}), & \text{if } x_i^{W,*}>0,\\[4pt]
    \frac{c_i}{F}+\delta_i, & \text{if } x_i^{W,*}=0,
    \end{cases}
\end{align*}
where $\delta_i>0$ is chosen sufficiently small so that $\pi_i\!\left(0,\frac{c_i}{F}+\delta_i\right)<W$ for all $i$ with $x_i^{W,*}=0$. Such a choice is possible whenever $\frac{c_i}{F}<c_i^O-\eta_T l_i$ by continuity of $\pi_i(0,\bar p_i)$ and the fact that $\pi_i\!\left(0,\frac{c_i}{F}\right)=0<W.$

By construction, every used route $i$ satisfies $\pi_i(x_i^{W,*},\bar p_i^{W,*})=W,$ while every unused route $j$ satisfies $\pi_j(0,\bar p_j^{W,*})<W.$ Hence, all used routes yield the same per-driver profit and every unused route yields strictly less. Therefore, $\x^{W,*}\in \Omega_W^{\mathrm{Eq}}(\bar \p^{W,*}).$

\color{blue} Consequently, since $W>0$ and all used routes earn per-driver profits equal to $W$, their per-rider profits $p_i(\bar{p}_i^{W,*})$ are
strictly positive. If, in addition, $\bar p_i^{W,*}<c_i^O-\eta_T l_i$ for all routes $i$ with $x_i^{W,*}>0$, then $S_i(\bar p_i^{W,*})>0$ on every used route. Under these conditions, per-driver profits are strictly decreasing on used routes, while unused routes yield strictly lower profits at zero driver mass. Hence, the equilibrium under $\bar{\p}^{W,*}$ is unique, and the fare-optimized rider throughput PoA remains equal to one in the reservation wage constrained setting. \color{black}

If, however, some used route satisfies $\bar p_i^{W,*}=c_i^O-\eta_T l_i$, then the strict uniqueness argument may fail because $S_i(\bar p_i^{W,*})=0$ on that route. In this case, there still exists an equilibrium attaining the reservation wage constrained rider throughput optimum, but the equilibrium need not be unique; consequently, other equilibria under the same fare vector may achieve lower rider throughput.

\subsection{Computational Tractability of Fare-optimized Rider Throughput Maximization with Reservation Wage} \label{apdx:welf-max-res-wage}

\begin{proposition}[Additive approximation for the reservation wage-constrained system optimum]
\label{prop:reservation-wage-approx}
Assume the reservation wage $W>0$, which can be interpreted as the outside option for drivers, so the planner may leave some drivers unmatched. Then, the planner solves
{\setlength{\abovedisplayskip}{2pt}
\setlength{\belowdisplayskip}{2pt}
\setlength{\jot}{1pt}
\begin{align*}
    \max_{\substack{x_i\ge 0,\ \bar p_i\in[\frac{c_i}{F},\,c_i^O-\eta_Tl_i]\\
\sum_{i=1}^n x_i\le D,\ \pi_i(x_i,\bar p_i)\ge W\ \forall i:\,x_i>0}}
\sum_{i=1}^n \Lambda_i^M(x_i,\bar p_i).
\end{align*}}
Let $\widehat{R}_I^{\mathrm{fare,opt}}(W)$ denote the optimal value of this problem. For each route $i$, define $M_i:=c_i^O-\eta_Tl_i-\frac{c_i}{F}$, $a_i:=\frac{F(t_2-t_1)}{2l_i}$, $b_i:=\frac{F(\eta_E+\eta_L)}{2l_i\eta_E\eta_L}$, $L_i:=a_i+b_iM_i$, and $U_i:= u_i(W).$ 
Also, define the route-wise reservation wage throughput function 
\[
R_i^W(x):= \max\Bigl\{\Lambda_i^M(x,\bar p_i): \bar p_i\in \Bigl[\frac{c_i}{F},\,c_i^O-\eta_Tl_i\Bigr],\ \pi_i(x,\bar p_i)\ge W \Bigr\},
\qquad R_i^W(0):=0.
\]
Then, for every $\varepsilon>0$, there is an algorithm that returns a feasible allocation $\hat \x$ with $\sum_{i=1}^n R_i^W(\hat x_i)\ge \widehat{R}_I^{\mathrm{fare,opt}}(W)-\varepsilon$. The algorithm runs in $O(n(m+1)^2)$, where $m :=\left\lfloor \frac{D\sum_{i=1}^n L_i}{\varepsilon}\right\rfloor,$ which is polynomial in $n$ and $1/\varepsilon$.

\end{proposition}

\begin{proof}
By Lemma~\ref{lem:fare-monotonicity}, for each fixed \(x_i\), rider throughput is weakly decreasing in \(\bar p_i\). Hence
\(R_i^W(x)\) is attained at the smallest fare satisfying the reservation wage constraint. Therefore,
for any fixed \(x\), each value \(R_i^W(x)\) can be computed in \(O(1)\) time by solving a single
linear or quadratic equation in the per-rider profit $q_i:=\bar p_i-\frac{c_i}{F}$: on the full-demand
branch one has \(q_i=\frac{Wx}{\Lambda_i}\), while on the quadratic branch \(q_i\) is the smaller root
of the quadratic equation $q_i\frac{\Lambda_i^M(x,\frac{c_i}{F}+q_i)}{x}=W.$

Moreover, for any feasible route allocation $x_i>0$, the reservation wage constraint implies that $q_i(x_i)\geq W$, and hence $x_i\leq u_i(W)=U_i$.


\textbf{Lipschitz Bound:} We next prove a one-sided Lipschitz bound for $R_i^W$. Fix $0\le y\le x\le U_i$, and let $\bar p_i^W(x)$ be a throughput-maximizing fare for route $i$ at allocation $x$. Since the
per-driver profit on a route is non-increasing in the route allocation, the same fare $\bar p_i^W(x)$ is also feasible at the smaller allocation $y$. Therefore, $R_i^W(y)\ge \Lambda_i^M\!\bigl(y,\bar p_i^W(x)\bigr)$ and $R_i^W(x)=\Lambda_i^M\!\bigl(x,\bar p_i^W(x)\bigr).$ Hence,
{\setlength{\abovedisplayskip}{1pt}
\setlength{\belowdisplayskip}{1pt}
\setlength{\jot}{1pt}
\begin{align*}
    R_i^W(x)-R_i^W(y) \le \Lambda_i^M\!\bigl(x,\bar p_i^W(x)\bigr) - \Lambda_i^M\!\bigl(y,\bar p_i^W(x)\bigr).
\end{align*}}
For any fixed fare, the function \(z\mapsto \Lambda_i^M(z,\bar p_i)\) is continuous, has derivative
\(0\) on the saturated branch, and on the quadratic branch has derivative at most $a_i+b_iM_i=L_i.$ Thus, for every fixed fare, $\Lambda_i^M(x,\bar p_i)-\Lambda_i^M(y,\bar p_i)\le L_i(x-y),$ which implies that: $R_i^W(x)-R_i^W(y)\le L_i(x-y).$

\textbf{Dynamic Program:}
We now use a dynamic programming approach to establish our guarantee. To this end, choose the grid size $h:=\frac{\varepsilon}{\sum_{i=1}^n L_i}$ and $m:=\left\lfloor \frac{D}{h}\right\rfloor.$ For each $i \in [n]$ and $s\in\{0,1,\dots,m\}$, define the set
{\setlength{\abovedisplayskip}{1pt}
\setlength{\belowdisplayskip}{1pt}
\setlength{\jot}{1pt}
\begin{align*}
    \mathcal F_i(s):= \left\{ (x_1,\dots,x_i): x_j\in\{0,h,2h,\dots\},\ 0\le x_j\le U_j,\ \sum_{j=1}^i x_j\le sh \right\},
\end{align*}}
and define $V_i(s):=\max_{(x_1,\dots,x_i)\in\mathcal F_i(s)} \sum_{j=1}^i R_j^W(x_j).$

Thus, $V_i(s)$ is the maximum rider throughput attainable using only the first $i$ routes and at
most $sh$ units of driver mass on the grid. Note that the boundary condition is $V_0(s)=0$ for all $s$. We claim that $V_i(s)$ satisfies the recursion
{\setlength{\abovedisplayskip}{1pt}
\setlength{\belowdisplayskip}{1pt}
\setlength{\jot}{1pt}
\begin{align*}
    V_i(s)=\max_{0\le t\le \min\{s,\lfloor U_i/h\rfloor\}} \Bigl\{V_{i-1}(s-t)+R_i^W(th) \Bigr\}.
\end{align*}}
To prove this, fix $i, s$. First, let $(x_1,\dots,x_i)\in\mathcal F_i(s)$ be any feasible grid
allocation with $x_i=th$. Then, $0\le t\le \min\{s,\lfloor U_i/h\rfloor\},$ and the vector $(x_1,\dots,x_{i-1})$ belongs to $\mathcal F_{i-1}(s-t)$, since $\sum_{j=1}^{i-1}x_j \le sh-th=(s-t)h.$ Thus,
{\setlength{\abovedisplayskip}{1pt}
\setlength{\belowdisplayskip}{1pt}
\setlength{\jot}{1pt}
\begin{align*}
    \sum_{j=1}^i R_j^W(x_j) = \sum_{j=1}^{i-1}R_j^W(x_j)+R_i^W(th) \le V_{i-1}(s-t)+R_i^W(th).
\end{align*}}
Since this holds for every feasible allocation in $\mathcal F_i(s)$, we obtain
{\setlength{\abovedisplayskip}{1pt}
\setlength{\belowdisplayskip}{1pt}
\setlength{\jot}{1pt}
\begin{align*}
    V_i(s)\le \max_{0\le t\le \min\{s,\lfloor U_i/h\rfloor\}} \Bigl\{ V_{i-1}(s-t)+R_i^W(th) \Bigr\}.
\end{align*}}
For the reverse inequality, fix any integer $0\le t\le \min\{s,\lfloor U_i/h\rfloor\},$ and let $(x_1,\dots,x_{i-1})\in\mathcal F_{i-1}(s-t)$ attain the value $V_{i-1}(s-t)$. Then
setting $x_i:=th$ yields a vector in $\mathcal F_i(s)$, because $th\le U_i$ and
{\setlength{\abovedisplayskip}{1pt}
\setlength{\belowdisplayskip}{1pt}
\setlength{\jot}{1pt}
\begin{align*}
    \sum_{j=1}^{i}x_j = \sum_{j=1}^{i-1}x_j+th \le (s-t)h+th = sh.
\end{align*}}
Hence, $V_i(s)\ge V_{i-1}(s-t)+R_i^W(th),$ and taking the maximum over all feasible $t$ yields the opposite inequality, and so the recursion holds.

By induction on $i$, the dynamic program computes the exact optimum over all grid allocations. The dynamic program has $n(m+1)$ states, and each state $V_i(s)$ is computed by scanning over at most $m+1$ possible values of $t$; resulting in a total running time of $O(n(m+1)^2)$ and a space complexity of $O(nm)$.

\textbf{Concluding the Proof:} Finally, consider a rounded vector $\Tilde{\x}$ with $\tilde x_i:=h\Bigl\lfloor \frac{x_i^*}{h}\Bigr\rfloor,$ then $\sum_{i=1}^n \frac{\tilde x_i}{h} = \sum_{i=1}^n \Bigl\lfloor \frac{x_i^*}{h}\Bigr\rfloor \le \Bigl\lfloor \sum_{i=1}^n \frac{x_i^*}{h}\Bigr\rfloor \le \Bigl\lfloor \frac{D}{h}\Bigr\rfloor = m,$ so the rounded solution $\tilde \x$ is feasible for the grid problem and is therefore dominated by the dynamic-programming optimum $V_n(m)$. Next, by the earlier obtained Lipschitz bound, it follows that:
{\setlength{\abovedisplayskip}{1pt}
\setlength{\belowdisplayskip}{1pt}
\setlength{\jot}{1pt}
\begin{align*}
    R_i^W(x_i^*)-R_i^W(\tilde x_i)\le L_i(x_i^*-\tilde x_i)\le L_i h.
\end{align*}}
Summing over $i$ gives $\widehat{R}_I^{\mathrm{fare,opt}}(W)-\sum_{i=1}^n R_i^W(\tilde x_i) \le h\sum_{i=1}^n L_i = \varepsilon.$ Since the dynamic program computes the best grid allocation, it follows that:
{\setlength{\abovedisplayskip}{1pt}
\setlength{\belowdisplayskip}{1pt}
\setlength{\jot}{1pt}
\begin{align*}
    V_n(m) = \max_{0\le s\le m}V_n(s)\ge \sum_{i=1}^n R_i^W(\tilde x_i)\ge \widehat{R}_I^{\mathrm{fare,opt}}(W)-\varepsilon.
\end{align*}}
This proves the additive-\(\varepsilon\) guarantee. Note here that since $R_i^W(x)\ge 0$ for all $i, x$, the values $V_n(s)$ are non-decreasing in $s$, so the best grid solution is $V_n(m)$. 
\end{proof}

\begin{remark} \label{rmk:dummyroute}
In the above proof, we can retain the original equality constraint $\sum_i x_i=D$, with the same proof going through after adding a dummy outside option route that yields payoff $W$ and rider throughput $0$. Then, the planner may send any unused drivers to the dummy route, and the approximation algorithm above applies without change.
\end{remark}

\section{Additional Discussions and Theoretical Results}

\color{blue}

\subsection{Extension to Total Social Welfare Objective} \label{apdx:total-social-welfare-obj}

This section extends the PoA analysis in Sections~\ref{subsec:profit-poa} and~\ref{subsec:demand-poa} to an objective that captures total social welfare. Appendix~\ref{apdx:social-welfare-setup} formally defines this objective and establishes the key route-level properties used in our analysis. Then, Appendix~\ref{apdx:social-welfare-poa} derives an instance-wise bound on the total social welfare ratio and establishes the tightness of its dependence on route-level parameters.

\subsubsection{Setup and welfare objective} \label{apdx:social-welfare-setup}
For this extension, define $p_i:=\bar p_i-\frac{c_i}{F}>0,$ and $S_i:=c_i^O-\bigl(\bar p_i+\eta_T l_i\bigr)\ge 0,$ where $p_i$ is the per-rider driver profit and $S_i$ is the cost advantage of the minibus over the outside option before waiting and schedule delay. For a minibus rider $r$ on route $i$ with $x_i$ drivers, let $d_{ir}(x_i)\ge 0$ be that rider's equilibrium waiting plus schedule-delay cost. The rider's gain relative to taking the outside option is $u_{ir}(x_i):=S_i-d_{ir}(x_i)\ge 0,$ where non-negativity follows from the rider's equilibrium mode choice. We normalize the total social welfare relative to the benchmark in which all riders use their outside options (note that fare payments are transfers that cancel out). In particular,
{\setlength{\abovedisplayskip}{1pt}
\setlength{\belowdisplayskip}{1pt}
\setlength{\jot}{1pt}
\begin{align*}
    c_i^O-\left(\frac{c_i}{F}+\eta_Tl_i+d_{ir}(x_i)\right) = \left(\bar p_i-\frac{c_i}{F}\right)
 +c_i^O-\left(\bar p_i+\eta_Tl_i+d_{ir}(x_i)\right) = p_i+u_{ir}(x_i),
\end{align*}}
which corresponds to the rider gain plus the driver's per-rider profit. Then, define route and total social welfare by
\[
W_i(x_i):=\int_{r\in\mathcal R_i^M(x_i)}\bigl(p_i+u_{ir}(x_i)\bigr)\,dr,
\qquad
W(\x):=\sum_{i=1}^n W_i(x_i),
\]
where $\mathcal R_i^M(x_i)$ is the set of riders using the minibus at equilibrium and has mass $\Lambda_i^M(x_i)$.


\begin{lemma}[Route-level welfare properties]\label{lem:welfare-properties}
For every route $i$ and allocation $x_i$,
{\setlength{\abovedisplayskip}{1pt}
\setlength{\belowdisplayskip}{1pt}
\setlength{\jot}{1pt}
\begin{equation}\label{eq:sandwich}
p_i \Lambda_i^M(x_i)\le W_i(x_i)\le (p_i+S_i) \Lambda_i^M(x_i).
\end{equation}}
Moreover, $W_i(x_i)$ is nondecreasing in $x_i$ under the Vickrey rider equilibrium.
\end{lemma}

\begin{proof}
Every served rider satisfies $0\le u_{ir}(x_i)\le S_i$: the first inequality follows from equilibrium mode choice, while the second follows from $d_{ir}(x_i)\ge0$. Integrating $p_i\le p_i+u_{ir}(x_i)\le p_i+S_i$ over the $\Lambda_i^M(x_i)$ served riders gives \eqref{eq:sandwich}.

Since \(\Lambda_i^M(x_i)\) is non-decreasing in \(x_i\) by Corollary~\ref{cor:monotonicity-rider-demands} and \(p_i>0\), the driver-profit component \(p_i\Lambda_i^M(x_i)\) is also non-decreasing. It therefore remains to show that aggregate rider gain is non-decreasing in \(x_i\). For monotonicity of aggregate rider gain, let $\mu=\mu_i(x_i),$ $\lambda=\lambda_i,$ $S=S_i,$ $q:=\frac{\mu}{\lambda}<1,$ $\beta:=\frac{\eta_E\eta_L}{\eta_E+\eta_L},$ and $\bar S(\mu):=\frac{\Lambda_i\beta}{\mu}.$

\emph{Case 1: $\mu<\lambda$ and $S<\bar S(\mu)$.}
In this regime, only part of the demand uses the minibus. Consider the early minibus riders, who are served over an interval $[t_A,t_B]$ of length $t_B-t_A=\frac{S}{\eta_E}.$ FIFO service implies that the mass of riders served between $t$ and $t_B$ equals the mass whose desired arrival times lie between $t^*$ and $t_B$. Hence $\mu(t_B-t)=\lambda(t_B-t^*).$ Then, letting $q:=\mu/\lambda$ and rearranging this equation, we obtain $t^*-t=(1-q)(t_B-t).$

The equilibrium waiting cost over this early interval is $w(t)=S-\eta_E(t_B-t).$ Thus, the gain of an early rider relative to the outside option is $u(t) =S-\bigl[w(t)+\eta_E(t^*-t)\bigr] =q\eta_E(t_B-t).$ It therefore decreases linearly from $qS$ at $t_A$ to $0$ at $t_B$.
Since riders are served at rate $\mu$, aggregate gain among early riders is $\mu\int_{t_A}^{t_B}u(t)\,dt
=
\frac{\mu^2S^2}{2\lambda\eta_E}.$

The analogous calculation for late riders gives $\frac{\mu^2S^2}{2\lambda\eta_L}.$ Middle minibus riders are indifferent to the outside option and obtain
zero gain. Hence aggregate rider gain is $\frac{\mu^2S^2}{2\lambda}
\left(\frac{1}{\eta_E}+\frac{1}{\eta_L}\right),$ which is increasing in $\mu$ and, hence, $x_i$.

\emph{Case 2: $\mu<\lambda$ and $S\ge \bar S(\mu)$.}
In this regime, all $\Lambda_i$ riders use the minibus, so it suffices to show that their average waiting plus schedule cost decreases with $\mu$. From the Vickrey equilibrium profile, this cost varies linearly on each side of the peak between $(1-q)\bar S(\mu)$ and $\bar S(\mu).$ Its average is therefore $\bar S(\mu)-\frac{q\bar S(\mu)}{2}.$ Using $\Lambda_i=\lambda (t_2-t_1)$, we have $q\bar S(\mu) = \frac{\mu}{\lambda}\frac{\Lambda_i\beta}{\mu} = (t_2 - t_1)\beta,$ so the average waiting plus schedule delay cost is $\bar S(\mu)-\frac{(t_2 - t_1)\beta}{2}.$ Since $\bar S(\mu)=\Lambda_i\beta/\mu$ decreases in $\mu$, aggregate rider gain is non-decreasing in $\mu$ and, hence, $x_i$.

\emph{Case 3: $\mu \ge \lambda$.} Finally, when $\mu\ge\lambda$, waiting and schedule delay vanish, so every rider obtains gain $S$. Hence aggregate rider gain, and therefore $W_i(x_i)$, is non-decreasing in $x_i$.
\end{proof}

\subsubsection{Welfare price of anarchy} \label{apdx:social-welfare-poa}
Fix an instance $I$, and let $\x^{*}\in\arg\max_{\x\in\Omega}W(\x)$ be a welfare-optimal driver allocation and $\x^{Eq}$ be any equilibrium driver allocation. Further, define $\rho(I):=\max_{i\in[n]}\frac{S_i}{p_i}.$ The welfare price of anarchy is the worst-case ratio $W(\x^{*})/W(\x^{Eq})$ over instances with rider demand functions in $\mathcal L_V$ and equilibrium allocations.

\begin{theorem}[Instance-Wise upper bound for total social welfare]\label{thm:welfare-poa}
For any fixed-fare informal transit instance $I$ with rider demand functions in $\mathcal L_V$ and $p_i>0$ for all routes, $\frac{W_I(\x^*)}{W_I(\x^{Eq})} \leq 2 + \rho(I) = 2 + \max_{i \in [n]} \frac{S_i}{p_i}$. Equivalently, every equilibrium achieves at least a $1/(2+\rho(I))$ fraction of the centralized welfare optimum.
\end{theorem}

\begin{proof}
Fix an instance and abbreviate $\rho = \rho(I)$. From Lemma~\ref{lem:welfare-properties},
{\setlength{\abovedisplayskip}{1pt}
\setlength{\belowdisplayskip}{1pt}
\setlength{\jot}{1pt}
\begin{equation}\label{eq:welfare-profit-comparison}
p_i \Lambda_i^M(x_i)\le W_i(x_i)\le (p_i+S_i) \Lambda_i^M(x_i)
\le (1+\rho)p_i \Lambda_i^M(x_i).
\end{equation}}
Next, partition the routes into the sets $L_1:=\{i:x_i^{Eq}\ge x_i^{*}\}$ and $L_2:=\{i:x_i^{Eq}<x_i^{*}\}.$

For $i\in L_1$, route welfare monotonicity gives $W_i(x_i^{*})\le W_i(x_i^{Eq}).$ For $i\in L_2$, the upper bound in \eqref{eq:welfare-profit-comparison} gives
{\setlength{\abovedisplayskip}{1pt}
\setlength{\belowdisplayskip}{1pt}
\setlength{\jot}{1pt}
\begin{align*}
    W_i(x_i^{*})
\le(1+\rho)p_i \Lambda_i^M(x_i^{*})
=(1+\rho)\pi_i(x_i^{*})x_i^{*}.
\end{align*}}
At a driver equilibrium, all used routes earn a common per-driver profit $\pi^{Eq}$ and every unused route offers at most $\pi^{Eq}$. Since $x_i^{*}>x_i^{Eq}$ for $i\in L_2$ and $\pi_i(\cdot)$ is non-increasing, $\pi_i(x_i^{*})\le \pi_i(x_i^{Eq})\le\pi^{Eq}.$

Therefore,
{\setlength{\abovedisplayskip}{1pt}
\setlength{\belowdisplayskip}{1pt}
\setlength{\jot}{1pt}
\begin{align*}
W(\x^{*})
&=\sum_{i\in L_1}W_i(x_i^{*})+\sum_{i\in L_2}W_i(x_i^{*}) \le \sum_{i\in L_1}W_i(x_i^{Eq})
 +(1+\rho)\pi^{Eq}\sum_{i\in L_2}x_i^{*} \le W(\x^{Eq})+(1+\rho)\pi^{Eq}D.
\end{align*}}
Because all $D$ drivers are allocated, their total equilibrium profit is $P(\x^{Eq})=\pi^{Eq}D=\sum_i p_i \Lambda_i^M(x_i^{Eq}).$

The lower bound in \eqref{eq:welfare-profit-comparison} implies $P(\x^{Eq})\le W(\x^{Eq})$. Hence
{\setlength{\abovedisplayskip}{1pt}
\setlength{\belowdisplayskip}{1pt}
\setlength{\jot}{1pt}
\begin{align*}
W(\x^{*})\le W(\x^{Eq})+(1+\rho)W(\x^{Eq})
=(2+\rho)W(\x^{Eq}),
\end{align*}}
which proves the claim.
\end{proof}

\begin{proposition}[Tightness]\label{prop:tightness}
For every $\rho\ge0$ and every $\varepsilon>0$, there exists an instance $I$ with $\rho(I) = \rho$ and an equilibrium allocation $\x^{Eq}$ such that $\frac{W_I(\x^*)}{W_I(\x^{Eq})} > 2 + \rho-\varepsilon$, where $\x^*$ is the total social welfare maximizing allocation. Thus, the dependence of the bound in Theorem~\ref{thm:welfare-poa} on $\rho$ is tight.
\end{proposition}

\begin{proof}
Consider two routes and driver supply $D=1$. Fix $F=2$, $l_1=l_2=1$, schedule-delay parameters $\eta_E,\eta_L>0$, and let $a:=\frac{1}{\eta_E}+\frac{1}{\eta_L}.$ Choose a peak length $t_2 - t_1 = H$ and, for $\delta\in(0,1)$, set desired rider arrival rates $\lambda_1=\delta$ and $\lambda_2=1-\delta$. Thus $\Lambda_1=\delta H$, $\Lambda_2=(1-\delta)H$, and the no-queue driver thresholds are $k_1^{*}=\delta$ and $k_2^{*}=1-\delta$. Further, set $S_1=0,$ $p_2=1,$ $S_2=\rho,$ and $p_1:=\frac{H+\rho a}{\delta H}.$

For the Vickrey demand relation, the per-driver profit available on an empty route is $\pi_i(0)=p_i\frac{F}{2l_i}
\left(H+S_i\left(\frac1{\eta_E}+\frac1{\eta_L}\right)\right).$ Hence $\pi_2(0)=H+\rho a$. 

On route 1, $S_1=0$ and $x_1=1\ge k_1^{*}$, so all $\Lambda_1$ riders are served and $\pi_1(1)=p_1\Lambda_1=H+\rho a=\pi_2(0).$

Therefore $\x^{Eq}=(1,0)$ is an equilibrium. Since $S_1=0$, we have $u_{1r}(1) \leq S_1 = 0$; hence, $u_{1r}(1) = 0$. Thus, the welfare reduces to the total profit and is given by $W(\x^{Eq})=H+\rho a.$

Now consider $\widehat \x=(\delta,1-\delta)$. Each route is exactly at its no-queue threshold, so all riders are served with zero waiting and schedule delay. Therefore
{\setlength{\abovedisplayskip}{1pt}
\setlength{\belowdisplayskip}{1pt}
\setlength{\jot}{1pt}
\begin{align*}
W(\widehat \x)
= (p_1+S_1)\Lambda_1 + (p_2+S_2) \Lambda_2= (H+\rho a)+(1+\rho)(1-\delta)H.
\end{align*}}
Since $x^{*}$ is welfare-optimal,
{\setlength{\abovedisplayskip}{1pt}
\setlength{\belowdisplayskip}{1pt}
\setlength{\jot}{1pt}
\begin{align*}
\frac{W(\x^{*})}{W(\x^{Eq})}
\ge \frac{W(\widehat \x)}{W(\x^{Eq})} =
1+\frac{(1+\rho)(1-\delta)H}{H+\rho a} = 1+\frac{(1+\rho)(1-\delta)}{1+\frac{\rho a}{H}}.
\end{align*}}
First taking $H\to\infty$ and then $\delta\to0$ makes the right-hand side approach $2+\rho$.
\end{proof}

Consequently, over the class of instances satisfying \(\max_i S_i/p_i\leq\kappa\), the total social welfare PoA is at most \(2+\kappa\), and this bound is tight; without a uniform bound on \(\max_i S_i/p_i\), the total social welfare PoA over the unrestricted class of instances is unbounded.

\subsection{Extension to Asymmetric Round-Trip Travel Times}
\label{app:asymmetric-travel-times}

We show that our results extend to asymmetric travel times on the two legs of
a minibus trip. Let $l_i^{\mathrm{out}}>0$ and $l_i^{\mathrm{ret}}>0$ denote
the outbound and return travel times on route $i$, respectively, and define
the round-trip travel time $r_i := l_i^{\mathrm{out}}+l_i^{\mathrm{ret}}.$ The service rate on route $i$ is then $\mu_i(x_i)=\frac{x_iF}{r_i},$ while the minibus cost absent queuing and schedule delays becomes $c_i^M=\bar p_i+\eta_T l_i^{\mathrm{out}},$ and hence $S_i=c_i^O-c_i^M$.

Replacing $2l_i$ by $r_i$ in the proof of Proposition~\ref{prop:rider-demand-func-queuing} yields the threshold $\tilde{k}_i^* = \min\left\{\frac{r_i\lambda_i}{F}, \frac{r_i\Lambda_i\eta_E\eta_L}{F S_i(\eta_E+\eta_L)} \right\},$ and the rider demand function
{\setlength{\abovedisplayskip}{1pt}
\setlength{\belowdisplayskip}{1pt}
\setlength{\jot}{1pt}
\begin{align*}
\Lambda_i^M(x_i)=
\begin{cases}
\displaystyle
\frac{F}{r_i}
\left(
t_2 - t_1+
S_i\frac{\eta_E+\eta_L}{\eta_E\eta_L}
\right)x_i
-
\left(\frac{F}{r_i}\right)^2
S_i\frac{\eta_E+\eta_L}{\eta_E\eta_L}
\frac{t_2-t_1}{\Lambda_i}x_i^2,
& x_i<\tilde{k}_i^*,\\[2ex]
\Lambda_i,
& x_i\geq\tilde{k}_i^*.
\end{cases}
\end{align*}}

Thus, asymmetric travel times only modify the route-specific coefficients in the demand function, while leaving its functional form unchanged. In particular, the continuity, monotonicity, and concavity properties underlying Corollaries~\ref{cor:monotonicity-rider-demands} and \ref{cor:monotonicity-per-driver-profits} continue to hold. Consequently, the subsequent equilibrium analysis and results apply without modification, including the PoA bounds, which do not depend on the round-trip travel time parameter. The remaining results similarly continue to hold by replacing $2l_i$ by $r_i$ wherever it represents round-trip service time and $l_i$ by $l_i^{\mathrm{out}}$ wherever it represents riders' in-vehicle travel time.

\subsection{Incorporating Physical Road Congestion} \label{apdx:road-congestion}

This section presents an extension of our informal transit model that incorporates physical road congestion. The key modeling choice is that congestion should depend on the actual minibus departure flow, not on the total number of drivers assigned to a route, as extra assigned drivers may wait in queue and only a fraction of the assigned drivers physically circulate in the system at any given time.

\subsubsection{Congestion-aware route model}

For route $i$, let $x_i$ be the mass of drivers assigned to the route and let $f_i$ denote the number of minibuses departing on the route per unit time. Let the one-way travel time be $\tau_i(f_i)$, where $\tau_i$ is positive, continuous, and strictly increasing. Sustaining a departure rate $f_i$ requires $g_i(f_i):=2\tau_i(f_i)f_i$ drivers to be actively circulating on the route. Hence, a departure rate $f_i$ is feasible under assignment $x_i$ only if $g_i(f_i)\leq x_i$, with the remaining $x_i-g_i(f_i)$ drivers waiting in a queue at the origin. Since $g_i$ is strictly increasing, its inverse $\phi_i(x_i):=g_i^{-1}(x_i)$ is well-defined and increasing. Thus, $\phi_i(x_i)$ gives the maximum minibus departure rate that can be sustained by $x_i$ assigned drivers, with corresponding rider service rate \(F\phi_i(x_i)\). This generalizes and is the congestion-aware analogue of the service rate $\frac{x_iF}{2l_i}$ in our baseline model with fixed travel times.


Next, let $\Delta=t_2-t_1$, $A=\frac{\eta_E+\eta_L}{\eta_E\eta_L},$ and define the net travel-cost advantage of the minibus at departure flow $f$ by $S_i(f)=c_i^O-\bar p_i-\eta_T\tau_i(f).$

For a fixed feasible departure flow $f$, define the induced rider demand as
\[
D_i(f)=
\begin{cases}
0, & S_i(f)<0,\\[3pt]
\Lambda_i, & S_i(f)\ge0,\ Ff\ge \lambda_i,\\[3pt]
\min\left\{\Lambda_i,
Ff\left[\Delta+A S_i(f)\left(1-\frac{Ff}{\lambda_i}\right)\right]\right\},
& S_i(f)\ge0,\ Ff<\lambda_i .
\end{cases}
\]
Because minibuses depart only when sufficiently loaded, not all drivers assigned to a route need to circulate simultaneously; excess drivers may instead wait at the origin. Thus, an assignment $x_i$ can support any departure rate $f_i\leq \phi_i(x_i)$. We define the induced route-level rider demand as the maximum demand attainable over these feasible departure rates:
\[
    \boxed{
    \Lambda_i^M(x_i)=\max_{0\le f\le \phi_i(x_i)}D_i(f).
    }
\]
Equivalently, writing $h=g_i(f)$ for active vehicles required and $d_i(h):=D_i(g_i^{-1}(h)),$ we have $\Lambda_i^M(x)=M_i(x):=\max_{0\le h\le x}d_i(h).$

\subsubsection{Key monotonicity properties}

\begin{lemma}[Demand per active vehicle is non-increasing]
The function $d_i(h)/h$ is non-increasing in $h>0$.
\end{lemma}

\begin{proof}
Because $g_i(f)$ is increasing, it suffices to show that $D_i(f)/g_i(f)$ is non-increasing in $f$.

In the uncapped, under-supplied region,
\[
\frac{D_i(f)}{g_i(f)}
=
\frac{F}{2\tau_i(f)}
\left[\Delta+A S_i(f)\left(1-\frac{Ff}{\lambda_i}\right)\right].
\]
Here $1/\tau_i(f)$ is non-increasing, $S_i(f)$ is non-increasing, and $1-Ff/\lambda_i$ is non-increasing and nonnegative. Hence the bracketed term is non-increasing, so the ratio is non-increasing.

When the demand cap binds, or when $Ff\ge\lambda_i$, we have $\frac{D_i(f)}{g_i(f)}=\frac{\Lambda_i}{g_i(f)},$ which is non-increasing. If $S_i(f)$ becomes negative, demand drops to zero. Across regime boundaries the ratio either matches continuously or jumps downward. Thus $D_i(f)/g_i(f)$, and therefore $d_i(h)/h$, is non-increasing. 
\end{proof}

\begin{lemma}[Running maximum preserves average monotonicity]
If $d_i(h)/h$ is non-increasing, then $\frac{M_i(x)}{x} = \frac{1}{x}\max_{0\le h\le x}d_i(h)$ is non-increasing in $x>0$.
\end{lemma}

\begin{proof}
Fix $0<a<b$. Let $h_b\in[0,b]$ attain $M_i(b)=d_i(h_b)$. If $M_i(b)=0$, the claim is immediate, so suppose $h_b>0$.

If $h_b\le a$, then the same choice was feasible under $a$, so $d_i(h_b)\le M_i(a)$. Hence
\[
    \frac{M_i(b)}{b}=\frac{d_i(h_b)}{b}
    \le \frac{M_i(a)}{b}
    \le \frac{M_i(a)}{a}.
\]

If $h_b>a$, then
\[
    \frac{M_i(b)}{b}
    =\frac{d_i(h_b)}{b}
    =\frac{h_b}{b}\frac{d_i(h_b)}{h_b}
    \le \frac{d_i(h_b)}{h_b}
    \le \frac{d_i(a)}{a}
    \le \frac{M_i(a)}{a},
\]
where the middle inequality uses that $d_i(h)/h$ is non-increasing. Thus $M_i(b)/b\le M_i(a)/a$ for all $0<a<b$. 
\end{proof}

\begin{proposition}[Congestion-aware analogues of Corollaries 1--2]
Under the active-flow congestion model: (i) $\Lambda_i^M(x_i)$ is non-decreasing in $x_i$, and (ii) if per-rider operating profit $p_i>0$ is constant, then per-driver payoff $\pi_i(x_i)=p_i\frac{\Lambda_i^M(x_i)}{x_i}$ is non-increasing in $x_i$.
\end{proposition}

\begin{proof}
For (i), increasing $x_i$ expands the feasible set $[0,\phi_i(x_i)]$, so the maximum cannot decrease. For (ii), use $\Lambda_i^M(x_i)=M_i(x_i)$ and the preceding lemma.
\end{proof}

\subsubsection{Implications for Section 4}

Let the feasible allocations satisfy $x_i\ge0$ and $\sum_i x_i=D$. Define total operator profit and rider throughput as $P(x)=\sum_i p_i\Lambda_i^M(x_i)$ and $R(x)=\sum_i \Lambda_i^M(x_i).$

A Wardrop equilibrium is an allocation $x^{Eq}$ such that any route with positive assigned drivers has weakly maximal per-driver payoff:
\begin{align*}
    x_i^{Eq}>0
    \quad\Rightarrow\quad
    \pi_i(x_i^{Eq})\ge \pi_j(x_j^{Eq})
    \quad\text{for all }j.
\end{align*}

\begin{remark}[Section 4 PoA Upper Bounds Survive]
The profit PoA bound remains $P\text{-PoA}\le2$, while for any fixed instance $I$ with $p_{\min}(I)>0$, the rider throughput ratio satisfies $\frac{R_I(\mathbf{x}^*)}{R_I(\mathbf{x}^{Eq})}
\leq 1 + \frac{p_{\max}(I)}{p_{\min}(I)}$. 
The proofs of Theorems~\ref{thm:poa-driver-profit} and~\ref{thm:poa-demand-served} carry over because they rely on the two route-level monotonicity properties established above: $\Lambda_i^M(x_i)$ is non-decreasing and $\pi_i(x_i)$ is non-increasing.
\end{remark}

\color{black}

\subsection{Limited Gains from Route Switching} \label{apdx:route-switching}

We show that, once all riders on a route have been served, the potential gains to a driver from switching routes are bounded and unlikely to outweigh the associated switching costs, such as the travel time required to move between geographically separated routes.

To this end, we first note from Corollary~\ref{cor:service-time} that the total service time $T_i$ on any route under our studied framework is bounded above by $\Delta + S_i \left( \frac{\eta_E + \eta_L}{\eta_E \eta_L} \right)$. Since $T_i \geq \Delta$, letting $\Bar{S} = \max_i S_i$, it follows for any two routes $i, i'$ that the difference in the service times satisfies $|T_i - T_{i'}| \leq \Bar{S} \left( \frac{\eta_E + \eta_L}{\eta_E \eta_L} \right)$.

We now show that there exists a bounded switching cost under which drivers would not change routes if $T_i < T_{i'}$. In this case, we can define the driver profit as the sum of three terms: (i) Driver profit on route $i$, (ii) Driver profit on route $i'$ during the period $T_{i'} - T_{i}$, and (iii) the negative of the cost to switch routes. Thus, for the drivers to not switch routes, we just require the switching cost to be high enough to cancel out the driver profit on route $i'$ during the period $T_{i'} - T_{i}$.

Now, the maximum per-driver profit on any route $i'$ is bounded above by
{\setlength{\abovedisplayskip}{4pt}
\setlength{\belowdisplayskip}{4pt}
\setlength{\jot}{1pt}
\begin{align*}
    p_{i'} (\Lambda_{i'}^M)'(0) = \frac{p_{i'} F}{2 l_{i'}} \left( \Delta + S_{i'} \left( \frac{\eta_E + \eta_L}{\eta_E \eta_L} \right) \right) \leq \frac{p_{\max} F}{2 \underline{l}} \left( \Delta + \Bar{S} \left( \frac{\eta_E + \eta_L}{\eta_E \eta_L} \right) \right),
\end{align*}}
as the rider demand function $\Lambda_i^M$ has a decreasing slope. Here, we take $p_{\max} = \max_i p_i$ and let $\underline{l} = \min_i l_i > 0$. Thus, for a bounded switching cost that is at least $\frac{p_{\max} F}{2 \underline{l}} \left( \Delta + \Bar{S} \left( \frac{\eta_E + \eta_L}{\eta_E \eta_L} \right) \right)$, drivers will not seek to change routes even if $T_i < T_{i'}$.

\subsection{Additional Properties of Budget-Balanced Cross-Subsidization Scheme} \label{apdx:generality-thm-cross-subsidies}

Beyond the algorithmic and computational properties highlighted in Section~\ref{sec:cross-subsidies}, we emphasize the generality of Theorem~\ref{thm:cross-subsidies-optimal} by highlighting several other properties of the optimal budget-balanced cross-subsidization scheme. 

First, if $\x^*$ is a solution to a system optimization problem beyond cumulative driver profit or rider throughput maximization that can only be solved approximately, up to an approximation factor $\beta$, then implementing this allocation via cross-subsidization results in an equilibrium that is likewise $\beta$-optimal. \color{blue} Thus, cross-subsidization preserves approximation guarantees at the target allocation when translating approximately optimal driver allocations into equilibrium outcomes. \color{black}



Next, since Theorem~\ref{thm:cross-subsidies-optimal} provides a method to implement any feasible target driver allocation $\x^*$ with at least one strictly positive entry as an equilibrium, it applies not only to allocations optimizing different system objectives but also to allocations satisfying additional constraints, such as upper or lower bounds on driver supply across routes. 

Furthermore, while the proof of Theorem~\ref{thm:cross-subsidies-optimal} analyzed the setting when the equilibrium condition in Equation~\eqref{eq:eq-condition-pf-cross-subsidy} holds with equality for all routes, it can be directly extended to settings when this equilibrium condition holds with a strict inequality for routes where the target allocation satisfies $x_j^* = 0$. 

Finally, while our equilibrium condition assumes that drivers choose routes based solely on relative per-driver profits, selecting the route that maximizes earnings without imposing requirements on the absolute level of profits, our framework can be readily generalized to incorporate individual rationality constraints when drivers have an outside option with reservation wage $W$. In particular, if a target allocation $\x^*$ satisfies individual rationality for drivers, so that $\pi_j^* \geq W$ for all routes $j$ with $x_j^*>0$, then this condition is preserved under cross-subsidization: 
{\setlength{\abovedisplayskip}{4pt}
\setlength{\belowdisplayskip}{4pt}
\setlength{\jot}{1pt}
\begin{align*}
    \Tilde{\pi}^{Eq} = \frac{\sum_{i \in [n]} \pi_i^* x_i^*}{\sum_{i \in [n]} x_i^*} \geq \frac{W \sum_{i \in [n]} x_i^*}{\sum_{i \in [n]} x_i^*} = W.
\end{align*}}
\color{blue} Thus, drivers at the target allocation are never made worse off under cross-subsidization relative to their reservation wage. \color{black}

\color{blue}

\subsection{Uniqueness of Equilibrium Driver Allocation} \label{apdx:pf-uniqueness}

Consider any fixed cross-subsidy scheme $\ttau$ (including $\ttau=\mathbf{0}$), and let $\Tilde{\pi}_i(x_i)=\pi_i(x_i)+\tau_i$ denote the transfer-adjusted per-driver profit. From Equation~\eqref{eq:rider-demand-function}, $\Tilde{\pi}_i(x_i)$ is non-increasing in $x_i$. Moreover, it is strictly decreasing when
$S_i>0$ and $p_i>0$. 

Now, consider any two equilibrium allocations $\x^{(1)}$ and $\x^{(2)}$. From the classical variational inequality characterization of equilibria it follows that
\begin{align*}
0 &\leq \sum_{i\in[n]}
\left(
\Tilde{\pi}_i(x_i^{(1)})
-\Tilde{\pi}_i(x_i^{(2)})
\right)
\left(x_i^{(1)}-x_i^{(2)}\right).
\end{align*}
However, if $S_i>0$ and $p_i>0$ for all $i\in[n]$, strict monotonicity implies that every summand is non-positive and strictly negative whenever $x_i^{(1)}\neq x_i^{(2)}$. This yields a contradiction unless $\x^{(1)}=\x^{(2)}$, establishing uniqueness under these conditions.

\color{black}

\section{Additional Details on Numerical Experiments} \label{apdx:additional-experiments}

\subsection{Details on Route Selection} \label{apdx:route-selection}

Our original dataset on the informal transit system in Nalasopara contains information on twenty-one routes. However, some of these routes serve destinations in the same geographic area, for which only a single destination coordinate is available. In these cases, we aggregate routes by constructing weighted averages of their route characteristics, including trip time, distance, and fare, with weights proportional to the recorded route volumes during our study period (5 p.m.-midnight). In addition, geographic coordinates are unavailable for one route, which are required for demand estimation. Thus, for our experiments, we focus on a final set of eighteen routes.

\subsection{Rider Demand Calibration Procedure} \label{appendix:demand_calibration}

We calibrate the demand $\Lambda_i$ for each route $i$ using population statistics from the 2011 Indian Census, the most recent census conducted in India, which reports data at the level of administrative sub-divisions referred to as \emph{wards}. To this end, we first assign each route $i$ to a ward based on the geographic coordinates of its destination stop. Since we study the evening commute period in our experiments, we take the Nalasopara railway station to be the common origin of all routes with the destination being the respective residential neighborhoods, reflecting that the majority of trips during the study period are home-bound.

Next, we estimate the population that can feasibly access each route  by defining a pedestrian catchment area around its destination stop based on walking accessibility. Specifically, for routes with total length of at most 5 km, we assume users are willing to walk five minutes to access the route, while for longer routes we assume a walking radius of ten minutes. Using a walking speed of 1.3 m/s, we obtain circular catchment areas around each destination stop. 

However, since these circular catchment areas may overlap across destination stops, we then generate disjoint catchments by assigning users to their nearest stop, ensuring that each user is associated with at most one route. To do so, we use Monte Carlo nearest-neighbor sampling. Specifically, for each circular catchment area, we repeatedly draw a large number of points uniformly at random, interpreting each draw as a representative residential location, and determine the fraction of those points that are geographically closest (in Euclidean distance) to that catchment area's destination stop. The fraction of points assigned to a given destination stop determines the share of the circular catchment area that can access that stop. Assuming that population density within the circular catchment area is equal to the ward-level population density of the ward to which the destination stop belongs, we compute the total population that can access each route by multiplying the total population within the circular catchment area by this fraction.

Finally, we assume that 40\% of the population commute daily by train in our baseline calibration, informed by the substantial share of suburban rail in motorized trips in the Mumbai Metropolitan Region (40.6\% in 2014)~\citep{mcgm2016cmp}, which we use as an approximate benchmark. Multiplying this commuting share by the above estimated population that can access each route yields the total mass of users $\Lambda_i$ seeking to make trips on that route during our study period. We further examine the sensitivity of our results to alternative rider demands in Appendix~\ref{apdx:demand-sensitivity}.


\subsection{Implementation Details to Compute Equilibrium Driver Allocation} \label{apdx:implementation-details-eq-allocation}

We now describe a procedure to compute an approximate equilibrium driver allocation, where we leverage the fact that, in equilibrium, all routes with a strictly positive mass of drivers have the same per-driver profit. Specifically, we search over a discretized set of possible equilibrium profit levels within a bounded interval $[\underline{\pi}, \Bar{\pi}]$, where $\underline{\pi}, \Bar{\pi}$ denote lower and upper bounds on the equilibrium per-driver profits. For each candidate per-driver profit level on the discretized grid, we compute the corresponding allocation of drivers across routes that would result in per-driver profits at that level and sum these allocations across the routes to obtain the total number of drivers willing to operate at that per-driver profit level. Then, we obtain an approximate equilibrium per-driver profit as the value on the discretized grid for which the resulting aggregate driver allocation is closest to the actual mass of drivers $D$, with the corresponding route-level allocation of drivers across routes taken as the approximate equilibrium driver allocation. \color{blue} Note that due to the discretization, the resulting aggregate driver allocation need not equal $D$ exactly. \color{black}

\subsection{Additional Numerical Results} \label{apdx:additional-numerical-results}

\subsubsection{Sensitivity to Rider Demand} \label{apdx:demand-sensitivity}

Figure~\ref{fig:demand-robustness} examines the sensitivity of the profit and rider throughput ratios as well as the equilibrium per-driver profits to variations in overall rider demand by scaling the calibrated route-level demands $\Lambda_i$ by up to $\pm20\%$.

\begin{figure}[tbh!]
    \centering
    \include{Fig/msom/revision/figure3_robustness}
    \vspace{-40pt}
    \caption{\small \sf Sensitivity of cumulative driver profit, rider throughput, and equilibrium per-driver profit to variations in rider demand across different driver supply levels. Solid lines represent the results under the baseline demand calibration, while shaded bands indicate the range of outcomes under uniform demand variations of $\pm 20\%$ across all routes. Profit and rider throughput ratios are computed as the corresponding centralized optimum divided by the decentralized equilibrium outcome.
}
    \label{fig:demand-robustness}
\end{figure}

\subsubsection{Total Social Welfare} \label{apdx:social-welfare-experiments}

Figure~\ref{fig:social-welfare-ratios} extends the fixed-fare analysis in Section~\ref{sec:experiments} to total social welfare, by examining the total social welfare ratio across different driver supply levels, where total social welfare combines cumulative driver profit with riders' generalized-cost savings relative to the outside option. For comparison, we also include the profit and rider throughput ratios from Figure~\ref{fig:welfare-profit-ratios}.

\begin{figure}[tbh!]
    \centering
    \include{Fig/msom/revision/fixed_fare_social_welfare_ratio}
    \vspace{-40pt}
    \caption{\small \sf Depiction of the profit ratio (left), rider throughput ratio (center), and total social welfare ratio (right) as the number of drivers in the system is varied.
}
    \label{fig:social-welfare-ratios}
\end{figure}

\end{document}

%% file: tex/macros2.tex
\newcommand{\x}{\mathbf{x}}
\newcommand{\y}{\mathbf{y}}
\newcommand{\p}{\mathbf{p}}
\newcommand{\ttau}{\boldsymbol{\tau}}
\newcommand{\LLambda}{\boldsymbol{\Lambda}}

%% file: Fig/msom/PoABounds/profit_ratio.tex
\definecolor{mycolor1}{rgb}{0.00000,0.44700,0.74100}%
\definecolor{mycolor2}{rgb}{0.85000,0.32500,0.09800}%
\definecolor{mycolor3}{rgb}{0.46600,0.67400,0.18800}%
\definecolor{mycolor4}{rgb}{0.49400,0.18400,0.55600}%

\begin{tikzpicture}

\begin{axis}[%
width=1.3in,
height=0.9in,
at={(0in,0in)},
scale only axis,
xmin=0,
xmax=2100,
ymin=0.98,
ymax=1.32,
xlabel style={font=\color{white!15!black}, font=\footnotesize, yshift=0.1cm},
xlabel={Number of Drivers ($D$)},
ylabel style={font=\color{white!15!black}, font=\footnotesize},
ylabel={Profit Ratio $\frac{P(\x^*)}{P(\x^{Eq})}$},
ytick={1.0,1.1,1.2,1.3},
yticklabels={1.0,1.1,1.2,1.3},
ylabel style={yshift=-3pt},
title style={font=\footnotesize, yshift=-0.27cm},
tick label style={font=\scriptsize, yshift=0.07cm},
axis background/.style={fill=white},
grid=both,
grid style={dashed,gray!30},
legend style={
    at={(0.5,0.22)},
    anchor=east,
    legend cell align=left,
    align=left,
    draw=white!0!black,
    inner xsep=1pt,
    inner ysep=0pt,
    font=\tiny,
    row sep=-0.05cm
}
]

\addplot [color=mycolor2, line width=1.0pt, mark=square*, mark size=0.75pt,
    mark options={line width=1pt}]
coordinates {(100,1.005830293) (150,1.205137959) (200,1.196128174) (250,1.156200051) (300,1.159398161) (350,1.167122406) (400,1.131478669) (450,1.176727412) (500,1.170520136) (550,1.194390557) (600,1.219382642) (650,1.221649444) (700,1.196229125) (750,1.160012236) (800,1.2219647) (850,1.181999163) (900,1.177069738) (950,1.224969745) (1000,1.197103209) (1050,1.220059758) (1100,1.260199606) (1150,1.297340653) (1200,1.274090435) (1250,1.202198079) (1300,1.14524453) (1350,1.14524453) (1400,1.14524453) (1450,1.097944305) (1500,1.063715777) (1550,1.063715777) (1600,1.057014159) (1650,1.028682674) (1700,1.006217004) (1750,1.006217004) (1800,1.006217004) (1850,1.006217004) (1900,1.006217004) (1950,1.006217004) (2000,1.006217004)};

\end{axis}

\begin{axis}[%
width=0in,
height=0in,
at={(0in,0in)},
scale only axis,
xmin=0,
xmax=1,
ymin=0,
ymax=1,
axis line style={draw=none},
ticks=none,
axis x line*=bottom,
axis y line*=left
]
\end{axis}

\end{tikzpicture}%

%% file: Fig/msom/PoABounds/welfare_ratio.tex
\definecolor{mycolor1}{rgb}{0.00000,0.44700,0.74100}%
\definecolor{mycolor2}{rgb}{0.85000,0.32500,0.09800}%
\definecolor{mycolor3}{rgb}{0.46600,0.67400,0.18800}%
\definecolor{mycolor4}{rgb}{0.49400,0.18400,0.55600}%

\begin{tikzpicture}

\begin{axis}[%
width=1.3in,
height=0.9in,
at={(0in,0in)},
scale only axis,
xmin=0,
xmax=2100,
ymin=0.98,
ymax=1.45,
xlabel style={font=\color{white!15!black}, font=\footnotesize, yshift=0.1cm},
xlabel={Number of Drivers ($D$)},
ylabel style={font=\color{white!15!black}, font=\footnotesize},
ylabel={\shortstack{Rider Throughput \\ Ratio $\frac{R(\x^*)}{R(\x^{Eq})}$}},
ytick={1.0,1.1,1.2,1.3, 1.4},
yticklabels={1.0,1.1,1.2,1.3, 1.4},
ylabel style={yshift=-3pt},
title style={font=\footnotesize, yshift=-0.27cm},
tick label style={font=\scriptsize, yshift=0.07cm},
axis background/.style={fill=white},
grid=both,
grid style={dashed,gray!30},
legend style={
    at={(0.5,0.22)},
    anchor=east,
    legend cell align=left,
    align=left,
    draw=white!0!black,
    inner xsep=1pt,
    inner ysep=0pt,
    font=\tiny,
    row sep=-0.05cm
}
]

\addplot [color=mycolor4, line width=1.0pt, mark=*, mark size=0.75pt,
    mark options={line width=1pt}]
coordinates {(100,1.19516321) (150,1.376250487) (200,1.400792205) (250,1.301806409) (300,1.198098696) (350,1.161427095) (400,1.137089972) (450,1.182421951) (500,1.181786517) (550,1.188105035) (600,1.21555758) (650,1.222139535) (700,1.204400884) (750,1.177502612) (800,1.240168452) (850,1.203647851) (900,1.202668069) (950,1.245929151) (1000,1.201223949) (1050,1.21749426) (1100,1.252705344) (1150,1.283899964) (1200,1.262711494) (1250,1.199854502) (1300,1.14942661) (1350,1.14942661) (1400,1.14942661) (1450,1.093693323) (1500,1.053880549) (1550,1.053880549) (1600,1.048061315) (1650,1.023359078) (1700,1.003654238) (1750,1.003654238) (1800,1.003654238) (1850,1.003654238) (1900,1.003654238) (1950,1.003654238) (2000,1.003654238)};

\end{axis}

\begin{axis}[%
width=0in,
height=0in,
at={(0in,0in)},
scale only axis,
xmin=0,
xmax=1,
ymin=0,
ymax=1,
axis line style={draw=none},
ticks=none,
axis x line*=bottom,
axis y line*=left
]
\end{axis}

\end{tikzpicture}%

%% file: Fig/msom/PoABounds/profit_per_driver.tex
\definecolor{mycolor1}{rgb}{0.00000,0.44700,0.74100}%
\definecolor{mycolor2}{rgb}{0.85000,0.32500,0.09800}%
\definecolor{mycolor3}{rgb}{0.46600,0.67400,0.18800}%
\definecolor{mycolor4}{rgb}{0.49400,0.18400,0.55600}%

\begin{tikzpicture}

\begin{axis}[%
width=1.3in,
height=0.9in,
at={(0in,0in)},
scale only axis,
xmin=0,
xmax=2100,
ymin=100,
ymax=500,
xlabel style={font=\color{white!15!black}, font=\footnotesize, yshift=0.1cm},
xlabel={Number of Drivers ($D$)},
ylabel style={font=\color{white!15!black}, font=\footnotesize},
ylabel={\shortstack{Equilibrium Profit \\ per Driver (Rs.)}},
ylabel style={yshift=-3pt},
title style={font=\footnotesize, yshift=-0.27cm},
tick label style={font=\scriptsize, yshift=0.07cm},
axis background/.style={fill=white},
grid=both,
grid style={dashed,gray!30},
legend style={
    at={(0.5,0.22)},
    anchor=east,
    legend cell align=left,
    align=left,
    draw=white!0!black,
    inner xsep=1pt,
    inner ysep=0pt,
    font=\tiny,
    row sep=-0.05cm
}
]

\addplot [color=mycolor3, line width=1.0pt, mark=triangle*, mark size=0.75pt,
    mark options={line width=1pt}]
coordinates {
 (100,  435.0)
(150,  380.0)
(200,  370.0)
(250,  363.0)
(300,  355.0)
(350,  352.0)
(400,  335.0)
(450,  327.0)
(500,  311.0)
(550,  301.0)
(600,  300.0)
(650,  299.0)
(700,  297.0)
(750,  282.0)
(800,  280.0)
(850,  278.0)
(900,  263.0)
(950,  261.0)
(1000, 255.0)
(1050, 243.0)
(1100, 232.0)
(1150, 230.0)
(1200, 229.0)
(1250, 228.0)
(1300, 223.0)
(1350, 218.0)
(1400, 217.0)
(1450, 215.0)
(1500, 208.0)
(1550, 205.0)
(1600, 204.0)
(1650, 200.0)
(1700, 194.0)
(1750, 189.0)
(1800, 183.0)
(1850, 178.0)
(1900, 174.0)
(1950, 169.0)
(2000, 165.0)
};

\end{axis}

\begin{axis}[%
width=0in,
height=0in,
at={(0in,0in)},
scale only axis,
xmin=0,
xmax=1,
ymin=0,
ymax=1,
axis line style={draw=none},
ticks=none,
axis x line*=bottom,
axis y line*=left
]
\end{axis}

\end{tikzpicture}%

%% file: Fig/msom/revision/figure4_normalized.tex
\begin{tikzpicture}
\begin{groupplot}[
 group style={group size=2 by 1, horizontal sep=2.5cm},
 width=0.38\textwidth, height=0.2\textwidth,
 xlabel={Weight on Cumulative Driver Profit ($\alpha$)},
 xlabel style={font=\color{white!15!black}, font=\footnotesize, yshift=0.1cm},
 xmin=0, xmax=1,
 xtick={0,0.2,0.4,0.6,0.8,1.0},
 xticklabels={0,0.2,0.4,0.6,0.8,1.0},
 grid=both,
 grid style={line width=.1pt, draw=gray!20},
 major grid style={line width=.2pt, draw=gray!35},
 legend style={
   at={(0.98,0.98)},
   anchor=north east,
   draw=none, fill=white, font=\small
 },
 tick label style={font=\scriptsize, yshift=0.07cm},
 label style={font=\small}
]
\nextgroupplot[ylabel={\shortstack{Profit Ratio \\ $\frac{P(\x^*_1)}{P(\x_{\alpha})}$}},ylabel style={font=\color{white!15!black}, font=\footnotesize},ymin=0.9]
\addplot[blue!75!black, solid, thick, mark=*] coordinates {(0,2.778985953) (0.1,1.001164754) (0.2,1.000806492) (0.3,1.000806492) (0.4,1.000806492) (0.5,1.000806492) (0.6,1.000401793) (0.7,1.000401793) (0.8,1.000401793) (0.9,1) (1,1)};
\addlegendentry{Centralized}
\addplot[orange!85!black, dashed, thick, mark=square*] coordinates {(0,2.778985953) (0.1,1.427526377) (0.2,1.427526377) (0.3,1.427526364) (0.4,1.427526369) (0.5,1.42752637) (0.6,1.209625002) (0.7,1.209625002) (0.8,1.209625002) (0.9,1.209625002) (1,1.209625002)};
\addlegendentry{Equilibrium}
\nextgroupplot[ylabel={\shortstack{Throughput \\ Ratio $\frac{R(\x^*_0)}{R(\x_{\alpha})}$}},ylabel style={font=\color{white!15!black}, font=\footnotesize},ymin=0.9]
\addplot[blue!75!black, solid, thick, mark=*] coordinates {(0,1) (0.1,1.000108878) (0.2,1.000163447) (0.3,1.000163447) (0.4,1.000163447) (0.5,1.000163447) (0.6,1.000654613) (0.7,1.000654613) (0.8,1.000654613) (0.9,1.002847928) (1,1.002847928)};
\addplot[orange!85!black, dashed, thick, mark=square*] coordinates {(0,1) (0.1,1.000071162) (0.2,1.000071162) (0.3,1.000071162) (0.4,1.000071162) (0.5,1.000071162) (0.6,1.205372428) (0.7,1.205372428) (0.8,1.205372428) (0.9,1.205372428) (1,1.205372428)};
\end{groupplot}
\end{tikzpicture}

%% file: Fig/msom/priceOpt/reservation_wage_profit.tex
\begin{tikzpicture}
\begin{axis}[
    width=1\textwidth,
    height=0.5\textwidth,
    xlabel={Reservation Wage (Rs. per driver)},
    ylabel={\shortstack{Cumulative \\ Profit (Rs.)}},
    xmin=0, xmax=800,
    xtick={0,200,400,600,800},
    xlabel style={font=\color{white!15!black}, font=\footnotesize, yshift=0.12cm},
    yticklabel style={
        /pgf/number format/fixed,
        /pgf/number format/precision=0
    },
    grid=both,
    grid style={line width=.1pt, draw=gray!20},
    major grid style={line width=.2pt, draw=gray!35},
    tick label style={font=\scriptsize, yshift=0.07cm},
    label style={font=\small},
]

\addplot[
    thick,
    mark=*,
    mark size=1.8pt
] coordinates {
    (0.00, 0.000000)
    (50.00, 50000.000000)
    (100.00, 100000.000000)
    (150.00, 150000.000000)
    (200.00, 200000.000000)
    (250.00, 250000.000000)
    (300.00, 300000.000000)
    (350.00, 350000.000000)
    (400.00, 400000.000000)
    (450.00, 450000.000000)
    (500.00, 480000.000000)
    (550.00, 440000.000000)
    (600.00, 438000.000000)
    (650.00, 396500.000000)
    (700.00, 280000.000000)
    (750.00, 240000.000000)
    (800.00, 192000.000000)
};

\end{axis}
\end{tikzpicture}

%% file: Fig/msom/priceOpt/reservation_wage_welfare.tex
\begin{tikzpicture}
\begin{axis}[
    width=1\textwidth,
    height=0.5\textwidth,
    xlabel={Reservation Wage (Rs. per driver)},
    xlabel style={font=\color{white!15!black}, font=\footnotesize, yshift=0.12cm},
    ylabel={\shortstack{Rider \\ Throughput}},
    xmin=0, xmax=800,
    xtick={0,200,400,600,800},
    yticklabel style={
        /pgf/number format/fixed,
        /pgf/number format/precision=0
    },
    grid=both,
    grid style={line width=.1pt, draw=gray!20},
    major grid style={line width=.2pt, draw=gray!35},
    tick label style={font=\scriptsize, yshift=0.07cm},
    label style={font=\small},
]

\addplot[
    thick,
    mark=*,
    mark size=1.8pt
] coordinates {
    (0.00, 87433.349555)
    (50.00, 87429.099177)
    (100.00, 87424.844169)
    (150.00, 87420.584493)
    (200.00, 87416.320111)
    (250.00, 87412.050984)
    (300.00, 87407.777073)
    (350.00, 87403.622244)
    (400.00, 87399.977953)
    (450.00, 87396.331815)
    (500.00, 81592.946939)
    (550.00, 72888.153554)
    (600.00, 66807.905806)
    (650.00, 58740.372090)
    (700.00, 44015.655585)
    (750.00, 36601.188608)
    (800.00, 29806.850906)
};

\end{axis}
\end{tikzpicture}

%% file: Fig/msom/revision/figure3_robustness.tex
\begin{tikzpicture}
\begin{groupplot}[
  group style={group size=3 by 1, horizontal sep=2.2cm},
  width=0.30\textwidth, height=0.25\textwidth,
  xlabel={Number of drivers ($D$)}, xlabel style={font=\color{white!15!black}, font=\footnotesize, yshift=0.1cm}, grid=major,
  tick label style={font=\scriptsize}
]
\nextgroupplot[ylabel style={font=\color{white!15!black}, font=\footnotesize},
ylabel={Profit Ratio $\frac{P(\x^*)}{P(\x^{Eq})}$}]
\addplot[name path=upper0, draw=none, forget plot] coordinates {(50,1.196760266) (100,1.246255289) (150,1.246255289) (200,1.190538186) (250,1.172327472) (300,1.172464102) (350,1.167844453) (400,1.194390557) (450,1.209635461) (500,1.211510514) (550,1.21284016) (600,1.213954138) (650,1.214900841) (700,1.202461915) (750,1.21487046) (800,1.220059758) (850,1.270130299) (900,1.282907914) (950,1.264616671) (1000,1.297871041) (1050,1.260199606) (1100,1.297340653) (1150,1.256889127) (1200,1.292593783) (1250,1.270938619) (1300,1.286619819) (1350,1.282907914) (1400,1.240165618) (1450,1.200384883) (1500,1.163269642) (1550,1.14524453) (1600,1.14336142) (1650,1.111716219) (1700,1.081946944) (1750,1.063715777) (1800,1.063715777) (1850,1.051924571) (1900,1.028295923) (1950,1.006853517) (2000,1.006217004)};
\addplot[name path=lower0, draw=none, forget plot] coordinates {(50,1) (100,1.150119875) (150,1.187586512) (200,1.159398161) (250,1.148459746) (300,1.139424307) (350,1.13294499) (400,1.139424307) (450,1.16017775) (500,1.15694871) (550,1.159471525) (600,1.175046737) (650,1.179752952) (700,1.164721263) (750,1.17062267) (800,1.175046737) (850,1.170412656) (900,1.184960725) (950,1.190164122) (1000,1.163269642) (1050,1.14524453) (1100,1.111716219) (1150,1.067714833) (1200,1.063715777) (1250,1.039949254) (1300,1.006853517) (1350,1.006217004) (1400,1.006217004) (1450,1.006217004) (1500,1.006217004) (1550,1.006217004) (1600,1.006217004) (1650,1.006217004) (1700,1.006217004) (1750,1.006217004) (1800,1.006217004) (1850,1.006217004) (1900,1.006217004) (1950,1.006217004) (2000,1.006217004)};
\addplot[fill=blue!20, draw=none, forget plot] fill between[of=upper0 and lower0];
\addplot[blue, very thick] coordinates {(50,1.005830293) (100,1.150119875) (150,1.207604711) (200,1.179029017) (250,1.159398161) (300,1.153583128) (350,1.13294499) (400,1.16422404) (450,1.162206627) (500,1.194390557) (550,1.21284016) (600,1.196107811) (650,1.179752952) (700,1.164721263) (750,1.211418445) (800,1.191552222) (850,1.177069738) (900,1.224969745) (950,1.211191043) (1000,1.220059758) (1050,1.260199606) (1100,1.297340653) (1150,1.256889127) (1200,1.208118909) (1250,1.163269642) (1300,1.14524453) (1350,1.130467118) (1400,1.093640256) (1450,1.063715777) (1500,1.063715777) (1550,1.047094977) (1600,1.019196655) (1650,1.006217004) (1700,1.006217004) (1750,1.006217004) (1800,1.006217004) (1850,1.006217004) (1900,1.006217004) (1950,1.006217004) (2000,1.006217004)};
\addplot[black!50, densely dotted, forget plot] coordinates {(50,1) (2000,1)};
\nextgroupplot[ylabel style={font=\color{white!15!black}, font=\footnotesize},
ylabel={\shortstack{Rider Throughput \\ Ratio $\frac{R(\x^*)}{R(\x^{Eq})}$}}]
\addplot[name path=upper1, draw=none, forget plot] coordinates {(50,1.448843129) (100,1.435756523) (150,1.435756523) (200,1.402699233) (250,1.299467895) (300,1.198098696) (350,1.180815977) (400,1.188105035) (450,1.208581523) (500,1.20942809) (550,1.209942957) (600,1.22932394) (650,1.210745563) (700,1.22447908) (750,1.22932394) (800,1.236404257) (850,1.261382143) (900,1.270360309) (950,1.256564691) (1000,1.283310267) (1050,1.252705344) (1100,1.283899964) (1150,1.249544136) (1200,1.281009286) (1250,1.259974233) (1300,1.275789622) (1350,1.270360309) (1400,1.233160256) (1450,1.198257735) (1500,1.165447542) (1550,1.14942661) (1600,1.147191678) (1650,1.109834191) (1700,1.075032432) (1750,1.053880549) (1800,1.053880549) (1850,1.043635764) (1900,1.023020733) (1950,1.004213962) (2000,1.003654238)};
\addplot[name path=lower1, draw=none, forget plot] coordinates {(50,1.19516321) (100,1.301070613) (150,1.399447602) (200,1.198098696) (250,1.146996159) (300,1.14192512) (350,1.138334309) (400,1.14192512) (450,1.170579057) (500,1.168557858) (550,1.1656195) (600,1.186984553) (650,1.189819895) (700,1.18173883) (750,1.184532785) (800,1.186984553) (850,1.187984587) (900,1.208249403) (950,1.214706366) (1000,1.165447542) (1050,1.14942661) (1100,1.109834191) (1150,1.058509856) (1200,1.053880549) (1250,1.03320205) (1300,1.004213962) (1350,1.003654238) (1400,1.003654238) (1450,1.003654238) (1500,1.003654238) (1550,1.003654238) (1600,1.003654238) (1650,1.003654238) (1700,1.003654238) (1750,1.003654238) (1800,1.003654238) (1850,1.003654238) (1900,1.003654238) (1950,1.003654238) (2000,1.003654238)};
\addplot[fill=blue!20, draw=none, forget plot] fill between[of=upper1 and lower1];
\addplot[blue, very thick] coordinates {(50,1.19516321) (100,1.301070613) (150,1.412936302) (200,1.340346244) (250,1.198098696) (300,1.150322194) (350,1.138334309) (400,1.171679309) (450,1.172159403) (500,1.188105035) (550,1.209942957) (600,1.19989464) (650,1.189819895) (700,1.18173883) (750,1.22932394) (800,1.213496543) (850,1.202668069) (900,1.245929151) (950,1.217769817) (1000,1.21749426) (1050,1.252705344) (1100,1.283899964) (1150,1.247752197) (1200,1.205064652) (1250,1.165447542) (1300,1.14942661) (1350,1.131924354) (1400,1.08866336) (1450,1.053880549) (1500,1.053880549) (1550,1.039431411) (1600,1.015051486) (1650,1.003654238) (1700,1.003654238) (1750,1.003654238) (1800,1.003654238) (1850,1.003654238) (1900,1.003654238) (1950,1.003654238) (2000,1.003654238)};
\addplot[black!50, densely dotted, forget plot] coordinates {(50,1) (2000,1)};
\nextgroupplot[ylabel style={font=\color{white!15!black}, font=\footnotesize},
ylabel={\shortstack{Equilibrium Profit \\ per Driver (Rs.)}}]
\addplot[name path=upper2, draw=none, forget plot] coordinates {(50,576.4994469) (100,435.9010315) (150,381.7375609) (200,379.0034144) (250,369.8395991) (300,362.9060995) (350,356.5676571) (400,352.9500654) (450,336.1586319) (500,333.7978805) (550,326.3038859) (600,311.1941172) (650,301.2211501) (700,300.0588372) (750,299.3707709) (800,298.3783652) (850,295.2774881) (900,281.3009718) (950,280.3485807) (1000,279.3965829) (1050,270.3698754) (1100,262.1140108) (1150,261.0964541) (1200,255.2660672) (1250,245.0554245) (1300,235.6302159) (1350,230.1143749) (1400,229.5436482) (1450,228.9731096) (1500,228.4027606) (1550,224.5138192) (1600,217.85598) (1650,217.2676746) (1700,216.6796454) (1750,214.096394) (1800,208.149272) (1850,204.7937385) (1900,203.9864347) (1950,202.9888108) (2000,198.0392871)};
\addplot[name path=lower2, draw=none, forget plot] coordinates {(50,456.9695422) (100,381.7375609) (150,370.6572596) (200,362.9060995) (250,354.104719) (300,336.1586319) (350,327.9405309) (400,311.1941172) (450,300.3882348) (500,299.3707709) (550,297.8822776) (600,281.3009718) (650,279.8725324) (700,270.3698754) (750,261.6051489) (800,255.2660672) (850,240.2504162) (900,230.1143749) (950,229.2583553) (1000,228.4027606) (1050,220.9501078) (1100,217.2676746) (1150,216.3857349) (1200,208.149272) (1250,204.3899712) (1300,202.9888108) (1350,195.5943576) (1400,188.6088448) (1450,182.1050916) (1500,176.0349218) (1550,170.356376) (1600,165.0327392) (1650,160.0317471) (1700,155.324931) (1750,150.8870759) (1800,146.6957682) (1850,142.7310177) (1900,138.9749383) (1950,135.4114783) (2000,132.0261914)};
\addplot[fill=blue!20, draw=none, forget plot] fill between[of=upper2 and lower2];
\addplot[blue, very thick] coordinates {(50,571.2119277) (100,434.6522195) (150,380.0969728) (200,370.1503479) (250,362.9060995) (300,354.7976787) (350,352.0265956) (400,334.7419278) (450,326.7950754) (500,311.1941172) (550,300.8620319) (600,299.7953333) (650,298.7752907) (700,297.5846622) (750,281.3009718) (800,280.1581496) (850,278.3219305) (900,262.859601) (950,261.2999119) (1000,255.2660672) (1050,243.1105402) (1100,232.0600611) (1150,229.7719164) (1200,229.0872022) (1250,228.4027606) (1300,223.0746281) (1350,217.6206248) (1400,216.9148238) (1450,215.3268331) (1500,208.149272) (1550,204.6322039) (1600,203.6637724) (1650,200.0396839) (1700,194.1561638) (1750,188.6088448) (1800,183.3697103) (1850,178.4137721) (1900,173.7186729) (1950,169.2643479) (2000,165.0327392)};
\end{groupplot}
\end{tikzpicture}

%% file: Fig/msom/revision/fixed_fare_social_welfare_ratio.tex
\begin{tikzpicture}
\begin{groupplot}[
group style={group size=3 by 1, horizontal sep=2cm},
width=0.29\textwidth, height=0.25\textwidth,
xlabel={Number of drivers ($D$)},
xlabel style={font=\footnotesize,yshift=0.1cm},
grid=both,
grid style={line width=.1pt,draw=gray!20},
major grid style={line width=.2pt,draw=gray!35},
tick label style={font=\scriptsize},
label style={font=\footnotesize}
]
\nextgroupplot[ylabel={Profit Ratio $\frac{P(\mathbf{x}_P^*)}{P(\mathbf{x}^{Eq})}$}]
\addplot[blue!75!black,line width=1.0pt, mark=square*, mark size=0.75pt,
    mark options={line width=1pt}] coordinates {(100,1.005830293) (150,1.205137959) (200,1.196128174) (250,1.156200051) (300,1.159398161) (350,1.167122406) (400,1.131478669) (450,1.176727412) (500,1.170520136) (550,1.194390557) (600,1.219382642) (650,1.221649444) (700,1.196229125) (750,1.160012236) (800,1.2219647) (850,1.181999163) (900,1.177069738) (950,1.224969745) (1000,1.197103209) (1050,1.220059758) (1100,1.260199606) (1150,1.297340653) (1200,1.274090435) (1250,1.202198079) (1300,1.14524453) (1350,1.14524453) (1400,1.14524453) (1450,1.097944305) (1500,1.063715777) (1550,1.063715777) (1600,1.057014159) (1650,1.028682674) (1700,1.006217004) (1750,1.006217004) (1800,1.006217004) (1850,1.006217004) (1900,1.006217004) (1950,1.006217004) (2000,1.006217004)};
\addplot[black!40,densely dotted,forget plot] coordinates {(100,1) (2000,1)};
\nextgroupplot[ylabel={\shortstack{Rider Throughput \\ Ratio $\frac{R(\mathbf{x}_R^*)}{R(\mathbf{x}^{Eq})}$}}]
\addplot[orange!85!black, line width=1.0pt, mark=square*, mark size=0.75pt,
    mark options={line width=1pt}] coordinates {(100,1.19516321) (150,1.376250487) (200,1.400792205) (250,1.301806409) (300,1.198098696) (350,1.161427095) (400,1.137089972) (450,1.182421951) (500,1.181786517) (550,1.188105035) (600,1.21555758) (650,1.222139535) (700,1.204400884) (750,1.177502612) (800,1.240168452) (850,1.203647851) (900,1.202668069) (950,1.245929151) (1000,1.201223949) (1050,1.21749426) (1100,1.252705344) (1150,1.283899964) (1200,1.262711494) (1250,1.199854502) (1300,1.14942661) (1350,1.14942661) (1400,1.14942661) (1450,1.093693323) (1500,1.053880549) (1550,1.053880549) (1600,1.048061315) (1650,1.023359078) (1700,1.003654238) (1750,1.003654238) (1800,1.003654238) (1850,1.003654238) (1900,1.003654238) (1950,1.003654238) (2000,1.003654238)};
\addplot[black!40,densely dotted,forget plot] coordinates {(100,1) (2000,1)};
\nextgroupplot[ylabel={\shortstack{Total Social Welfare \\ Ratio $\frac{W(\mathbf{x}_W^*)}{W(\mathbf{x}^{Eq})}$}}]
\addplot[teal!75!black,line width=1.0pt, mark=square*, mark size=0.75pt,
    mark options={line width=1pt}] coordinates {(100,1.369741294) (150,1.275534306) (200,1.177662102) (250,1.119722906) (300,1.111245439) (350,1.190687457) (400,1.150592594) (450,1.140517594) (500,1.160058921) (550,1.213432489) (600,1.250801396) (650,1.210273422) (700,1.247781096) (750,1.183337174) (800,1.204461613) (850,1.142530765) (900,1.191569299) (950,1.21688708) (1000,1.193119221) (1050,1.251723801) (1100,1.293025282) (1150,1.281022585) (1200,1.257905198) (1250,1.235665041) (1300,1.14524453) (1350,1.138624285) (1400,1.121655688) (1450,1.063715777) (1500,1.063715777) (1550,1.056248081) (1600,1.043362482) (1650,1.006217004) (1700,1.006217004) (1750,1.006217004) (1800,1.006217004) (1850,1.006217004) (1900,1.006217004) (1950,1.006217004) (2000,1.006217004)};
\addplot[black!40,densely dotted,forget plot] coordinates {(100,1) (2000,1)};
\end{groupplot}
\end{tikzpicture}